\documentclass[a4paper, reqno, 11pt]{amsart}        

\usepackage{amssymb}

\usepackage{epsfig}  		% For postscript
\usepackage{epic,eepic}       % For epic and eepic output from xfig

\let\oldtocsection=\tocsection
 
\let\oldtocsubsection=\tocsubsection
 
\let\oldtocsubsubsection=\tocsubsubsection
 
\renewcommand{\tocsection}[2]{\hspace{0em}\oldtocsection{#1}{#2}}
\renewcommand{\tocsubsection}[2]{\hspace{1em}\oldtocsubsection{#1}{#2}}
\renewcommand{\tocsubsubsection}[2]{\hspace{2em}\oldtocsubsubsection{#1}{#2}}

\usepackage{amsfonts}
\usepackage{amsthm}
\usepackage{amsmath}
\usepackage{amscd}
\usepackage[latin2]{inputenc}
\usepackage{t1enc}
\usepackage[mathscr]{eucal}
\usepackage{indentfirst}
\usepackage{graphicx}
\usepackage{graphics}
\usepackage{pict2e}
\usepackage{epic}
\numberwithin{equation}{section}
\usepackage[margin=2.9cm]{geometry}
\usepackage{hyperref}
\usepackage{dsfont}
\usepackage{csquotes}
\usepackage{stmaryrd}
\usepackage{blindtext}
\usepackage{tikz-cd}
\usepackage{mathrsfs}
\usepackage{cite}
\usepackage{mathtools}
\usepackage{extarrows}
\usepackage{epstopdf}
\usepackage[T1]{fontenc}
\usepackage{braket}

\theoremstyle{definition}
\newtheorem{definition}[equation]{Definition}
\newtheorem{example}[equation]{Example}
\newtheorem{proposition}[equation]{Proposition}
\newtheorem{theorem}[equation]{Theorem}
\newtheorem{remark}[equation]{Remark}

\newtheorem{corollary}[equation]{Corollary}

\numberwithin{equation}{section}

\newcommand{\midwedge}{\text{\Large$\wedge$}}
\newcommand{\midodot}{\text{\Large$\odot$}}

\newcommand{\be}{\begin{equation}}
\newcommand{\ee}{\end{equation}}
\def\beqa{\begin{eqnarray}}
\def\eeqa{\end{eqnarray}}
\def\bean{\begin{eqnarray*}}
\def\eean{\end{eqnarray*}}

\newcommand{\R}{\mathbb{R}}

\newcommand{\de}{\mathrm{d}}
\newcommand{\eqn}[1]{(\ref{#1})}
\newcommand{\del}{\partial}

\newcommand{\De}{\mathrm{D}}

\newcommand{\IZ}{\mathbb{Z}}
\newcommand{\IC}{\mathbb{C}}

\newcommand{\IR}{\mathbb{R}}

\newcommand{\IT}{\mathbb{T}}
\newcommand{\IS}{\mathbb{S}}
\newcommand{\IX}{\mathbb{X}}

\newcommand{\frg}{\mathfrak{g}}

\newcommand{\frd}{\mathfrak{d}}
\newcommand{\frm}{\mathfrak{m}}

\renewcommand{\Im}{\ensuremath{\mathfrak{Im}}}
\renewcommand{\Re}{\ensuremath{\mathfrak{Re}}}

\newcommand{\cW}{{\mathcal W}}
\newcommand{\cN}{{\mathcal N}}

\newcommand{\cS}{{\mathcal S}}

\newcommand{\cH}{{\mathcal H}}
\newcommand{\cA}{{\mathcal A}}
\newcommand{\cE}{{\mathcal E}}

\newcommand{\cQ}{{\mathcal Q}}

\newcommand{\cF}{{\mathcal F}}

\newcommand{\cK}{{\mathcal K}}

\newcommand{\cG}{{\mathcal G}}

\newcommand{\ccA}{{\mathscr A}}

\newcommand{\ccC}{{\mathscr C}}

\newcommand{\ccL}{{\mathscr L}}

\newcommand{\ccT}{{\mathscr T}}

\newcommand{\sfA}{{\mathsf{A}}}
\newcommand{\sfe}{{\mathsf{e}}}
\newcommand{\sfa}{{\mathsf{a}}}

\newcommand{\sfD}{{\mathsf{D}}}

\newcommand{\sfg}{{\mathsf{g}}}

\newcommand{\sfH}{{\mathsf{H}}}
\newcommand{\sfG}{{\mathsf{G}}}
\newcommand{\sfi}{{\mathsf{i}}}
\newcommand{\sfL}{{\mathsf{L}}}
\newcommand{\sfO}{{\mathsf{O}}}
\newcommand{\sfx}{{\mathsf{x}}}
\newcommand{\sfy}{{\mathsf{y}}}
\newcommand{\sfz}{{\mathsf{z}}}
\newcommand{\sfw}{{\mathsf{w}}}

\newcommand{\unit}{\mathds{1}}   			% identity map/matrix

\begin{document}

\title[Born Sigma-Models for Para-Hermitian Manifolds]{Born Sigma-Models for Para-Hermitian Manifolds \\[5pt] and Generalized T-Duality}

\author[V.~E.~ Marotta]{Vincenzo Emilio Marotta}
\address[Vincenzo Emilio Marotta]
{Department of Mathematics and Maxwell Institute for Mathematical
  Sciences\\ Heriot-Watt
  University\\ Edinburgh EH14 4AS\\ United Kingdom}
\email{vm34@hw.ac.uk}

\author[R.~J. Szabo]{Richard J.~Szabo}
  \address[Richard J.~Szabo]
  {Department of Mathematics, Maxwell Institute for Mathematical Sciences and The Higgs Centre for Theoretical Physics\\
  Heriot-Watt University\\
  Edinburgh EH14 4AS \\
  United Kingdom}
  \address{Dipartimento di Scienze e Innovazione
  Tecnologica\\ Universit\`a del Piemonte Orientale, Alessandria, Italy\\ and
Arnold--Regge Centre, Torino, Italy
  }
  \email{R.J.Szabo@hw.ac.uk}

\vfill

\begin{flushright}
\footnotesize
{\sf EMPG--19--21}
\normalsize
\end{flushright}

\vspace{1cm}

\begin{abstract}
We give a covariant realization of the doubled sigma-model formulation of
duality-symmetric string theory within the general framework of para-Hermitian geometry.
We define a notion of generalized metric on a para-Hermitian
manifold and discuss its relation to Born
geometry. We show that a Born
geometry uniquely defines a worldsheet
sigma-model with a
para-Hermitian target space, and we describe its Lie algebroid
gauging as a means of recovering the conventional sigma-model
description of a physical string
background as the leaf space of a foliated para-Hermitian
manifold. Applying the Kotov-Strobl gauging leads to a generalized notion of T-duality
when combined with transformations 
that act on Born geometries. We obtain a geometric interpretation of the
self-duality constraint that halves the degrees of freedom in doubled
sigma-models, and we give geometric characterizations of 
  non-geometric string backgrounds in this setting. We illustrate
  our formalism with detailed worldsheet descriptions
  of closed string phase spaces, of doubled groups where our notion of
  generalized T-duality includes non-abelian 
  T-duality, and of doubled nilmanifolds.
\end{abstract}

\maketitle

{\baselineskip=12pt
\tableofcontents
}

\bigskip

\section{Introduction}

Para-Hermitian geometry offers a simple yet effective mathematical
framework for the description of generalized flux compactifications of
string theory and the geometry underlying double field theory. Its modern
inspirations come from its complex analogue --- Hermitian geometry --- and
the differential geometry of exact Courant algebroids which are central to
generalized geometry and its applications to supergravity. Let us begin by recalling
some basic concepts surrounding Courant algebroids and their
counterparts in para-Hermitian geometry which will set the stage for
the investigation carried out in this paper.

\medskip

\subsection{Supergravity on Exact Courant Algebroids} ~\\[5pt]
Exact Courant algebroids \cite{Courant1990, Weinstein1997,gualtieri:tesi,
  Severa-letters} were originally introduced in \cite{Courant1990} to give a geometric interpretation of Dirac's theory of
constrained dynamical systems. The original example considered in
\cite{Courant1990} comprised the {standard Courant algebroid} or
{generalized tangent bundle} 
$$
\IT \cQ=T\cQ \oplus T^*\cQ
$$ 
over a smooth manifold $\cQ$ with the natural {Courant bracket}
$$
\llbracket X+ \xi, Y + \nu \rrbracket= [X,Y]+ \pounds_X \nu -\pounds_Y
\xi -\tfrac{1}{2}\,\de (\iota_X \nu -\iota_Y \xi)\ ,
$$ 
for all vector fields
$X, Y \in \mathsf{\Gamma}(T\cQ)$ and 1-forms $\xi, \nu \in \mathsf{\Gamma}(T^*\cQ)$;
here $[X,Y]$ denotes the usual Lie bracket of vector fields, $\pounds_X$ denotes the Lie derivative in the direction of $X$, and $\iota_X$ is the contraction with $X$.
The Courant bracket is skew-symmetric but violates the Jacobi identity
by the exterior derivative of the associated Nijenhuis tensor. 

The general definition of a
Courant algebroid is modeled on the properties of $\IT \cQ$~\cite{Weinstein1997}: A Courant
algebroid is a vector bundle $E\rightarrow \cQ$ of even rank
endowed with a fiberwise split signature metric $\eta$ together with a bracket $\llbracket\,\cdot\, , \, \cdot\, \rrbracket_E$ called a
{Dorfman bracket}, whose skew-symmetrization is called a Courant
bracket, which satisfies the Jacobi identity and is compatible with
the split signature metric. It is also equiped with an {anchor map} $\rho:
E \rightarrow T\cQ$ which is a bracket-preserving homomorphism between the Dorfman
bracket on $E$ and the Lie bracket on $T\cQ,$ such that the Dorfman bracket satisfies an anchored Leibniz rule.

Exact Courant algebroids are those Courant algebroids $E\to\cQ$ which also fit into a short exact sequence of vector bundles 
\be \label{shortcou}
0 \longrightarrow T^*\cQ \xlongrightarrow{} E \xlongrightarrow{} T\cQ \longrightarrow 0 \ .
\ee
Isotropic splittings of this short exact sequence are isomorphic to the standard Courant algebroid $\IT \cQ= T\cQ \oplus T^*\cQ$ with Dorfman bracket
$$
\llbracket X+ \xi, Y + \nu \rrbracket_E^H= [X,Y]+ \pounds_X \nu  -\iota_Y \,
\de \xi + H(X,Y) 
$$ 
which is `twisted' by a closed 3-form $H\in\Omega^3(\cQ)$ whose de~Rham
class $[H]\in \mathsf{H}^3(\cQ, \IR)$ is called the {{\v
    S}evera class} of $E$ \cite{Severa-letters}. Each distinct isotropic
splitting of \eqref{shortcou} is associated with a different 3-form
$H$ up to $B$-transformations by closed 2-forms. 
Generic $B$-transformations~\cite{gualtieri:tesi, Severa2002} of exact Courant
algebroids are generated by arbitrary 2-forms $b\in \Omega^2(\cQ)$ and
preserve the fiberwise metric. They also preserve the Dorfman bracket
if $b$ is a closed 2-form. When $b$ is not closed the corresponding
Courant algebroid with Dorfman bracket twisted by $H+\de b$ is
associated with a different splitting of the short exact sequence
\eqref{shortcou}, but with the same {\v S}evera class. In the language of string theory, an isotropic
splitting is a choice of Kalb-Ramond field on $\cQ$ for which
the string background carries NS--NS $H$-flux.

Another important structure that can be defined on any Courant
algebroid is a {generalized metric}. A generalized metric can be
regarded as a choice of a sub-bundle of $E$ which is positive-definite
in the fiberwise split signature metric on $E.$ This is equivalent to
defining a fiberwise Riemannian metric on
$E$~\cite{gualtieri:tesi}. For the particular case of an exact Courant
algebroid, any generalized metric corresponds to a pair $(g, b)$ of a
Riemannian metric $g$ on $\cQ$ and a 2-form $b \in \Omega^2(\cQ),$ which
are dynamical fields in the bosonic sector of the low-energy effective supergravity
theory on $\cQ$ underlying the string theory. Type~II supergravity was
formulated in terms of the generalized geometry of exact Courant
algebroids in~\cite{Grana2008,Coimbra:2011nw}.

Two further key facts about exact Courant
algebroids will play a role in this paper: Firstly, an isotropic splitting of
\eqref{shortcou} and a generalized metric on $E$ uniquely define a
two-dimensional sigma-model with target space $\cQ$~\cite{Severa2015}. Secondly, topological
T-duality can be implemented as an isomorphism between exact Courant
algebroids~\cite{Cavalcanti2010,Crilly:2019vmu}; however, factorized T-dualities are
not manifest symmetries of supergravity in this context. 

\medskip

\subsection{Double Field Theory on Para-Hermitian Manifolds} ~\\[5pt]
Double field
theory~\cite{Siegel:1993xq,Siegel:1993th,HullZw2009,Hull:2009zb, hohmhz,hullzw,Hohm:2010xe}
is a duality-covariant formulation of supergravity in which T-duality symmetry
is manifest, and it provides a geometric setting for the description of non-geometric backgrounds of string theory~\cite{Hellerman:2002ax,Dabholkar2002,Hull2005} (see~\cite{Wecht:2007wu,Berman2013,Plauschinn:2018wbo} for reviews
and further references). It was originally derived for (flat) toroidal
compactifications, and it can be extended to curved backgrounds which
are local doubled torus
fibrations; in these instances T-duality acts geometrically as a subgroup
of the group of large diffeomorphisms of the doubled space. The background independent formulation of double field
theory~\cite{hohmhz} suggests defining it on more general doubled
manifolds $M$. In such formulations one is faced with the conceptual
problem of the meaning of T-duality and the doubling of local coordinates, since a general spacetime manifold $\cQ$
need not admit an equal number of momentum and winding modes as the
latter are associated to the non-contractible 1-cycles of $\cQ$. 

A fully dynamical doubled geometry beyond the strong constraint
is rather poorly understood; it has been realized that the
correct global picture of a doubled spacetime $M$ in this framework is that of a foliated manifold
for which a maximally isotropic polarization (a solution to the strong constraint) selects
a conventional physical spacetime as a quotient $\cQ=M/\cF$ by the
equivalence relation induced by the
leaves of the
foliation $\cF$~\cite{Hull2009,Schulz:2011ye,Vaisman2012,Park:2013mpa,Lee:2015xga}, rather
than as a submanifold as in the case of flat spaces $M$, and the physical fields as foliated tensors. Global aspects of
doubled geometry were considered
by~\cite{Park:2013mpa,Hohm2013,Berman:2014jba,
  Papadopoulos:2014mxa,Hull:2014mxa,Howe:2016ggg} in a bottom-up approach whereby
flat open subsets are patched together using the physical symmetries
of double field theory as transition functions, which thereby become manifest geometric symmetries of the dynamical theory. One issue surrounding the global formulation of double field theory is whether the underlying split signature metric $\eta$ should be flat, which appears to severely restrict the possible doubled manifolds~\cite{Lee:2015xga}; non-constant metrics $\eta$ were considered in~\cite{Cederwall:2014kxa} where the most general consistent metrics were suggested to take on a pp-wave type form.

Para-Hermitian geometry first appeared in~\cite{Vaisman2012,
  Vaisman2013} as a top-down approach to the geometric description of
double field theory,\footnote{The origins of this approach can be traced back
  to the mathematical structures suggested by Hull~\cite{Hull2005} in
  the context of doubled torus fibrations, where the terminology
  `pseudo-Hermitian' was used instead of `para-Hermitian'.} 
and was further developed along these lines in~\cite{Freidel2017,
  Svoboda2018, Freidel2019,SzMar,Mori2019}. In this framework one is
faced with the problem of understanding the reductions of the theory to the usual flat
space doubled geometry and to generalized geometry. The perspective on para-Hermitian geometry which we adopt in this paper is that it allows for the
introduction of structures similar to exact Courant
algebroids directly on the tangent bundle $TM$ of a smooth manifold $M$ of
even dimension. We introduce a split signature metric $\eta$ and an
automorphism $K\in {\sf Aut}_\unit(TM)$ such that $K^2 = \unit,$
which defines a splitting 
$$
TM=L_+\oplus L_-
$$ 
where $L_\pm$ are the
$\pm\,1$-eigenbundles of $K$ which are maximally isotropic with respect to $\eta.$ 
These structures are {compatible} in the sense that they satisfy the condition
$$
\eta\big(K(X),K(Y)\big)=-\eta(X,Y) \ , 
$$ 
for all $ X, Y \in \mathsf{\Gamma}(TM).$ We call the triple $(M, K, \eta)$ an
{almost para-Hermitian manifold}; it encodes the kinematical content
of double field theory on $M$. This compatibility condition defines a
2-form $\omega$ on $M$ called the {fundamental 2-form} of the almost
para-Hermitian structure, which is almost symplectic by
construction. When $\omega$ is closed we call $(M,K,\eta)$ an almost
para-K\"ahler manifold, and this is the situation that most closely
resembles the original flat space formulation of double field theory.

An almost para-Hermitian manifold can be endowed with a
metric-compatible bracket satisfying the Leibniz rule and for which
$L_\pm$ are involutive. This gives the tangent bundle $TM$ the structure of a metric
algebroid~\cite{Vaisman2012,Freidel2017,Jonke2018,
  Svoboda2018, Freidel2019,SzMar,Mori2019,Chatzistavrakidis:2019huz}. The bracket is called a {D-bracket}, and it is neither
skew-symmetric nor satisfies the Jacobi identity. A different D-bracket is associated to
each para-Hermitian structure $(K,\eta)$ on the same manifold $M$.
When one
of the eigenbundles $L_\pm$ is Frobenius integrable and so gives a
foliation of $M,$ one can construct a standard Courant algebroid on
each leaf of the foliation. It is shown in~\cite{Freidel2017,
  Svoboda2018, Freidel2019} that the metric $\eta$ induces a
bracket-preserving isomorphism between the D-bracket on $TM$ and the
Dorfman bracket associated with the standard Courant algebroid on the
foliation. When both eigenbundles $L_\pm$ are integrable, then $L_\pm
= T\cF_\pm$ and locally the para-Hermitian manifold is of the form
$\cS_+\times\cS_-$, where $\cS_\pm\subset\cF_\pm$ are leaves of the
foliations $\cF_\pm$ with local coordinates $\IX^I=(x^i,\tilde x_i)$
adapted to $\cS_\pm$ that can be thought of as spacetime and winding coordinates, respectively; the D-bracket then recovers the D-bracket of double field theory.

As discussed in \cite{Freidel2019}, one can also consider almost
para-Hermitian manifolds that admit a Riemannian metric $\cH$ which is
compatible with the para-Hermitian structure in the sense that it satisfies the conditions
$$ 
\mathcal{H}^{-1} = \eta^{-1}\,\circ \, \mathcal{H} \,\circ\, \eta^{-1} = 
- \omega^{-1}\,\circ \, \mathcal{H} \,\circ \, \omega^{-1} \ .
$$
The quadruple $(M, K, \eta, \cH)$ is called a {Born manifold}; it
encodes the dynamical field content of double field theory on $M$. The metric $\cH$ is a special case of the generalized metrics introduced in~\cite{Vaisman2012} which are constrained only by the first equality. In the integrable case, a Born metric is equivalent to the introduction of a spacetime metric.

Para-Hermitian manifolds encode the mathematical structure of a
`doubled geometry', serving as the extended spacetime of double field theory. To provide a physical meaning to
this structure one needs to
clarify how to recover a conventional closed string background from a para-Hermitian
manifold. In this paper we will define a physical spacetime
from a worldsheet perspective by
introducing a non-linear sigma-model for a foliated para-Hermitian
manifold whose coupling to background fields on the target space is
uniquely determined
by a Born geometry. When permitted, the gauging of this sigma-model
using the techniques of~\cite{Kotov2014, Kotov2018} gives a worldsheet
description of the quotient space represented by the leaf space of the
foliation. We interpret this leaf space as the physical background,
and the gauging as a worldsheet description of a non-linear version of the strong constraint
of double field theory.

\medskip

\subsection{Doubled Sigma-Models} ~\\[5pt]
This paper is mostly concerned with the formulation and analysis of two-dimensional
non-linear sigma-models for para-Hermitian
manifolds, and how they provide a worldsheet formulation of string
theory in such doubled spaces. They form the basis for the
worldsheet approach to double field theory on para-Hermitian target
spaces, which provides a very general geometric realization of the
duality-symmetric formulations of string theory via doubled
sigma-models, see e.g.~\cite{Duff:1989tf,Tseytlin:1990nb,Siegel:1993xq,Siegel:1993th,Hull2005}.
We aim to understand target space dualities in these sigma-models as 
a consequence of vector bundle automorphisms which preserve the split signature metric
$\eta$ and map para-Hermitian structures into para-Hermitian
structures. Our approach is largely inspired by Hull's doubled 
formalism for local torus fibrations~\cite{Hull2005,Hull2007}, in which dual coordinates conjugate to 
winding modes of the closed string are introduced
alongside the torus fiber coordinates conjugate to
momentum modes. As particular instances of our construction, we will
obtain new perspectives on the doubled sigma-models for group
manifolds, twisted tori and nilmanifolds which were originally
developed in~\cite{Hull:2007jy,Hull2009,ReidEdwards2009}.

Quantum aspects of the sigma-model formulation for doubled torus fibrations
were developed by~\cite{Berman:2007xn} where the vanishing of the
1-loop beta-functions were found to give effective field equations reminescent of the equations of motion in double field
theory. This was extended by~\cite{Copland:2011wx} to a doubled sigma-model whose
effective spacetime field theory at 1-loop is double field
theory. Gaugings of doubled sigma-models were considered
by~\cite{Hull2007,Hull2009,Lee:2013hma,Bakas:2016nxt} which implement
the strong constraint of double field theory in the form of a
chirality constraint on the worldsheet fields: In these worldsheet formulations, the
choice of polarization is achieved by the gauging and the subsequent quotient yields a conventional
description of a physical string background.

Topologically twisted versions of doubled sigma-models are considered
in~\cite{Kokenyesi:2018xgj}, where they are related to extensions
of generalized complex geometry and used to describe covariant geometric
theories for other string dualities such as S-duality. Generalized
para-K\"ahler structures also appear in~\cite{Hu2019} as target space
geometries for doubled sigma-models with $\cN=(2,2)$ twisted
supersymmetry, similarly to the appearence of affine and projective
versions of the special para-K\"ahler geometry of rigid and local
$\cN=2$ vector multiplets in Euclidean
spacetimes~\cite{Cortes2004}. Born structures further appear in target
spaces for sigma-models with $\cN=(1,1)$ supersymmetry
in~\cite{Stojevic:2009ub,Hu2019}, where their non-integrability leads
to non-geometric string backgrounds in a similar fashion to what we
describe in the present paper, and as the target geometries of $\cN=2$
hypermultiplets in Euclidean
signature~\cite{Cortes:2005uq}. Supersymmetric extensions of our
non-linear sigma-models will not be discussed in this paper.

\medskip

\subsection{Overview of Results and Outline} ~\\[5pt]
In this paper we will follow and expand on the analogies between exact Courant algebroids
and almost para-Hermitian manifolds. For this, we define in
Section~\ref{intropara} a notion of
generalized metric on an almost para-Hermitian manifold $M$ and discuss its
properties. We show that generalized metrics are in one-to-one
correspondence with pairs of a fiberwise Riemannian metric $g_+$ on $L_+$
and a 2-form $b_+\in \mathsf{\Gamma}(\midwedge^2 L_+^*)$. We also show that compatible Riemannian metrics $\cH$ which define Born
geometries are a special class of generalized metrics; we refer to them
as {generalized metrics} which are compatible with the para-Hermitian
structure. 

In Section~\ref{ContTduality} we then turn to the characterization of transformations mapping Born
geometries into Born geometries, which are given by vector bundle
automorphisms of $TM$ preserving the split signature metric $\eta.$ We
denote the group of such transformations by ${\sf O}(d,d)(M),$ where
$\dim(M)=2d$. They are the crux of our interpretation of {generalized
  T-duality} in the framework of para-Hermitian geometry. As an
explicit example, we recover the $B$-transformations presented
in~\cite{Svoboda2018}. The natural group of discrete transformations
is
\begin{align} \label{eq:discreteTduality}
\sfO(d,d)(M)\, \cap\, {\sf Diff}(M;\IZ) \ , 
\end{align}
where ${\sf Diff}(M;\IZ)\subset{\sf Diff}(M)$ is the subgroup of large
diffeomorphisms of $M$; for example, if $M=\IT^{2d}$ is a torus, then
${\sf Diff}(\IT^{2d};\IZ)={\sf GL}(2d,\IZ)$,
$\sfO(d,d)(\IT^{2d})=\sfO(d,d)\subset{\sf GL}(2d,\IR)$ and
\eqref{eq:discreteTduality} is the T-duality group $\sfO(d,d;\IZ)$,
which is a symmetry of toroidally compactified string theory. The
issue in general is then which subgroup of \eqref{eq:discreteTduality}
yields proper T-duality symmetries of string theory that can be used
as transition functions in constructing candidate physical string
backgrounds from the Born geometry of $M$. However, we are also interested in the kind of
generalization of T-duality proposed by~\cite{Dabholkar2005}, which is
naturally encompassed by our general formalism, so we will not address the issue of which transformations define physically equivalent string backgrounds in the quantum theory any further in this paper. Physically inequivalent backgrounds from our perspective also offer the possibility of describing different quantum completions of the same classical theory, such as those which are related by non-abelian or Poisson-Lie T-duality.

To complete the analogy with exact Courant algebroids, in
Section~\ref{BSM} we show that
two-dimensional sigma-models for $M$ are naturally associated with a choice of a
Born structure on $M$ by using the compatible generalized metric $\cH$
and the fundamental 2-form $\omega.$ We call them {Born sigma-models},
and we propose them as a covariant realization of the
duality-symmetric formulation of string theory via doubled
sigma-models. Born sigma-models with target space $M$ are in
one-to-one correspondence with Born structures on $M$. We use them to
give a more precise meaning to our notion of {generalized T-duality}
by starting with the simple observation that T-dual Born sigma-models
are obtained via $\sfO (d,d)(M)$-transformations of Born structures. 

To explore the connections between this construction and the usual
worldsheet descriptions of conventional string backgrounds, we develop the gauging
 of Born sigma-models. We focus on
the cases of para-Hermitian manifolds $M$ which admit a maximally isotropic
foliation. A foliation is not generally induced by the action of a Lie group, nor will a generic Born manifold admit a Lie algebra of
Killing vectors for $\cH$ that is normally required to
gauge an isometry in the traditional approach to gauging
two-dimensional sigma-models. Following the approach of Kotov and Strobl~\cite{Kotov2014,
  Kotov2018}, we apply the Lie algebroid gauging
of sigma-models for foliated manifolds as a means of obtaining a
worldsheet description of the leaf space of the foliation, and we further elaborate on the characterization of gauged sigma-models on manifolds admitting a regular foliation. We show
that the generalized isometry condition allowing for the
gauging is equivalent to the existence of a bundle-like metric on $M$;
in this case the Lie algebroid connection required for the gauging
can be naturally chosen to be the
Bott connection defined by the foliation. The Lie algebroid gauging of the Born sigma-model also leads to a geometric interpretation of the usual self-duality constraint
imposed on doubled
sigma-models~\cite{Hull2005,Hull2007,Hull2009}, which eliminates half
of the $2d$ closed string degrees of freedom by restricting $d$ of
them to be right-moving and $d$ of them to be left-moving on the worldsheet.
The leaf space of a foliated para-Hermitian manifold $M$ defines the
physical spacetime which is recovered from the doubled geometry,
i.e.~from the para-Hermitian structure. The spacetime is endowed
with a metric and $B$-field descending from the Born metric on $M$ via a Riemannian submersion,
and in particular from its fiberwise component along the sub-bundle
$L_+$ of the tangent bundle $TM$ which defines a Riemannian foliation
of $M$.

The role of ${\sf O}(d,d)(M)$-transformations is crucial in our
interpretation of generalized T-duality, since they map Born
sigma-models into one another. This becomes particularly relevant when
discussing target space dualities between reduced sigma-models. Whenever two dually
related Born sigma-models have target spaces whose associated almost
para-Hermitian structures admit at least one integrable eigenbundle
for which the generalized isometry condition is satisfied, we obtain a pair of
conventional non-linear sigma-models for different leaf spaces; we say that the reduced
sigma-models are {T-dual} to one another. In a similar vein as
discussed above, in this paper we do not address conformal or modular
invariance of our Born sigma-models, nor which subgroup of
\eqref{eq:discreteTduality} would be an automorphism of the worldsheet
conformal field theory. The proper implementation of conformal
invariance, by possibly adding other sectors as necessary,
would lead to field equations for the Born geometry
$(\eta,\omega,\cH)$ which 
provide a global generalization of the equations of motion of double
field theory.

Our worldsheet theory also gives novel geometric characterizations of
non-geometric string
backgrounds. When the leaf space of a
reduced sigma-model is smooth, the sigma-model describes a
{geometric background}; this typically happens when the underlying almost para-Hermitian manifold is the total space of a fiber bundle. On the other hand, if the leaf space does not admit the structure
of a smooth manifold, the background is only locally geometric, and
following standard terminology~\cite{Hull2005} we call it a {T-fold}; in
particular, those T-folds that arise from foliations with compact
leaves and finite leaf holonomy group have the structure of an
orbifold. Finally, it may happen that the gauging condition for a Born
sigma-model holds with respect to a non-involutive eigenbundle of $K,$
so that the reduction through gauging cannot be performed since there
is no foliation and thus no conventional spacetime even locally; following the terminology
of~\cite{Hull:2019iuy}, we say that the Born sigma-model describes an essentially
doubled background.

In the final three sections we turn our attention
to special classes of examples that illustrate our general
formalism. We characterize the para-K\"ahler structures of cotangent
bundles, and describe the gauging and related generalized T-duality of
the associated Born sigma-models in Section~\ref{sigmacot}. This
provides a general extension and covariant realization of Tseytlin's doubled sigma-model
approach via the closed string phase space~\cite{Tseytlin:1990nb} (see
also~\cite{Freidel2014}), and moreover exhibits the main qualitative
features of the gauging of Born sigma-models in a simplified yet
explicit setting.

In Section~\ref{DoubledGroups} we show how to define left-invariant
almost para-Hermitian structures on Lie groups, particularly those
associated with Manin pairs and Manin triples, for which our notion of
generalized T-duality includes the non-abelian T-duality of~\cite{delaOssa:1992vci}
and some aspects of the Poisson-Lie T-duality of~\cite{Klimcik:1995ux}, as well as
their generalizations to generic doubled groups 
proposed in~\cite{Hull2009,ReidEdwards2009,edwards:nonpub,Osten:2019ayq}. We demonstrate that
the generalized isometry conditions imply that the leaf space
is a reductive homogeneous space with a bi-invariant metric and the surjective submersion from
the Born manifold is a principal bundle, and we consider various
examples of the corresponding Born sigma-models; this includes the examples of symmetric spaces as special cases, which have previously
appeared as natural and explicit solutions to the strong constraint of
double field theory on doubled groups in~\cite{Hassler:2016srl,Demulder:2018lmj}. 

Using the Born
structure of the Drinfel'd double $T^*\sfH$ of the three-dimensional
Heisenberg group $\sfH,$ in Section~\ref{sec:BornDTT} we obtain a Born
structure on the corresponding doubled
twisted torus, i.e. the compact manifold given by the quotient of
$T^*\sfH$ with respect to the left action of a discrete cocompact
subgroup. Starting from this Born manifold, we discuss how to obtain
T-dual sigma-models for the different leaf spaces which reproduces the
standard T-duality chain of geometric and non-geometric backgrounds starting from the torus $\IT^3$
with NS--NS $H$-flux~\cite{Kachru:2002sk,Shelton2005}; this gives a
somewhat more precise geometric approach to the worldsheet theory for the
doubled twisted torus formalism developed in~\cite{Hull2009} using
standard isometric gauging techniques.

\medskip

\subsection{Acknowledgments} ~\\[5pt]
We thank Chris Blair, Athanasios~Chatzistavrakidis, Lorenzo~Foscolo, Chris~Hull,
Alexei~Kotov, Thomas Mohaupt, Franco~Pezzella, Erik~Plauschinn,
Felix~Rudolph, Alexander~Schen{-}kel, Thomas~Strobl and Patrizia~Vitale for helpful discussions.
This work was supported in part by the Action MP1405 ``Quantum Structure of Spacetime'' funded by the European
Cooperation in Science and Technology (COST). The work of V.E.M. was
funded by the Doctoral Training Grant ST/R504774/1 from the UK Science and Technology
Facilities Council (STFC). The work of R.J.S. was supported by the STFC
Consolidated Grant ST/P000363/1 and the
Istituto Nazionale di Fisica Nucleare (INFN, Torino).

\section{Generalized Metrics in Para-Hermitian Geometry} \label{intropara}

In this section we shall introduce a notion of generalized
metric in para-Hermitian geometry which extends the more familiar
structure from generalized geometry~\cite{gualtieri:tesi,Hitchin2011},
and compare it to Born geometry.

\medskip

\subsection{Para-Hermitian Manifolds} ~\\[5pt]
We begin by introducing the basic concepts and constructions of
para-Hermitian geometry that we
shall need, following~\cite{Cortes2004,Boyer2007,SzMar} for the most
part. Throughout this paper all manifolds, fibrations and sections
of vector bundles are assumed to be smooth, while all vector bundles,
vector spaces, Lie groups and Lie algebras are taken to be real, unless otherwise
explicitly stated.

\theoremstyle{definition}
\begin{definition}
An \emph{almost product structure} on a manifold $M$ is an
automorphism $K \in {\sf Aut}_\unit(TM)$ covering the
identity\footnote{A more common nomenclature for ${\sf Aut}_\unit(TM)$
  is the `gauge subgroup' ${\sf Gau}(TM)$ of the automorphism group
  ${\sf Aut}(TM)$ of the tangent bundle $TM.$} such that $K^2=\mathds{1}$ and $K\neq\pm\,\mathds{1}$. The pair $(M,K)$ is an \emph{almost product manifold}.
\end{definition}
The automorphism $K$ induces a $(1,1)$-tensor field on $M,$ 
denoted $\underline{K} \in \mathsf{\Gamma}(TM \otimes T^* M).$
We immediately notice the analogy with almost \emph{complex}
manifolds, which are even-dimensional manifolds endowed with a $(1,1)$-tensor field $J$ such that $J^2=-\mathds{1}.$ 
This analogy is a useful guide to understanding the structures and the terminology introduced in the following, for further details see \cite{Cortes2004}.

\theoremstyle{definition}
\begin{definition} \label{paraco}
An \emph{almost para-complex manifold} is an almost product manifold
$(M,K)$ with $M$ of even dimension such that the two eigenbundles
$L_+$ and $L_-$ associated, respectively, with the eigenvalues $+1$
and $-1$ of $K$ have the same rank. A splitting of the tangent bundle
\be\label{eq:splittingTM}
TM=L_+\oplus L_-
\ee
of a manifold $M$ into the Whitney sum of two
sub-bundles $L_+$ and $L_-$ of the same fiber dimension is an
\emph{almost para-complex structure} on $M$. The splitting
\eqref{eq:splittingTM} is
a \emph{polarization} of the almost para-complex manifold $M$.
\end{definition}
Using the almost product structure, we can define two projection operators
\begin{align*}
\Pi_\pm=\tfrac{1}{2}\,(\mathds{1}\pm K): \mathsf{\Gamma}(TM) \longrightarrow \mathsf{\Gamma}(L_\pm) \ .
\end{align*}

\begin{remark}\label{rem:reduction}
A \emph{$\sf G$-structure} on a $2d$-dimensional manifold
$M$, for a subgroup ${\sf G} \subset {\sf GL}(2d,\mathbb{R})$, is a
$\sf G$-sub-bundle of the frame bundle $FM$, i.e. a reduction on the
frame bundle of the structure group ${\sf GL}(2d,\mathbb{R})$ to $\sf
G$. The definition of almost para-complex structure can therefore be recast by saying that it is a $\sf G$-structure on $M$ with structure group ${\sf G}={\sf GL}(d,\mathbb{R})\times {\sf GL}(d,\mathbb{R})$.  
These reductions are in one-to-one correspondence with sections of the
bundle associated with $FM$ whose typical fibers are the coset ${\sf GL}(2d,
\mathbb{R})/ {\sf GL}(d, \mathbb{R})\times {\sf GL}(d, \mathbb{R}).$
This also gives a one-to-one correspondence between ${\sf GL}(d,
\mathbb{R})\times {\sf GL}(d, \mathbb{R})$-reductions and
$(1,1)$-tensor fields induced by tangent bundle automorphisms $K \in {\sf Aut}_\unit(TM)$ as in Definition~\ref{paraco}. 
Furthermore, a ${\sf GL}(d, \mathbb{R})\times {\sf GL}(d,
\mathbb{R})$-reduction of $FM$ implies that $TM,$ the vector bundle associated with $FM,$ has structure group ${\sf GL}(d, \mathbb{R})\times {\sf GL}(d, \mathbb{R}).$ 
\end{remark}

We shall now study the integrability of the sub-bundles $L_+$ and $L_-.$ We start by characterizing the integrability of an almost para-complex structure.

\begin{definition}
An almost product structure $K$ is (\emph{Frobenius}) \emph{integrable} if
its eigenbundles $L_+$ and $L_-$ are both integrable:
$[\mathsf{\Gamma}(L_\pm),\mathsf{\Gamma}(L_\pm)]\subseteq\mathsf{\Gamma}(L_\pm)$. In this case
$K$ is a \emph{product structure}. A \emph{para-complex structure} is an integrable almost para-complex structure, i.e. a product structure with ${\rm rank}(L_+)={\rm rank}(L_-)$. 
\end{definition}
By Frobenius' Theorem, in this instance the manifold $M$ admits two
foliations $\mathcal{F}_+$ and $\mathcal{F}_-$, such that
$L_+=T\mathcal{F}_+$ and $L_-=T\mathcal{F}_-$. From the
  definition of para-complex structure, the distributions $L_\pm$ have
  constant rank, and hence the foliations $\cF_\pm$ are regular. 

Another way to characterize the integrability of an almost product structure is through the Nijenhuis tensor field, continuing the analogy with almost complex structures.

\begin{definition}
The \emph{Nijenhuis tensor field} of an almost product structure $K$ is the map $N_K: \mathsf{\Gamma}(TM)\times \mathsf{\Gamma}(TM) \rightarrow \mathsf{\Gamma}(TM)$ given by
\begin{equation*}
N_K(X,Y)=[X,Y]+[K(X),K(Y)] -K\big([K(X),Y] + [X,K(Y)]\big) \ ,
\end{equation*}
for all $X, Y \in \mathsf{\Gamma}(TM).$
\end{definition}
Then we can state the para-complex counterpart of the
Newlander-Nirenberg theorem as 
\begin{theorem}
An almost product structure $K$ on a manifold $M$ is integrable if and only if $N_K(X,Y)=0$ for all $X,Y \in \mathsf{\Gamma}(TM).$ 
\end{theorem}

Using the projection tensors $\Pi_\pm$, together with $K=\Pi_+-\Pi_-$, we can decompose the Nijenhuis tensor as 
\begin{equation}\label{eq:NKdecomp}
N_K(X,Y)=N_{\Pi_+}(X,Y)+N_{\Pi_-}(X,Y) \ ,
\end{equation}
where 
\begin{equation}
N_{\Pi_\pm}(X,Y)=\Pi_\mp\big([\Pi_\pm(X),\Pi_\pm(Y)]\big) \ .
\label{parnij}
\end{equation}
From \eqref{parnij} it follows that $N_{\Pi_\pm}(X,Y)\in \mathsf{\Gamma}(L_\mp)$. Hence the two components of the Nijenhuis tensor obstruct the closure of the Lie bracket of vector fields restricted to $L_+$ and $L_-$, respectively. 
In particular, $N_{\Pi_+}$ and $N_{\Pi_-}$ are independent of each
other. Thus one of them may vanish while the other one may not. In
this case, $M$ admits only one foliation and the para-complex
structure is only partially integrable in the sense that $N_K(X,Y)$ is
still non-vanishing, but it is controlled by only one of its
components introduced in the decomposition~\eqref{eq:NKdecomp}.

Following the analogy with complex geometry, we will now
introduce a compatible metric on almost para-complex manifolds, as in
the case of almost Hermitian manifolds. For this, let us return to the
description from Remark~\ref{rem:reduction} of an almost para-complex
structure on a $2d$-dimensional manifold $M$ in terms of a reduction
of the structure group of the frame bundle $FM.$ As explained in~\cite{Boyer2007},
when the frame bundle $FM$ admits a reduction of
the structure group to ${\sf GL}(d,\mathbb{R}) \times {\sf
  GL}(d,\mathbb{R}),$ it also admits a reduction to ${\sf O}(d,d),$
since these two subgroups are homotopy equivalent. In fact, in both cases the maximal compact subgroup is ${\sf O}(d) \times {\sf O}(d),$ which is also allowed as a reduction of the structure group when an almost para-complex structure can be defined.

\begin{example} \label{fiberexample}
Let 
$$
\pi: M \longrightarrow \cQ
$$ 
be a fibered manifold of even dimension $2d$ with ${\rm dim}(\cQ)= d.$ The
surjective submersion $\pi$ induces a short exact sequence of vector bundles on $M$ given by
\be \label{bunshort}
0 \longrightarrow L_{\tt v}(M)\xlongrightarrow{i} TM \xlongrightarrow{\hat{\pi}} \pi^*(T\cQ)\longrightarrow 0 
\ee
where $L_{\tt v}(M)= {\rm Ker}(\pi_*)$ is the vertical sub-bundle of
$TM,$ $i:L_{\tt v}(M) \xhookrightarrow{} TM$ is the inclusion map  and
$\pi^*(T\cQ)$ is the pullback of the tangent bundle of $\cQ$ to $M$ by the projection $\pi.$ The surjective map $\hat{\pi}: TM \rightarrow \pi^*(T\cQ)$ is induced by the differential of the projection $\pi_*: TM \rightarrow T\cQ.$ 
A splitting of the short exact sequence \eqref{bunshort} is given by the choice of a right inverse to $\hat{\pi},$ called an \emph{Ehresmann connection} 
$$
s: \pi^* (T\cQ) \longrightarrow TM \qquad \mbox{with} \quad \hat{\pi}
\circ s = \unit_{\pi^*(T\cQ)} \ .
$$
Then there is a decomposition
$$
TM=  {\rm Im}(s) \oplus L_{\tt v}(M) 
$$ 
which is associated with an almost para-complex structure, since
Whitney sums of vector bundles are in one-to-one correspondence with
almost product structures. In other words, there is an automorphism
$K_s \in {\sf Aut}_\unit (TM)$ such that $K_s^2=\unit_{TM}$ which is given by 
$$
K_s\big(s(X)+X_{\tt v}\big)= s(X)-X_{\tt v} \ , 
$$  
with $X \in \mathsf{\Gamma}(T\cQ)$ and $X_{\tt v} \in \mathsf{\Gamma}(L_{\tt v}(M))$; thus
${\rm Im}(s) $ is the $+1$-eigenbundle of $K_s$ and $L_{\tt v}(M)$ is the
$-1$-eigenbundle. The distribution $L_{\tt v}(M)$ is always involutive
and its integral manifolds are the fibers of $M,$ while ${\rm Im}(s)$
is generally non-integrable. Hence the choice of a splitting $s$ is
equivalent to a ${\sf GL}(d,\mathbb{R}) \times {\sf
  GL}(d,\mathbb{R})$-reduction of the structure group of the frame
bundle of $M.$ Since a splitting of the exact sequence
\eqref{bunshort} can always be found, an almost para-complex structure
on $M$ can always be defined. This also implies that a 
metric $\eta$ of signature $(d,d)$, or equivalently an ${\sf O}(d,d)$-reduction of the structure group, always exists on $TM.$ However, $K_s$ and $\eta$ do not necessarily satisfy any kind of compatibility condition.
\end{example}

Motivated by Example~\ref{fiberexample} and the usual construction of
almost Hermitian manifolds, we now introduce structures in which $K$ and $\eta$ satisfy a compatibility condition.

\begin{definition} \label{paraher}
An \emph{almost para-Hermitian manifold} $(M,K,\eta)$ is an almost
para-complex manifold $(M,K)$ together with a metric $\eta$ of  signature $(d,d)$ which is compatible with the automorphism $K$ in the sense that 
$$
\eta\big(K(X),K(Y)\big)=-\eta(X,Y) \ , 
$$ 
or equivalently 
\begin{equation}
\eta\big(K(X),Y\big)+\eta\big(X,K(Y)\big)=0 \ , 
\label{compcon}
\end{equation}
for all $X, Y \in \mathsf{\Gamma}(TM).$
\end{definition}
The condition \eqref{compcon} implies that the distributions $L_+$ and $L_-$ are maximally isotropic with respect to $\eta,$ so that they define para-Hermitian versions of Dirac structures. From \eqref{compcon} we also deduce the existence of a non-degenerate 2-form field $\omega$ on $M$ given by
$$
\omega(X,Y)=\eta\big(K(X),Y\big) \ , 
$$ 
for all $X,Y \in \mathsf{\Gamma}(TM)$, called the \emph{fundamental $2$-form};
it defines an {almost symplectic structure}, since it is
generally not closed. From this definition it follows that
\begin{equation}
\omega(X_+,Y_+)=0 \ , 
\label{ome1}
\end{equation}
for all $X_+,Y_+ \in \mathsf{\Gamma}(L_+)$, and 
\begin{equation}
\omega(X_-,Y_-)=0 \ , 
\label{ome2}
\end{equation}
for all $X_-, Y_- \in \mathsf{\Gamma}(L_-)$; in other words, the sub-bundles
$L_\pm$ are also maximally isotropic with respect to $\omega.$ If the
fundamental 2-form $\omega$ is {symplectic}, i.e. ${\rm
  d}\omega=0$, then $(M,K,\eta)$ is called an \emph{almost
  para-K\"{a}hler manifold}. In this case, the conditions \eqref{ome1}
and \eqref{ome2} imply that $L_+$ and $L_-$ are {Lagrangian
sub-bundles} of the tangent bundle $TM$.

An {almost para-Hermitian structure} $(K,\eta)$ on a manifold $M$ can be regarded as a $\sf G$-structure on $M$ given by a reduction of the structure group of the frame bundle $FM$ from ${\sf GL}(2d,\mathbb{R})$ to the subgroup which preserves both $\eta$ and $\omega$:
$$
{\sf G}={\sf O}(d,d)\cap {\sf Sp}(2d,\mathbb{R})={\sf
  GL}(d,\mathbb{R}) \ . 
$$
Integrability of an almost para-Hermitian structure can be described as well. If the eigenbundles $L_+$ and $L_-$ of $K$, such that $TM=L_+\oplus L_-,$ are both integrable then the triple $(M,K,\eta)$
is called a \emph{para-Hermitian manifold}. If in addition the
fundamental 2-form $\omega$ is closed, then $(M,K,\eta)$ is said to be a \emph{para-K\"{a}hler manifold}, in which case it has two transverse Lagrangian foliations with respect to the symplectic structure $\omega.$  

\begin{remark}
The splitting \eqref{eq:splittingTM} of the tangent bundle $TM$ gives rise to a decomposition
of tensors analogous to the type decomposition in complex geometry. In
particular, there is a decomposition for differential forms. We denote $\midwedge^{(+p, -0)}T^* M= \midwedge^p L_+^*$ and $\midwedge^{(+0, -p)}T^* M= \midwedge^p L_-^*,$ so that any $p$-form on $M$ is decomposed according to the splitting
$$ 
\midwedge^p \, T^*M= \bigoplus_{m+n=p} \
\midwedge^{(+m,-n)}\, T^*M \ .
$$
The fundamental 2-form $\omega$ of an almost para-Hermitian manifold is a $(+1,-1)$-form with respect to the almost para-Hermitian structure $(K,\eta),$ since both $L_+$ and $L_-$ are Lagrangian with respect to $\omega,$ i.e. $\omega \in \mathsf{\Gamma}(L_+^* \wedge L_-^*).$
\end{remark}

\medskip

\subsection{Para-Hermitian Vector Bundles} \label{sec:paravector} ~\\[5pt]
The definition of almost para-Hermitian manifold is the special case of
a para-Hermitian structure on the vector bundle $TM.$ This notion can
be generalized in the following way. 
\begin{definition} \label{parahermvector}
Let $E\rightarrow \cQ$ be a real vector bundle with ${\rm rank}(E)=2d.$ A
\emph{para-complex structure} on $E$ is a vector bundle automorphism
$K \in {\sf Aut}_\unit(E)$ covering the identity such that
$K^2=\unit$ and $K\neq\pm\,\mathds{1}$, and the $\pm\,1$-eigenbundles of $K$ have equal rank; the pair $(E,K)$ is a \emph{para-complex vector bundle}. If $E$ admits a fiberwise  metric\footnote{By
  $\odot$ we denote the symmetric tensor product.} $\eta\in \mathsf{\Gamma}(\midodot^2 E^*)$ of signature $(d,d)$ such that 
$$
\eta\big(K(Z),K(W)\big)=-\eta(Z,W) \ , 
$$
for all $Z, W \in \mathsf{\Gamma}(E),$ then the pair $(K, \eta)$ is
a \emph{para-Hermitian structure} on $E$ and the triple $(E,K,\eta)$ is a \emph{para-Hermitian vector bundle}. 
\end{definition}
In this case $K$ admits two eigenbundles $L_\pm$ with
eigenvalues $\pm\,1,$ so that 
$$
E= L_+ \oplus L_- \ ,
$$ 
which are maximally
isotropic with respect to the fiberwise metric $\eta.$ Conversely,
given a vector bundle $E\to\cQ$ of rank $2d$ endowed with a split
signature metric $\eta$, a choice of maximally isotropic sub-bundle
$L_-$ of $E$ determines a short exact sequence of vector bundles
\be
0\longrightarrow L_-\longrightarrow E\longrightarrow
E/L_-\longrightarrow 0 \ ,
\label{eq:EL-}\ee
and a choice of maximally isotropic splitting of this exact sequence
gives a para-Hermitian structure on $E$.
The case $E=TM$ for an
(almost) para-Hermitian manifold $M$ is particularly relevant because it
allows one to formulate conditions for the integrability of the
eigenbundles, and hence on the possibility that $M$ is a foliated
manifold. This will be especially important in our discussions of
sigma-models for para-Hermitian manifolds later on. 

It is straightforward to see that the compatibility condition between $\eta$ and $K$ in Definition~\ref{parahermvector} is equivalent to 
$$
\eta\big(K(Z), W\big)= -\eta\big(Z, K(W)\big) \ , 
$$ 
for all $Z,\, W \in \mathsf{\Gamma}(E).$ The para-Hermitian vector
bundle $E$ is therefore endowed with a skew-symmetric non-degenerate \emph{fundamental $(0,2)$-tensor} $\omega$ given by
$$
\omega(Z,W)= \eta\big(K(Z), W\big) \ , 
$$ 
for all $Z, W \in \mathsf{\Gamma}(E),$ i.e. $\omega\in \mathsf{\Gamma}(\midwedge^2 E^*).$ The eigenbundles $L_\pm \subset E$ are maximally isotropic with respect to $\omega.$

\begin{example} \label{ex:gentanbun}
Let $E=\mathbb{T}\cQ$ be the generalized tangent bundle 
$$
\mathbb{T}\cQ= T\cQ \oplus T^*\cQ
$$ 
over a manifold $\cQ.$ It is naturally endowed with a fiberwise split signature metric 
$$
\eta(X+ \xi, Y+ \nu)= \iota_X \nu + \iota_Y \xi \ , 
$$
for all $X+\xi,Y+ \nu \in \mathsf{\Gamma}(\mathbb{T}\cQ).$ The natural para-complex structure $K$ of $\mathbb{T}\cQ$ is given by
$$
K(X+\xi)= X- \xi \ ,
$$ 
for all $X+\xi \in \mathsf{\Gamma}(\mathbb{T}\cQ),$ so that $T\cQ$ and
$T^*\cQ$ are the respective $\pm\,1$-eigenbundles. Clearly $\eta$ and
$K$ are compatible in the sense of Definition~\ref{parahermvector},
and the bundles $T\cQ$ and $T^*\cQ$ are maximally isotropic with respect to $\eta.$ Thus we obtain a fundamental $(0,2)$-tensor
$$
\omega(X+ \xi, Y+ \nu)= \iota_X \nu - \iota_Y \xi \ , 
$$ 
for all $X+\xi , Y+ \nu \in \mathsf{\Gamma}(\mathbb{T}\cQ),$ which is the additional natural non-degenerate pairing that can be defined in this case \cite{gualtieri:tesi}. 
\end{example}

\begin{example}
A natural extension of Example~\ref{ex:gentanbun} is given by an exact
Courant algebroid $E$ on $\cQ$ specified by an exact sequence
\be\label{eq:CAexact}
0 \longrightarrow T^*\cQ \xlongrightarrow{\rho^*} E
\xlongrightarrow{\rho} T\cQ \longrightarrow 0 \ ,
\ee
with fiberwise metric $\eta,$ Dorfman bracket $\llbracket\,\cdot\, , \, \cdot\, \rrbracket_E$, and anchor map
$\rho: E \rightarrow T\cQ.$ The map $\rho^*:T^*\cQ\to E$ is defined
by\footnote{Here and in the following the superscript ${}^\sharp$ denotes
the bundle isomorphism $E^*\to E$ induced by a non-degenerate
$(2,0)$-tensor in $\mathsf{\Gamma}(E\otimes E)$. For the inverse $(0,2)$-tensor in $\mathsf{\Gamma}(E^*\otimes E^*)$ we
will use the superscript ${}^\flat$ for the induced bundle isomorphism
$E\to E^*$. Conversely, the tensor associated
to a vector bundle isomorphism $T$ will be 
underlined as $\underline{T}\,$.} $\rho^*= \eta^{-1}{}^\sharp\, \circ\, \rho^{\rm t}.$
From the definition of $\rho^*$ and the exactness of the 
sequence \eqref{eq:CAexact}, it follows that the sub-bundle
$\mathrm{Im}(\rho^*) \subset E,$ which is isomorphic to $T^*\cQ,$ is maximally isotropic with respect to $\eta.$
The para-Hermitian structure of $E$ is given by the choice of an
isotropic splitting of \eqref{eq:CAexact}:
$$
s: T\cQ \longrightarrow E \qquad \mbox{with} \quad \rho \circ s=
\unit_{T\cQ} \ . 
$$
It follows that 
$$
E= \mathrm{Im}(s) \oplus \mathrm{Im}(\rho^*)
$$ 
with associated para-complex structure defined by
$$
K_s\big(s(X)+ \rho^*(\xi)\big)=s(X)- \rho^*(\xi) \ , 
$$ 
for all $ X \in \mathsf{\Gamma}(T\cQ)$ and $\xi \in
\mathsf{\Gamma}(T^*\cQ).$  The para-complex structure $K_s$ is
compatible with the metric $\eta,$ and thus $E$ is endowed with a
para-Hermitian structure. This para-Hermitian structure of an exact
Courant algebroid is isomorphic to the para-Hermitian structure of the
generalized tangent bundle $\mathbb{T}\cQ$ from
Example~\ref{ex:gentanbun}. The Dorfman bracket $\llbracket\,\cdot\, ,
\, \cdot\, \rrbracket_E$ maps to the Dorfman bracket on $\mathbb{T}\cQ$ twisted by a representative of the {\v S}evera class in $\mathsf{H}^3(\cQ, \IR)$~\cite{gualtieri:tesi, Severa-letters}.
\end{example}

\medskip

\subsection{Generalized Metrics} \label{genmesubs} ~\\[5pt]
The similarity between exact Courant algebroids and para-Hermitian
geometry suggests the introduction of a suitable notion of {generalized metrics} on almost para-Hermitian manifolds.  
In the following we will closely follow \cite{Jurco2016, gualtieri:tesi, Severa2015, Freidel2019} to introduce a generalized metric compatible with the almost para-Hermitian structure.
\begin{definition} \label{genme}
Let $E\rightarrow \cQ$ be a vector bundle endowed with a fiberwise
metric $\eta$ of signature $(n,m).$ A \emph{generalized 
  metric} on $E$ is an automorphism $I \in {\sf Aut}_\unit(E)$ with
$I^2=\unit$ and $I\neq\pm\,\mathds{1}$ which defines a fiberwise Riemannian metric
\be \nonumber
\cH(Z,W)=\eta\big(I(Z),W\big) \ ,
\ee  
for all $Z, W \in \mathsf{\Gamma}(E).$
\end{definition}
This definition has an equivalent
formulation when the base space $\cQ$ is connected, see e.g.~\cite{Jurco2016}.
\begin{definition} 
Let $E\rightarrow \cQ$ be a vector bundle endowed with a fiberwise
metric $\eta$ of signature $(n,m).$ A \emph{generalized metric} on $E$
is a sub-bundle $V_+ \subset E$ which is maximally positive-definite with respect to $\eta.$
\end{definition}
In this equivalence the eigenbundles of $I$ are $V_+$ associated to
the eigenvalue $+1$ and $V_-,$ the orthogonal complement of $V_+$ with
respect to the metric $\eta,$ associated to the eigenvalue $-1.$ A
generalized metric determines a decomposition 
$$
E=V_+\oplus V_- \ , 
$$
and the restriction of $\eta$ to $V_-$ is negative-definite, so that 
$$
\cH=\eta|_{V_+}-\eta|_{V_-}
$$ 
is indeed a Riemannian metric on $E$.

\begin{remark}
Definition \ref{genme} can also be recast in a different form.
Any generalized metric induces a vector bundle isomorphism $\cH^\flat \in {\sf Hom}(E, E^*)$ which satisfies the condition
$$
\eta(Z,W)= \eta^{-1}\big(\cH^\flat(Z), \cH^\flat(W)\big) \ .
$$ 
Using the induced vector bundle isomorphisms $\eta^{-1}{}^\sharp, \cH^{-1}{}^\sharp\in {\sf Hom}(E^*, E),$ this can be recast in the form 
\be \label{condgenme}
\eta^{-1}{}^\sharp\big(\cH^\flat(Z)\big)= \cH^{-1}{}^\sharp\big(\eta^\flat(Z)\big) \ ,
\ee
for all $Z \in \mathsf{\Gamma}(E),$ so that
$\eta^{-1}{}^\sharp \circ \cH^\flat \in {\sf End}(E)$. In Definition \ref{genme}
this is nothing but $I= \eta^{-1}{}^\sharp \circ \cH^\flat \in {\sf Aut}_\unit(E),$
and \eqref{condgenme} implies that $\eta^{-1}{}^\sharp\circ \cH^\flat$ squares to the
identity map in ${\sf End}(E).$ The tensor induced by this map can
be regarded as a section $\underline{I} \in \mathsf{\Gamma}(E^* \otimes E).$
\end{remark}

\begin{remark}
Almost para-Hermitian manifolds can yield vector bundles which admit a generalized metric. Let $(M, K, \eta)$ be
an almost para-Hermitian manifold. A generalized metric on the
underlying vector bundle $E=TM$ is defined by
$$
\cH(X,Y):=\eta\big(I(X),Y\big) \ , 
$$
for all $X, Y \in \mathsf{\Gamma}(TM),$ 
where $I \in {\sf Aut}_\unit(TM)$ with $I^2=\unit$ and $I\neq\pm\,\mathds{1}$. It satisfies 
$$
\eta^{-1}{}^\sharp\big(\cH^\flat(X)\big)= \cH^{-1}{}^\sharp\big(\eta^\flat(X)\big) \ ,
$$
for all $X \in \mathsf{\Gamma}(TM).$ Then\footnote{The automorphism $I
  \in {\sf Aut}_\unit(TM),$ together with the split signature metric
  $\eta,$ is called a `chiral structure' in~\cite{Freidel2019} where it is denoted by $J$. Thus a chiral structure defines a generalized metric on an almost para-Hermitian manifold.} $I(X)= \eta^{-1}{}^\sharp(\cH^\flat(X)),$
for all $X \in \mathsf{\Gamma}(TM)$.
\end{remark}

Following \cite{Jurco2016, gualtieri:tesi}, we will use the splitting
of the tangent bundle of an almost para-Hermitian manifold to
demonstrate some similarities with the differential geometry of exact
Courant algebroids. To explore these analogies, we establish
\begin{proposition} \label{gbparaherm}
Let $(M,K, \eta)$ be an almost para-Hermitian
manifold. A generalized metric $V_+\subset TM$ defines a unique pair $(g_+, b_+)$
of a fiberwise
Riemannian metric $g_+ \in \mathsf{\Gamma}(\midodot^2 L_+^*)$ on the
sub-bundle $L_+\subset TM$ and a 2-form $b_+\in \mathsf{\Gamma}(\midwedge^2 L_+^*).$
Conversely, any such pair $(g_+,b_+)$ uniquely defines a generalized metric.
\begin{proof}
Since $L_+$ and $L_-$ are both maximally isotropic with respect to
$\eta$, and $V_+$ is maximally positive-definite, it follows that $L_+ \cap V_+= L_- \cap V_+ = 0.$ 
The orthogonal complement $V_-$ is maximally negative-definite with
respect to $\eta,$ so also $L_+ \cap V_-= L_- \cap V_- = 0.$
Thus given any vector bundle isomorphism $\gamma \in {\sf Hom}(L_+,
L_-),$ we can regard $V_+$ as the bundle
$$
V_+=\big\{X_V=X_+ + \gamma(X_+) \ \big| \  X_+ \in L_+ \big\} \ .
$$
Positive-definiteness of $V_+$ also implies
$$
\eta(X_V,X_V)=\eta\big(X_+ + \gamma(X_+),X_+ + \gamma(X_+)\big)= 2\,
\eta\big(\gamma(X_+), X_+ \big) \geq 0 \ .
$$

Since $\gamma \in {\sf
  Hom}(L_+,L_-)$ is a vector bundle isomorphism, let us consider the
associated tensor $\underline{\gamma} \in \mathsf{\Gamma}(L_+^*\otimes L_-)$ and
decompose it into a symmetric part and a skew-symmetric part: $\underline{\gamma}= \underline{\gamma_g} + \underline{\gamma_b},$ where $\underline{\gamma_g}
, \underline{\gamma_b} \in \mathsf{\Gamma}(L_+^*\otimes L_-)$ induce 
vector bundle morphisms $\gamma_g,\gamma_b$ such that
\be \label{decgam}
\eta\big(\gamma_g(X_+),Y_+\big)= \eta(X_+, \gamma_g(Y_+)\big) \qquad
\mbox{and} \qquad \eta\big(\gamma_b(X_+),Y_+\big)= -\eta\big(X_+,
\gamma_b(Y_+) \big) \ ,
\ee
for all $X_+,Y_+ \in \mathsf{\Gamma}(L_+).$ 
Then
$$
\eta(X_V,X_V)=2\,\eta \big(\gamma_g (X_+), X_+\big) \geq 0 \ ,
$$ 
and $\underline{\gamma_g}$ is non-degenerate.
Thus the symmetric part of $\underline{\gamma} \in
\mathsf{\Gamma}(L_+^* \otimes L_-)$ defines a fiberwise Riemannian
metric on $L_+,$ which we denote by $g_+ \in \mathsf{\Gamma}(\midodot^2 L_+^*)$,
such that
$$
g_+(X_+, Y_+)= \eta\big(\gamma_g(X_+),Y_+\big) \ ,
$$ 
for all $X_+, Y_+ \in \mathsf{\Gamma}(L_+).$ 
Similarly, the inverse map $\gamma_g^{-1}: L_- \rightarrow L_+$
induces a fiberwise metric on $L_-$ which we denote by $g_-\in \mathsf{\Gamma}(\midodot^2 L_-^*).$ The
fiberwise metrics $g_+$ and $g_-$ are not independent, since 
\be \label{eq:g-}
g_-(X_-,Y_-)= g^{-1}_+\big(\eta^\flat(X_-),\eta^\flat(Y_-)\big) \ ,
\ee
for all $X_-,Y_-\in\mathsf{\Gamma}(L_-)$.
The skew-symmetric part of $\underline{\gamma}$ defines a 2-form
$b_+ \in
\mathsf{\Gamma}(\midwedge^2 L_+^*)$, such that
\be\nonumber
b_+(X_+,Y_+) = \eta\big(\gamma_b(X_+),Y_+\big)
\ee
for all $X_+, Y_+ \in \mathsf{\Gamma}(L_+).$ 

We can now introduce an automorphism $I \in {\sf Aut}_\unit(TM)$ by 
\be
I=
\begin{pmatrix}
- \gamma_g^{-1} \circ \gamma_b & \gamma_g^{-1} \\
\gamma_g - \gamma_b \circ \gamma_g^{-1} \circ \gamma_b & \gamma_b \circ \gamma_g^{-1}
\end{pmatrix}
\ , \nonumber
\ee
in the splitting \eqref{eq:splittingTM} defined by $K.$ It is
straightforward to show that $I^2=\unit$, and that the eigenbundles of
$I$ are $V_+$ and its orthogonal complement $V_-$ with respect to
$\eta$.\footnote{The eigenbundle $V_-$ can be regarded as 
$$
V_-=\big\{X_+
  +(-\gamma_g+ \gamma_b)(X_+) \ \big| \ X_+ \in L_+\big\} \ .
$$} We finally obtain the corresponding Riemannian metric $\cH,$ as in Definition~\ref{genme}, given by 
\begin{align*}
\cH(X,Y)&=  \eta\big(I(X),Y\big) \\[4pt]
&= \eta\big(\gamma_g(X_+),Y_+\big)
  -\eta\big(\gamma_b(\gamma_g^{-1}(\gamma_b(X_+))), Y_+
  \big)-\eta\big(\gamma_g^{-1}(\gamma_b(X_+)),Y_- \big) \\
& \qquad + \eta\big(\gamma_b(\gamma_g^{-1}(X_-)), Y_+\big)+
  \eta\big(\gamma_g^{-1}(X_-),Y_- \big) \\[4pt]
&= g_+(X_+,Y_+) +g_-\big(\gamma_b(X_+),
  \gamma_b(Y_+)\big)-g_-\big(\gamma_b(X_+),Y_- \big) \\
& \qquad - g_-\big(X_-, \gamma_b(Y_+)\big)+ g_-(X_-,Y_-)
\end{align*}
for all $X, Y \in \mathsf{\Gamma}(TM),$ 
where we used the skew-symmetry of $\gamma_b$ from \eqref{decgam}. 
In matrix form, by fixing the splitting \eqref{eq:splittingTM} of $TM$ associated with $K,$ the generalized metric reads
\be \label{gemeju}
\cH= \bigg( \begin{matrix}
g_+ + \underline{\gamma_b^{\rm t}}\ g_-\ \underline{\gamma_b} & -\underline{\gamma_b^{\rm t}}\ g_- \\ -g_-\ \underline{\gamma_b} & g_-
\end{matrix} \bigg) \ ,
\ee
where $\gamma_b^{\rm t}: L_-^* \rightarrow L_+^*$ is the transpose
map.

Conversely, starting with a pair $(g_+, b_+),$ we can write the
generalized metric $\cH$ in \eqref{gemeju} and then identify the
sub-bundle $V_+$ by using the inverse of the metric $\eta.$  
\end{proof} 
\end{proposition}

\begin{example}\label{ex:sl2c}
Let $M={\sf SL}(2,\mathbb{C})$ regarded as a six-dimensional
real Lie group. As a complex Lie group, it has a non-degenerate
Cartan-Killing form 
$$
\mathrm{Tr}: {\sf SL}(2, \mathbb{C}) \longrightarrow \mathbb{C} \ .
$$ 
As a real Lie group, ${\sf SL}(2, \mathbb{C})$ inherits two
distinct non-degenerate real pairings $2\,\Im\circ \mathrm{Tr}$ and
$2\,\Re\circ\mathrm{Tr}.$ The former has split signature and defines
the Manin triple polarization 
$$
{\sf SL}(2,\IC)={\sf SU}(2)\Join{\sf SB}(2,\IC) \ ,
$$
where ${\sf SB}(2,\IC)$ is the Borel subgroup of $2\times2$ upper
triangular complex matrices, while the
latter defines a generalized metric on the tangent bundle $ T{\sf SL}(2, \mathbb{C}).$

For this, we recall that, in a suitable basis of the Lie algebra $\mathfrak{sl}(2,\IC),$ the generators satisfy the commutation relations
\be\nonumber
[T_i,T_j] = \tfrac12\,\varepsilon_{ij}{}^k\, T_k \ , \quad [ T_i,\tilde T^j] = \tfrac12\,
\varepsilon_{ki}{}^j \,\tilde T^k-\tfrac12\,
\varepsilon^{kjl}\,\varepsilon_{l3i}\, T_k\quad \mbox{and} \quad [\tilde T^i,\tilde
T^j]=\tfrac12\, \varepsilon^{ijl}\, \varepsilon_{l3k}\,\tilde T^k
\ .
\ee
The splitting of the Lie algebra 
$$
\mathfrak{sl}(2,\IC)=
\mathfrak{su}(2)\oplus {\mathfrak{sb}}(2,\IC)
$$ 
as a vector space
induces a left-invariant para-complex structure on ${\sf SL}(2,
\mathbb{C}).$  The ${\sf O}(d,d)$-invariant metric $\eta$ compatible
with the para-complex structure is obtained from the Cartan-Killing
form as $\langle a,b\rangle = 2\,\Im \big(\mathrm{Tr}(a\,b)\big),$ for
$a,b\in\mathfrak{sl}(2,\mathbb{C}),$ which gives the duality pairing
between the Lie subalgebras $\mathfrak{su}(2)$ and
$\mathfrak{sb}(2,\mathbb{C})$, with respective generators $\{T_i\}$
and $\{\tilde T^i\}$, and hence realizes ${\sf SU}(2)$ and
${\sf SB}(2,\IC)$ as dual Lie subgroups of the Drinfel'd double ${\sf
  SL}(2,\mathbb{C})$.\footnote{See \cite{SzMar} for further details
  regarding para-Hermitian structures on Drinfel'd double Lie groups.}
Writing 
$$
e_i^\pm=\tfrac1{\sqrt2}\,\big(T_i\pm(\delta_{ij}\pm\varepsilon_{ij3}\,\tilde
T^j)\big) \ ,
$$
from the isotropy of $\mathfrak{su}(2)$ and
$\mathfrak{sb}(2, \mathbb{C})$ it follows that 
$$
\langle
e_i^+,e_j^+\rangle=\delta_{ij}=-\langle e_i^-,e_j^-\rangle \qquad
\mbox{and} \qquad
\langle e_i^+,e_j^-\rangle=0 \ .
$$

On the other hand, we also see that
$\langle e_i^\pm,e_j^\pm\rangle=\pm\, 2\,\Re\bigl( {\rm Tr}(e^\pm_i \,
e^\pm_j)\bigr).$ The generalized metric $\mathcal{H}$ is therefore
obtained from the other natural inner product $(a,b) = 2\,\Re\big(\mathrm{Tr}(a\,b)\big)$ (which does not define a Manin triple polarization), for which one writes
\be\label{eq:sl2Cmetric}
\cH = \delta^{ij}\, \big(e^{*\,+}_i\otimes e^{*\,+}_j+ e^{*\,-}_i\otimes e^{*\,-}_j\big) \ .
\ee
The scalar product $2\,\Re\big(\mathrm{Tr}(\, \cdot
\,)\big)$ thus identifies a generalized metric: the sub-bundle $V_+
\subset T {\sf SL}(2, \mathbb{C})$ spanned by $e_i^+$ which is defined
via the map $\gamma: \mathfrak{su}(2) \rightarrow \mathfrak{sb}(2,
\mathbb{C})$ given by 
$$
\underline{\gamma} = (\delta_{ij}+
\varepsilon_{ij3})\,{T^*}^i \otimes \tilde{T}^j \ .
$$
Expanding this out
with respect to the splitting $\mathfrak{sl}(2,\IC)=
\mathfrak{su}(2)\oplus{\mathfrak{sb}}(2,\IC)$, and comparing with
\eqref{gemeju}, then identifies the metric $(g_+)_{ij}=\delta_{ij}$ as the
Cartan-Killing metric and the 2-form $(b_+)_{ij}=\varepsilon_{ij3}$ on
$\mathfrak{su}(2)$, which lead to the standard round metric and
Kalb-Ramond field (whose $H$-flux is the volume form) on the 3-sphere
${\sf SU}(2)=\IS^3$; see~\cite{mpv, mpv2} for further details. This
example will be generalized to generic Drinfel'd double Lie groups in
Section~\ref{sec:Manin}.
\end{example}

The statement of Proposition \ref{gbparaherm} has a counterpart for
any para-Hermitian vector bundle, with exactly the same proof,
resulting in

\begin{proposition} 
Let $(E,K, \eta)$ be a para-Hermitian vector bundle over a manifold $\cQ$. A
generalized metric $V_+\subset E$ defines a unique pair $(g_+, b_+),$
where $g_+ \in \mathsf{\Gamma}(\midodot^2 L_+^*)$ is a fiberwise
Riemannian metric on the sub-bundle $L_+\subset E$ and $b_+\in
\mathsf{\Gamma}(\midwedge^2 L_+^*)$ is a 2-form on $L_+$.
Conversely, any such pair $(g_+,b_+)$ uniquely defines a generalized metric on $(E,K,\eta)$.
\end{proposition}

\begin{example}
Let $E=\mathbb{T}\cQ=T\cQ\oplus
T^*\cQ$ be the generalized tangent bundle over a manifold $\cQ$. A generalized metric $V_+\subset \mathbb{T}\cQ$ is equivalent to
a Riemannian metric $g_+$ and a 2-form $b_+ $ on $\cQ$. In this case the
bundle maps $\gamma_g$ and $\gamma_b$ appearing in the proof of
Proposition~\ref{gbparaherm} correspond to $g_+$ and $b_+$ themselves. See \cite{Jurco2016, gualtieri:tesi} for further details.
\end{example}

\begin{remark}
A generalized metric on an almost para-Hermitian manifold can also be related to a generalized metric on a generalized tangent bundle. 
 For this, we assume that the eigenbundle $L_-$ of the almost
 para-Hermitian manifold is involutive, i.e. it admits integral
 manifolds given by the leaves of a regular foliation $\cF_-.$ We can
 construct a generalized tangent bundle $\mathbb{T}\cS_-= T\cS_-
 \oplus T^*\cS_-$ on a leaf $\cS_-$ of the foliation. There is a morphism from
 $\mathbb{T}\cS_-$ to $TM$ covering the inclusion
 $\cS_-\hookrightarrow M$, which is induced at the level of sections by the split signature metric $\eta$ through
$$
{\sf p}_-:\mathsf{\Gamma}(\mathbb{T}\mathcal{S}_-) \longrightarrow \mathsf{\Gamma}(TM) 
\ , \quad X+\xi \longmapsto {\sf
  p}_-(X+\xi)=X+\eta^{-1}{}^\sharp(\xi) \ . 
$$
This pulls back a generalized metric on a foliated almost
para-Hermitian manifold, with the foliation associated with the almost
para-complex structure, to a generalized metric on the generalized
tangent bundle $\mathbb{T}\cS_-$ constructed on a leaf space $\cS_-$
of the foliation $\cF_-$. This description differs from that of~\cite{Freidel2017,Freidel2019,
   Svoboda2018, SzMar} where the union of the leaf spaces $\cF_-$ was used
 instead of a single leaf space $\cS_-$.\footnote{See also~\cite{Vysoky:2019cra} for a similar approach to
 this relation in the setting of exact Courant algebroids.}
\end{remark}

\medskip

\subsection{Born Geometry} \label{sec:Born} ~\\[5pt]
We will now connect with the formalism of \cite{Freidel2019}, starting
with the following notion.
\begin{definition} \label{compagenmetr}
A \emph{compatible generalized metric} on an almost para-Hermitian manifold $(M,K,\eta)$ is a generalized metric $\mathcal{H}$ on $M$ which is compatible with the fundamental 2-form $\omega$ in the sense that
\be\nonumber
\omega^{-1}{}^\sharp\big(\mathcal{H}^\flat(X)\big) =
-\mathcal{H}^{-1}{}^\sharp\big(\omega^\flat(X)\big) \ ,  
\ee
or equivalently
$$ 
\omega^{-1}\big(\cH^\flat(X), \cH^\flat(Y)\big)=-\omega(X,Y)\ ,
$$
for all $ X, Y \in \mathsf{\Gamma}(TM).$ 
The triple $(\eta,\omega,\mathcal{H})$ is a \emph{Born geometry} on
$M$ and $(M,\eta,\omega,\mathcal{H})$ is a \emph{Born manifold}.
\end{definition}

Definition \ref{compagenmetr} is equivalent to the original definition
of \cite{Freidel2019} where the compatibility conditions are written as
\be\nonumber
\eta^{-1}{}^\sharp\,\circ \, \mathcal{H}^\flat = \mathcal{H}^{-1}{}^\sharp\,\circ\, \eta^\flat \qquad \mbox{and} \qquad \omega^{-1}{}^\sharp\,\circ \, \mathcal{H}^\flat = -\mathcal{H}^{-1}{}^\sharp\,\circ \, \omega^\flat \ .
\ee
A Born geometry can be regarded as a $\sf G$-structure on $M$ with 
$$
{\sf G}={\sf O}(d,d)\cap{\sf Sp}(2d,\mathbb{R})\cap {\sf O}(2d) = {\sf
  O}(d) \ . 
$$ 
A Born geometry is also a special type of generalized metric, as we show in
\begin{proposition}\label{prop:cHg+}
A Born structure on an almost para-Hermitian manifold $(M,K,\eta)$ is a generalized metric $\cH$ specified solely by a fiberwise metric $g_+$ on the eigenbundle $L_+.$ 
\begin{proof}
A generalized metric on $(M, K, \eta)$ satisfies $\eta^{-1}{}^\sharp\,\circ
\,\mathcal{H}^\flat = \mathcal{H}^{-1}{}^\sharp\,\circ \, \eta^\flat$ by definition. Using
\eqref{gemeju} it is then easy to see that the condition $\omega^{-1}{}^\sharp\,\circ \,\mathcal{H}^\flat = -\mathcal{H}^{-1}{}^\sharp\,\circ \,\omega^\flat$ holds if and only if $\gamma_b=0,$ i.e. $b_+=0.$
\end{proof}
\end{proposition}

In other words, the compatible Riemannian metric $\cH$ can be regarded
as a choice of a metric on the sub-bundle $L_+$ in the polarization
\eqref{eq:splittingTM} associated with $K.$ In this polarization, any vector field decomposes as 
$$
X= X_+ + X_-  \ \in \ \mathsf{\Gamma}(TM)\ , 
$$ 
where $X_+ \in \mathsf{\Gamma}(L_+)$ and $X_- \in
\mathsf{\Gamma}(L_-),$ and thus we write
\be \label{eq:diaghermmetric}
\cH(X,Y)= g_+(X_+ , Y_+)+g_-(X_-, Y_-)\ ,
\ee
where $g_+$ is a fiberwise metric on the sub-bundle $L_+$ and $g_-$ is
given by \eqref{eq:g-}. In matrix notation, the compatible generalized metric reads
\be \nonumber
\mathcal{H} = \bigg( \begin{matrix}
g_+ & 0 \\ 0 & g_-
\end{matrix} \bigg) \ ,
\ee
in the splitting given by $K.$

\begin{example} \label{fiberborn}
Let $\pi: M \rightarrow \cQ$ be a fibered manifold and define a splitting
of $TM= {\rm Im}(s) \oplus L_{\tt v}(M)$ as in
Example~\ref{fiberexample}. This defines a split signature metric $\eta$ on
$TM.$ Choose an isotropic splitting of the short exact sequence
\eqref{bunshort} with respect to $\eta$. Hence we choose an Ehresmann
connection $s: \pi^*(T\cQ) \rightarrow TM$ for which the almost para-complex structure $K_s$ induced by  the splitting and the split signature metric $\eta$ on $TM$ are compatible in the sense of Definition~\ref{paraher}, i.e. $TM$ carries an almost para-Hermitian structure $(K_s,\eta).$ Assume that the base $\cQ$ is a Riemannian manifold with metric $g.$ The horizontal lift of $g,$ defined by 
$$
g_+\big(s(X), s(Y)\big)= g(X,Y)\ , 
$$ 
for all $X, Y \in \mathsf{\Gamma}(T\cQ),$ gives a fiberwise Riemannian
metric on  ${\rm Im}(s)$. This then defines a Born geometry on $M$ with compatible generalized metric $\cH$ given by 
$$
\cH\big(s(X) + X_{\tt v}, s(Y)+Y_{\tt v}\big)=g_+\big(s(X),
s(Y)\big)+g_-(X_{\tt v}, Y_{\tt v})\ ,
$$
for all $X, Y \in \mathsf{\Gamma}(T\cQ)$ and $X_{\tt v}, Y_{\tt v} \in
\mathsf{\Gamma}(L_{\tt v}(M))$. In other words, the pullback $g_+= \pi^*\, g$ defines a compatible
generalized metric on $TM.$ This is the standard example of a
bundle-like metric obtained from the lift of structures\footnote{See \cite{Kappos2002} for the lifts to
tangent bundles as Sasaki metrics.} from the base
$\cQ$ to the total space $M$; such a metric is characterized by
the property of having a horizontal component which is constant along
the fibers, and these will play a prominent role in our discussions of
gauged sigma-models later on. Since any manifold $\cQ$ admits a Riemannian metric, we can
always define a Born structure of this type for any almost
para-Hermitian structure on $TM,$ where $M\to \cQ$ is a fibered manifold.
\end{example}

\begin{remark}
A similar definition of a Born geometry can be given for a
para-Hermitian vector bundle $E\to \cQ$. Then an analogous description to that
above follows by simply
replacing $TM$ with $E$ everywhere.
\end{remark}

\medskip

\subsection{Generalized Flux Formulation}  \label{sec:fluxes} ~\\[5pt]
The generalized flux picture of double field theory~\cite{Siegel:1993xq,Siegel:1993th,Geissbuhler2013,Aldazabal2013} can be described in the framework of para-Hermitian
geometry by appealing to a local characterization of the eigenbundles $L_\pm$ underlying an almost para-Hermitian manifold $(M, K, \eta).$
Specifying two complementary sub-bundles of $TM$ is equivalent to fixing a local frame on $\mathsf{\Gamma}(TM),$ i.e. a set of local vector fields $\{Z_I\} \subset \mathsf{\Gamma}(TM)$ that are linearly independent over $C^{\infty}(M)$, and which splits into two sets $\{Z_i\}$ and $\{ \tilde{Z}^i\}$ respectively spanning $\mathsf{\Gamma}(L_+)$ and $\mathsf{\Gamma}(L_-)$ locally. The basis $\{Z_I\}$ closes a Lie algebra 
\be \label{structfun}
[Z_I,Z_J]= {C_{IJ}}^K\, Z_K \ ,
\ee
where ${C_{IJ}}^K \in C^{\infty}(M),$ which can be written in the form
of a Roytenberg algebra
\begin{align}
[Z_m, Z_n]&={f_{mn}}^k\, Z_k + H_{mnk}\, \tilde{Z}^k \ , \nonumber \\[4pt] [ Z_m,\tilde{Z}^n]&= {f_{km}}^n \, \tilde{Z}^k + Q_m{}^{nk}\, Z_k \ , \label{fluxes}\\[4pt] [\tilde{Z}^m, \tilde{Z}^n]&=Q_k{}^{mn}\, \tilde{Z}^k + R^{mnk}\, Z_k \ , \nonumber
\end{align}
where the structure functions are called \emph{generalized fluxes} associated with the chosen frame. The Jacobi identity for the Lie brackets \eqref{structfun} then yields the Bianchi identities for the generalized fluxes. Here we did not
assume that either of the sub-bundles $L_\pm$ are integrable, and in general neither $\{Z_i\}$ nor $\{\tilde{Z}^i\}$ close a Lie subalgebra. 

Analogously, we may consider the dual local coframe given by 1-forms $\{\Theta^I\}\subset \mathsf{\Gamma}(T^*M)$, that split into a set $\{\Theta^i\}$ which spans $\mathsf{\Gamma}(L_+^*)$ and a set $\{\tilde{\Theta}_i\}$ which spans $\mathsf{\Gamma}(L_-^*).$
Since $\{ \Theta^I\}$ is a coframe, it satisfies the Maurer-Cartan equations
$$ 
\de \Theta^K = -\tfrac{1}{2}\, {C_{IJ}}^K\, \Theta^I \wedge \Theta^J \ ,
$$ 
which can also be read as
\begin{align}
\de \Theta^p &= -\tfrac{1}{2}\,\bigl({f_{mn}}^p\, \Theta^m \wedge
               \Theta^n +R^{mnp}\, \tilde{\Theta}_m \wedge
               \tilde{\Theta}_n \bigr) - Q_n{}^{mp}\, \Theta^n \wedge
               \tilde{\Theta}_m \ , \nonumber \\[4pt]
\de \tilde{\Theta}_p &= {f_{pm}}^n\, \tilde{\Theta}_n \wedge
                       \Theta^m-\tfrac{1}{2}\, \big(Q_p{}^{mn}\,
                       \tilde{\Theta}_m \wedge \tilde{\Theta}_n +
                       H_{pmn}\, \Theta^m \wedge
                       \Theta^n\big) \label{eq:MC} \ .
\end{align}

If the chosen frame $\{Z_I\}$ diagonalizes the almost para-complex
structure $K,$ then using the compatibility conditions the almost para-Hermitian structure $(K, \eta, \omega)$ can be written as
\be \label{diagpara}
\underline{K} = \Theta^i \otimes Z_i - \tilde{\Theta}_i \otimes \tilde{Z}^i\ , \quad \eta= \Theta^i \otimes \tilde{\Theta}_i+\tilde{\Theta}_i \otimes \Theta^i \qquad \mbox{and} \qquad \omega=\Theta^i \wedge \tilde{\Theta}_i \ .
\ee 
Since an almost para-Hermitian structure $(K, \eta, \omega)$ is
a ${\sf GL}(d, \IR)$-structure, this means that in the local
description \eqref{diagpara}, $(K, \eta, \omega)$ retain the same form under transformations given by
\be \label{gldstruct}
\underline{\cA}= 
\begin{pmatrix}
\cA & 0 \\
0 & \cA^{\rm t}
\end{pmatrix}
\qquad \mbox{with} \quad \cA \in {\sf GL}(d, \IR)\ ,
\ee
in the polarization $TM=L_+\oplus L_-.$ The effect of these
transformations on the frame is to change the local structure
functions describing the Lie algebra \eqref{structfun}.

In addition to obstructing integrability of the eigenbundles $L_\pm$,
the generalized fluxes also present obstructions to the closure of the
fundamental 2-form $\omega$, i.e. to $(M,K,\eta)$ being an almost
para-K\"ahler manifold. This can be seen by introducing the field strength 
$$
\cK=\de\omega \ ,
$$
and using \eqref{diagpara} together with the Maurer-Cartan equations
\eqref{eq:MC} to write it in the coframe $\{\Theta^I\}$ as
\begin{align} 
\cK =
            \tfrac12\,\big( & H_{ijk}\,\Theta^i\wedge\Theta^j\wedge\Theta^k
            + f_{ij}{}^k\,\Theta^i\wedge\Theta^j\wedge\tilde\Theta_k
\nonumber \\ &  -
       Q_i{}^{jk}\,\Theta^i\wedge\tilde\Theta_j\wedge\tilde\Theta_k +
       R^{ijk}\,
       \tilde\Theta_i\wedge\tilde\Theta_j\wedge\tilde\Theta_k \big) \ .
\label{eq:cK}\end{align}
The Bianchi identity $\de\cK=0$ is equivalent to the Jacobi identity
for the Lie brackets \eqref{structfun}, and so also yields the Bianchi identities for the generalized fluxes.

We conclude this section by describing the local form of a Born geometry.
By Proposition~\ref{prop:cHg+}, a compatible generalized metric $\cH$ on an almost para-Hermitian manifold is given by a fiberwise Riemannian metric $g_+$ on the sub-bundle $L_+.$ In particular, we can always write it in a given coframe $\{ \Theta^I\}$ for $T^*M=L^*_+\oplus L^*_-$ as 
\be  \nonumber
\cH= (g_+)_{ij}\, \Theta^i \otimes \Theta^j + (g_-)^{ij}\,\tilde{\Theta}_i \otimes \tilde{\Theta}_j \ ,
\ee
where the metric $g_-$ on the sub-bundle $L_-$ is given by \eqref{eq:g-}.
A Born structure is an ${\sf O}(d)$-structure and there always exists an element $\cA \in {\sf O}(d)$ inducing a transformation $\underline{\cA}$ of the coframe for which the compatible generalized metric is locally flat:
\be \label{locgm}
\cH= \delta_{ij}\, \Theta_\cA^i \otimes \Theta_\cA^j + \delta^{ij}\,\tilde{\Theta}^\cA_i \otimes \tilde{\Theta}^\cA_j\ , 
\ee
where $\{ \Theta_\cA^I \}=\{ \Theta_\cA^i\, , \tilde{\Theta}^\cA_i \}$ is the coframe obtained from $\{ \Theta^I \}$ by applying the ${\sf O}(d)$-transform{-}ation $\cA$ in the form \eqref{gldstruct}. Such a transformation leaves the almost para-Hermitian structure $(K, \eta, \omega)$ in the same form \eqref{diagpara}, since ${\sf O}(d)\subset{\sf GL}(d, \IR)$~\cite{Freidel2019}.

\section{Generalized T-Duality} \label{ContTduality}

In this section we shall discuss how all of the structures introduced
in Section~\ref{intropara} transform under the action of a special
group generating what we will call generalized T-duality transformations.

\medskip

\subsection{${\sf O}(d,d)(M)$-Transformations} \label{int:Odd} ~\\[5pt]
Given a para-Hermitian manifold $M$, we shall characterize ${\sf
  O}(d,d)(M)$-transformations as a subgroup of the group of
fiber-preserving automorphisms of the tangent bundle ${\sf Aut}(TM).$
We start in a more general setting.

\begin{definition}
An \emph{automorphism} of
a vector bundle $\pi:E\to \cQ$ is a pair $\vartheta=(f, \bar{f}\,),$
where $f: \cQ \rightarrow \cQ$ is a diffeomorphism and $\bar{f}: E
\rightarrow E$ is a vector bundle isomorphism for which the diagram
\begin{center} 
\begin{tikzcd}
E \arrow[r, "\bar{f}"] \arrow[d, swap, " \pi "]
& E \arrow[d, "\pi "] \\ 
\cQ \arrow[r, swap, "f"]
& \cQ
\end{tikzcd}
\end{center}
commutes, i.e. $\pi \circ \bar{f}= f \circ \pi $. The map $\bar{f}$ is
a \emph{covering} of $f.$ The set of automorphisms of $E$ forms a
group under composition of diffeomorphisms of $\cQ$ and bundle
isomorphisms of $E$, which we denote by ${\sf Aut}(E)$.
\end{definition}
The action of an element $\vartheta=(f, \bar{f}\,) \in {\sf Aut}(E)$ on
sections of $E$ is denoted by $\bar{f}(Z) \in \mathsf{\Gamma}(E),$ for
all $Z \in \mathsf{\Gamma}(E).$ An important subgroup of
${\sf Aut}(E)$ is given by the automorphisms of $E$ covering the
identity, which as before will be denoted by ${\sf Aut}_\unit (E).$
Denoting by ${\sf Diff}(\cQ)$ the group of diffeomorphisms of the
manifold $\cQ$, these fit into the exact sequence of groups
$$ 
1 \longrightarrow{} {\sf Aut}_\unit (E) \longrightarrow{} {\sf Aut}(E)
\longrightarrow{} {\sf Diff}(\cQ)\ .
$$ 

\theoremstyle{definition} 
\begin{definition}
Let $E \rightarrow \cQ$ be a vector bundle with ${\rm rank}(E)=2d$ which
is endowed with a fiberwise  metric $\eta$ of signature $(d,d)$. Let ${\sf
  O}(d,d)(E)$ be the subgroup of the automorphism group ${\sf
  Aut}(E)$ which preserves $\eta,$ i.e. $\vartheta=(f,\bar{f}\,) \in
{\sf Aut}(E)$ is an element of the subgroup ${\sf O}(d,d)(E) \subset
{\sf Aut}(E)$ if and only if
\be \label{oddcon}
(\bar f{}^*\eta)(Z,W)=\eta\big(\bar{f} (Z),\bar{f} (W)\big)=\eta(Z,W) \ ,
\ee
for all $Z, W \in\mathsf{\Gamma}(E).$
\end{definition}

This section will be mainly dedicated to the case $E=TM,$ which arises
in considerations of almost para-Hermitian manifolds. 
 Then the subgroup\footnote{Here we denote this subgroup by ${\sf O}(d,d)(M)$ instead of ${\sf O}(d,d)(TM)$ for brevity.} ${\sf O}(d,d)(M) \subset {\sf Aut}(TM)$ is the natural group of isometries of the almost para-Hermitian
manifold $(M,K,\eta),$ which we call the
\emph{generalized T-duality group}. The elements of this subgroup are
also called \emph{changes of polarization} for reasons that will
become apparent later on.

\begin{example}\label{ex:autdiff}
A particularly relevant class of elements in ${\sf O}(d,d)(M)$ arise
from diffeomorphisms of $M$. Let $f \in {\sf Diff}(M)$ be a diffeomorphism of
the base space $M$ whose pullback $f^*$ preserves the metric $\eta,$ i.e. $f^*\eta=\eta.$ Then $f$ induces an element $\vartheta\in{\sf O}(d,d)(M)$ given by $\vartheta=(f, f_*),$ where $f_*: TM \rightarrow TM$ is the pushforward by $f.$ In the following we will discuss other classes of elements belonging to ${\sf O}(d,d)(M),$ particularly $B$-transformations which are examples of automorphisms covering the identity.
\end{example}

We are also particularly interested in the action induced by ${\sf
  Aut}(TM)$ on ${\sf End}(TM),$ i.e.~on smooth $(1,1)$-tensor
fields. Let $S \in {\sf End}(TM)$ and $\vartheta=(f,\bar{f}\,)\in {\sf
  Aut}(TM).$ Then the \emph{pullback} $S_\vartheta \in {\sf End}(TM)$
is defined by the commutative diagram
\begin{center} 
\begin{tikzcd}
TM \arrow[r, "\bar{f}"] \arrow[d, swap, " S_\vartheta "]
& TM \arrow[d, "S "] \\ 
TM \arrow[r, swap, "\bar f"]
& TM
\end{tikzcd}
\end{center}
which implies 
$$
S_\vartheta= \bar{f}^{-1}\circ S\circ \bar{f} \ .
$$ 
Similarly, the \emph{pushforward} $S^\vartheta\in {\sf End}(TM)$ is
defined by the commutative diagram
\begin{center} 
\begin{tikzcd}
TM \arrow[r, "\bar{f}"] \arrow[d, swap, " S "]
& TM \arrow[d, "S^\vartheta "] \\ 
TM \arrow[r, swap, "\bar f"]
& TM
\end{tikzcd}
\end{center}
so that 
$$
S^\vartheta=\bar{f}\circ S\circ \bar{f}^{-1} \ .
$$ 
It then follows that 
$$
S^\vartheta=S_{\vartheta^{-1}} \qquad \mbox{and} \qquad
S_{\vartheta}=S^{\vartheta^{-1}} \ , 
$$ 
for all $ \vartheta=(f, \bar{f}\,)\in {\sf Aut}(TM).$

We can now apply these considerations to almost para-Hermitian
manifolds to get

\begin{proposition} \label{push}
Let $(M, K, \eta)$ be an almost para-Hermitian manifold with
fundamental 2-form $\omega,$ and let $\vartheta=(f,\bar{f}\,)\in {\sf
  O}(d,d)(M).$ Then the {pullback} of $K$ by $\vartheta,$
$K_\vartheta= \bar{f}^{-1}\circ K\circ \bar{f},$ and $\eta$ form an almost para-Hermitian structure $(K_\vartheta, \eta)$ on $M$ whose fundamental 2-form is $\omega_\vartheta=\bar{f}^* \omega.$  

\begin{proof}
We first show that $K_\vartheta$ is an almost para-complex
structure. Since $K\in {\sf End}(TM)$ and $\vartheta=(f,\bar{f}\,) \in
{\sf Aut}(TM)$, it follows that $\bar{f}^{-1}\circ K \circ \bar{f} \in
{\sf End}(TM),$ and therefore $K_\vartheta \in {\sf End}(TM).$ Then
$$
K_\vartheta^2= \bar{f}^{-1}\circ K\circ \bar{f}\circ\bar{f}^{-1}\circ K\circ \bar{f}=
\bar{f}^{-1}\circ K^2\circ \bar{f}= \unit\ ,
$$ 
where we used $K^2=\unit.$ In this way ${\sf Aut}(TM)$ maps an almost
para-complex structure into an almost para-complex structure, and so $K_\vartheta$ is an almost para-complex structure.

We now prove that $(K_\vartheta, \eta)$ satisfies the compatibility condition \eqref{compcon}:
\begin{align*}
\eta\big(K_\vartheta(X), K_\vartheta(Y)\big) &
                                               =\eta\big(\bar{f}^{-1}\,K(\bar{f}(X)),
                                               \bar{f}^{-1}\,K(\bar{f}(Y))
                                               \big) \\[4pt]
&= \big((\bar{f}^{-1})^*\eta\big)\big(K(\bar{f}(X)),K(\bar{f}(Y))\big) \\[4pt]
&= \eta\big(K(\bar{f}(X)),K(\bar{f}(Y))\big) \\[4pt] 
&= - \eta\big(\bar{f}(X), \bar{f}(Y)\big) \\[4pt]
&= - (\bar{f}^*\eta)(X,Y) \\[4pt]
&= -\eta(X,Y)\ ,
\end{align*}
for all $X, Y \in \mathsf{\Gamma}(TM),$ where we used the
compatibility condition \eqref{compcon} for $(K,\eta)$ in the fourth
equality, and the isometry conditions $(\bar{f}^{-1})^*\eta=\eta$ and
$\bar{f}^* \eta= \eta$ in the third and sixth equalities respectively. This shows that $(M, K_\vartheta, \eta)$ is an almost para-Hermitian manifold.

We finally show that $\omega_\vartheta= \bar{f}^* \omega.$ The
fundamental 2-form of $(K_\vartheta, \eta)$ is given by
$\omega_\vartheta(X,Y)=\eta(K_\vartheta(X),Y),$ for all $X,Y \in \mathsf{\Gamma}(TM).$ Then
\begin{align*}
(\bar{f}^*\omega)(X,Y)&= \omega\big(\bar{f}(X), \bar{f}(Y)\big) \\[4pt]
&= \eta\big(K(\bar{f}(X)), \bar{f}(Y)\big) \\[4pt]
&=\big((\bar{f}^{-1})^* \eta\big)\big(K(\bar{f}(X)), \bar{f}(Y)\big) \\[4pt]
&= \eta\big(\bar{f}^{-1}\,K(\bar{f}(X)),Y\big) \\[4pt]
&= \eta\big(K_\vartheta(X),Y\big) \\[4pt] 
&= \omega_\vartheta(X,Y)\ ,
\end{align*}
for all $X, Y \in \mathsf{\Gamma}(TM),$ where in the third equality we
used $\eta=(\bar{f}^{-1})^*\eta.$
\end{proof}
\end{proposition}

\begin{corollary}
The projectors $\Pi_{\pm}=\frac12\,(\unit \pm K)$ associated with $K$
transform under $\vartheta=(f, \bar{f}\,)\in {\sf Aut}(TM)$ to 
$$
(\Pi_\vartheta)_\pm=\bar{f}^{-1}\circ \Pi_\pm\circ \bar{f} \ .
$$ 
\end{corollary}

\begin{corollary}
An element $\vartheta=(f, f_*)\in {\sf O}(d,d)(M)$ preserves the
exterior derivative of the fundamental 2-form, and hence it maps an almost
para-K\"ahler structure $(K, \eta, \omega)$ into another almost
para-K\"ahler structure $(K_\vartheta, \eta, \omega_\vartheta)$ with $\omega_\vartheta= f^*\omega$.
\end{corollary}

A similar statement holds for the pushforward of an almost para-Hermitian
structure, and we have
\begin{proposition}
Let $(M, K, \eta)$ be an almost para-Hermitian manifold with
fundamental 2-form $\omega,$ and let $\vartheta=(f,\bar{f}\,)\in {\sf O}(d,d)(M).$ Then the {pushforward} of $K$ by $\vartheta,$ 
$K^{\vartheta}= \bar{f}\circ K\circ \bar{f}^{-1},$ and $\eta$ form an almost para-Hermitian structure $(K^{\vartheta}, \eta)$ on $M$ with fundamental 2-form $\omega^\vartheta=(\bar{f}^{-1})^* \omega.$
\end{proposition}
\begin{proof}
Replace $\vartheta$ with $\vartheta^{-1}$ in Proposition~\ref{push},
and use $K^\vartheta=K_{\vartheta^{-1}}$.
\end{proof}

An automorphism $\vartheta=(f, \bar{f}\,) \in {\sf Aut}(TM)$ does not
necessarily preserve the splitting $TM= L_+ \oplus L_-$ induced by the
almost para-complex structure $K.$ Thus $K_\vartheta$ can have
different eigenbundles from $K.$ Furthermore, if $K$ is a
(Frobenius integrable) para-complex structure, so that $N_K=0,$ then an arbitrary element $\vartheta\in {\sf Aut}(TM)$ need not preserve the integrability, i.e. $N_{K_\vartheta}\neq 0.$
This also means that such transformations neither generally preserve the (Frobenius)
integrability of the eigenbundles, nor the closure of the
fundamental 2-form. In this sense, the choice of polarization contains
all the information about the background fluxes and the physical 
spacetime: in~\cite{Svoboda2018,SzMar} it is shown that the
generalized fluxes
appear as obstructions to weak integrability with respect to a
reference para-K\"ahler structure. We will return to this point in
Section~\ref{sec:Btransformations}.

Nevertheless, there is a simple case in which we can say something
concrete about the integrability of the transformed (almost)
para-complex structure, as asserted through
\begin{proposition}
Let $(M, K)$ be an almost para-complex manifold and $f \in {\sf
  Diff}(M).$ Then the tangent bundle automorphism induced by the
differential of $f,$ $\vartheta=(f, f_*) \in {\sf Aut}(TM),$ maps the
Nijenhuis tensor $N_K$ of $K$ to the Nijenhuis tensor
$N_{K_\vartheta}=f_* N_K$ of the pulled back almost para-complex structure $K_\vartheta.$
\begin{proof}
The crux of the proof is the naturality of the Lie bracket of vector fields,
i.e.~$f_*[X,Y]=[f_*(X),f_*(Y)],$ for all $ f \in {\sf Diff}(M)$ and $
X, Y \in \mathsf{\Gamma}(TM).$ It is also easy to show that the only
Lie bracket-preserving tangent bundle automorphisms are given
by $(f, f_*),$ with $f \in{ \sf Diff}(M)$ (see e.g.~\cite{gualtieri:tesi}).

Since $K_\vartheta= f_*^{-1}\circ K\circ f_*, $ the Nijenhuis tensor of $K_\vartheta$ reads
\begin{align*}
N_{K_\vartheta}(X,Y)&=[X,Y]+\big[f_*^{-1}\,K\big(f_*(X)\big),f_*^{-1}\,K\big(f_*(Y)\big)\big] \\
&\qquad -f_*^{-1}\,K\,f_*\big(\big[f_*^{-1}\,K(f_*(X)),Y\big] + \big[X,f_*^{-1\,}K(f_*(Y))\big]\big) \\[4pt]
&= f_*^{-1} \big( f_*[X,Y]+\big[K(f_*(X)),K(f_*(Y))\big] \\ 
&\qquad -K([K(f_*(X)),f_*(Y)] + [f_*(X),K(f_*(Y))])\big) \\[4pt]
&= f_*^{-1}\big(N_K(f_*(X), f_*(Y))\big)\ ,
\end{align*}
for all $ X, Y \in \mathsf{\Gamma}(TM),$ where in each step we used the naturality of the Lie bracket.
\end{proof} 
\end{proposition}

\begin{corollary}
Let $(M, K)$ be a para-complex manifold, i.e. $N_K=0.$ Then a tangent bundle automorphism $\vartheta=(f,f_*)$ maps $K$ into another para-complex structure $K_\vartheta$ with $N_{K_\vartheta}=0.$  
\end{corollary}
This proof relies on the naturality of the Lie bracket
of vector fields under the pushforward by any diffeomorphism of $M.$
This property holds only for pushforwards and not for generic elements
$\vartheta \in {\sf Aut}(TM),$ so it is not possible to find any
general relation between the Nijenhuis tensors of an almost
para-complex structure $K$, and the associated pullback $K_\vartheta$ and
pushforward $K^\vartheta$ under $\vartheta.$ Hence the (lack of)
integrability of an almost para-complex structure is not generally
preserved by an automorphism of the tangent bundle $TM.$  

\begin{remark} \label{rem:SOdd}
We denote by ${\sf SO}(d,d)(M)$ the Lie subgroup of ${\sf O}(d,d)(M)$
which also preserves the canonical orientation of $M$ provided by its
fundamental 2-form $\omega$; its Lie algebra $\mathfrak{so}(d,d)(M)$
consists of tangent bundle endomorphisms $\tau\in{\sf End}(TM)$ for
which 
$$
\eta\big(\tau (X),Y\big)=-\eta\big(X,\tau (Y)\big)
$$ 
for all $X,Y\in\mathsf{\Gamma}(TM)$. Any element $\tau\in\mathfrak{so}(d,d)(M)$ can be decomposed with respect to the splitting $TM=L_+\oplus L_-$ as
\be\nonumber
\tau = \bigg(
\begin{matrix}
A & B_- \\
B_+ & -A^{\rm t}
\end{matrix} \bigg) \ ,
\ee
where $A\in{\sf End}(L_+)$ with transpose $A^{\rm t}\in{\sf End}(L_-)$
defined via $\eta(A(X),Y)=\eta(X,A^{\rm t}(Y))$, while
$B_+:\mathsf{\Gamma}(L_+)\to \mathsf{\Gamma}(L_-)$ and
$B_-:\mathsf{\Gamma}(L_-)\to \mathsf{\Gamma}(L_+)$ are skew morphisms
in the sense that $\eta(B_\pm(X),Y)=-\eta(X,B_\pm(Y))$. By identifying
$L_-$ with $L_+^*$ using the split signature metric $\eta$, we can regard
$B_+$ as a 2-form in $\mathsf{\Gamma}(\midwedge^2 L_+^*)$ and $B_-$ as
a bivector in $ \mathsf{\Gamma}(\midwedge^2 L_+)$, so that as a vector space
\be\nonumber
\mathfrak{so}(d,d)(M) = {\sf End}(L_+) \oplus \, \mathsf{\Gamma}(\midwedge^2 L_+^*) \oplus \, \mathsf{\Gamma}(\midwedge^2 L_+) \ .
\ee
In Section~\ref{sec:Btransformations} we shall discuss the important
class of ${\sf O}(d,d)(M)$-transformations generated by the last two
summands, which are called $B$-transformations.
\end{remark}

\medskip

\subsection{${\sf O}(d,d)(M)$-Transformations of Born Geometry}\label{sec:Borntransf} ~\\[5pt]
Applying the transformations of Section~\ref{int:Odd} to Born geometry can be described by starting from the pullback of a Born structure by an automorphism $\vartheta=(f,\bar{f}\,)\in {\sf O}(d,d)(M).$
First we show that an ${\sf O}(d,d)(M)$-transformation of a
generalized metric $V_+\subset TM$ gives another generalized metric,
as asserted by
\begin{proposition}
Let $V_+\subset TM$ be a generalized metric on an almost para-Hermitian manifold
$(M,K, \eta)$, and let $\vartheta=(f, \bar{f}\,)\in {\sf O}(d,d)(M).$
Then the pullback of $V_+$ given by
$$
{V_\vartheta}_+=\bar{f}(V_+)= \big\{ X'=\bar{f}(X)\ \big|\ X\in V_+
\big\}
$$
is a generalized metric on $(M, K, \eta).$
\begin{proof}
The proof is straightforward:
$$
\eta(X',Y')=\eta \big(\bar{f}(X), \bar{f}(Y)\big)=\eta(X,Y)\geq 0 \ ,
$$
for all $ X' , Y'\in {V_\vartheta}_+,$ since $X, Y\in V_+.$ 
\end{proof} 
\end{proposition}
The same argument also applies to any vector bundle $E\to \cQ$ endowed
with a  metric $\eta$ of signature $(d,d)$, and any automorphism in ${\sf
  O}(d,d)(E)\subset {\sf Aut}(E)$ which preserves~$\eta$.

We can now characterize the generalized T-duality transformations of a
Born structure through
\begin{proposition}
Let $(K, \eta, \cH)$ be a Born geometry on a manifold $M$ with
fundamental 2-form $\omega$, and let $\vartheta=(f, \bar{f}\,)\in {\sf
  O}(d,d)(M).$ Then $(K_\vartheta, \eta , \cH_\vartheta)=(
\bar{f}\circ K\circ\bar f{}^{-1}, \eta,  \bar{f}^* \cH)$ is a Born geometry on $M$ with fundamental 2-form $\omega_\vartheta= \bar{f}^* \omega.$
\begin{proof}
We have already shown that $(K_\vartheta, \eta)$ is an almost
para-Hermitian structure on $M$ with fundamental 2-form
$\omega_\vartheta$ in Proposition \ref{push}. It remains to prove that $\cH_\vartheta=\bar{f}^*\cH$ satisfies the compatibility conditions 
\be\nonumber
\eta^{-1}{}^\sharp\, \circ \,\cH_\vartheta{}^\flat = \cH_\vartheta^{-1}{}^\sharp\,\circ\, \eta^\flat \qquad \mbox{and} \qquad \omega_\vartheta^{-1}{}^\sharp\,\circ \,\mathcal{H}_\vartheta{}^\flat = -\mathcal{H}_\vartheta^{-1}{}^\sharp\,\circ \, \omega_\vartheta{}^\flat \ .
\ee
We first check that the inverse of $\cH_\vartheta$ is given by 
$$
\cH^{-1}_\vartheta=\bar{f}_* \cH^{-1} \ ,
$$  
where 
$$
\bar{f}_* \cH^{-1}(\nu, \xi)=
\cH^{-1}\big((\bar{f}^{-1})^*(\nu),(\bar{f}^{-1})^*(\xi) \big)\  , 
$$
for all $ \nu, \xi \in \mathsf{\Gamma}(T^*M).$ We also need the expression
$$
\cH_\vartheta{}^\flat(X)(Y)=\cH_\vartheta(X,Y)= \cH\big(\bar{f}(X),\bar{f}(Y)\big)=\big(\bar{f}^*
\cH^\flat(\bar{f}(X))\big)(Y)\ ,
$$  
for all $X, Y \in \mathsf{\Gamma}(TM),$ so that $
\cH_\vartheta{}^\flat=\bar{f}^* \circ\cH^\flat\circ\bar f.$
Then
\begin{align*}
\big(\bar{f}_* \cH^{-1}{}^\sharp \circ \cH_\vartheta{}^\flat(X)\big)(\nu)&= \bar{f}_* \cH^{-1}\big(\cH_\vartheta{}^\flat(X), \nu\big) \\[4pt]
&= \bar{f}_* \cH^{-1}\big(\bar{f}^* \cH{}^\flat(\bar{f}(X)), \nu\big) \\[4pt]
&= \cH^{-1}\big(\cH^\flat(\bar{f}(X)), (\bar{f}^{-1})^*\nu\big) \\[4pt]
&= \bar{f}(X)\big((\bar{f}^{-1})^*\nu\big) \\[4pt] 
&= X(\nu)
\end{align*}
for all $ X \in \mathsf{\Gamma}(TM)$ and $\nu \in
\mathsf{\Gamma}(T^*M),$ which shows $\cH^{-1}_\vartheta=\bar{f}_*
\cH^{-1}.$  

We are now ready to prove the first compatibility condition
between $\eta$ and $\cH_\vartheta.$ Since $\eta^{-1}{}^\sharp \circ
\cH_\vartheta{}^\flat \in {\sf End}(TM),$ we compute
\begin{align*}
\big(\eta^{-1}{}^\sharp \circ \cH_\vartheta{}^\flat(X)\big)(\nu)&=\eta^{-1}\big(\cH_\vartheta{}^\flat(X), \nu\big) \\[4pt]
&= \eta^{-1}\big(\bar{f}^* \cH^\flat(\bar{f}(X)), \nu\big) \\[4pt]
&= \bar{f}^{-1}_* \eta^{-1} \big(\cH^\flat(\bar{f}(X)),
  (\bar{f}^{-1})^*\nu \big) \\[4pt]
&= \eta^{-1} \big(\cH^\flat(\bar{f}(X)), (\bar{f}^{-1})^*\nu \big)\\[4pt]
&= \cH^{-1} \big(\eta^\flat(\bar{f}(X)), (\bar{f}^{-1})^*\nu \big)\\[4pt]
&= \big(\bar{f}_* \cH^{-1}\big)\big(\eta^\flat(X), \nu\big)\\[4pt]
&= \big(\cH_\vartheta^{-1}{}^\sharp \circ \eta^\flat(X)\big)(\nu)
\end{align*}
for all $X \in \mathsf{\Gamma}(TM)$ and $\nu \in
\mathsf{\Gamma}(T^*M),$ where in the fifth equality we used the
compatibility condition $\eta^{-1}{}^\sharp\, \circ\,\cH^\flat=
\cH^{-1}{}^\sharp\,\circ\,\eta^\flat$ for the original Born geometry. This is
basically a more complicated way of proving that a generalized metric is
mapped into a generalized metric under $\vartheta \in {\sf
  O}(d,d)(M).$ It is useful to also prove it in this way, which is more
in the spirit of the original definition given in~\cite{Freidel2019}, because checking the second compatibility
condition between $\omega_\vartheta$ and $\cH_\vartheta$ is then straightforward: The required relations are 
$$
\omega_\vartheta^{-1}= \bar{f}_* \omega^{-1} \qquad \mbox{and} \qquad
\omega_\vartheta{}^\flat= \bar{f}^* \circ \omega^\flat \circ \bar{f} \ .
$$ 
The proof then follows
exactly the same steps taken for the first compatibility condition. 
\end{proof}
\end{proposition}

\begin{remark}
Except for the issues concerning integrability, all of our discussion
thus far concerning generalized T-duality transformations carries through as well for arbitrary
para-Hermitian vector bundles $(E, K, \eta)$ on a manifold $\cQ$. 
\end{remark}

\medskip

\subsection{$B$-Transformations}\label{sec:Btransformations} ~\\[5pt]
We shall now define special isometries relating two different almost
para-Hermitian structures on the same manifold $M$. We will focus on a
specific class of transformations that cover the identity. Such
transformations give a nice example in which the general description
introduced in Section~\ref{sec:Borntransf} can be explicitly worked
out. In general, automorphisms covering the identity play a
fundamental role in the description of automorphisms of principal
bundles, since they form the subgroup of gauge transformations, and a
similar notion can be introduced for Poisson structures. In this description we will see similarities with the transformations ocurring in the context of exact Courant algebroids \cite{Severa2002, gualtieri:tesi}. It was shown in \cite{Svoboda2018, SzMar} that geometric and non-geometric fluxes appear in this discussion as obstructions to integrability with respect to the D-bracket.

We first introduce the notion of a $B$-transformation for an almost
para-Hermitian manifold $(M,K,\eta)$. Let us fix the splitting $TM=L_+ \oplus L_-$ induced by $K.$
In this polarization, a vector field $X\in\mathsf{\Gamma}(TM)$ decomposes as 
$$
X= X_+ + X_- \qquad \mbox{with} \quad X_+ \in \mathsf{\Gamma}(L_+) \ ,
\ X_-\in \mathsf{\Gamma}(L_-)\ .
$$
\theoremstyle{definition}
 \begin{definition}\label{def:B+}
Let $(M, K, \eta)$ be an almost para-Hermitian manifold. A
$B_+$-\emph{transforma{-}tion} is an isometry $e^{B_+}: TM \rightarrow
TM$ of $\eta$ covering the identity which is given by
$$
e^{B_+}(X)= X_+ + B_+(X_+)+ X_- \ \in \ \mathsf{\Gamma}(TM)\ , 
$$
for all $X \in \mathsf{\Gamma}(TM)$, or in matrix notation
\be \label{btra}
e^{B_+}= 
\bigg(\begin{matrix}
\mathds{1} & 0 \\
B_+ & \mathds{1}
\end{matrix}\bigg)
: TM \longrightarrow TM \qquad \mbox{with} \quad \big(\unit, e^{B_+}\big) \in {\sf O}(d,d)(M) \ ,
\ee
in the chosen splitting induced by $K$, where $B_+ : L_+
\rightarrow L_-$ is a smooth skew map in the sense
that it satisfies 
$$
\eta\big(B_+(X),Y\big)=- \eta\big(X, B_+(Y)\big) \ ,
$$ 
for all $ X, Y \in \mathsf{\Gamma}(TM).$ 
\end{definition}
The endomorphism $B_+$ defines both a 2-form $b_+$ and a bivector $\beta_-$ by 
$$
\eta\big(B_+(X),Y\big)=b_+(X,Y)= \beta_-\big(\eta(X), \eta(Y)\big) \ .
$$
The 2-form $b_+$ is of type $(+2,-0)$ while the bivector $\beta_-$ is of type $(+0,-2)$ with respect to $K.$ This will be relevant to understanding how the fundamental 2-form $\omega$ changes under a $B_+$-transformation.

The inverse map is given by $e^{-B_+}:TM \rightarrow TM.$ The induced map on 1-forms $\nu \in \mathsf{\Gamma}(T^*M)$ is given by 
$$
\big((e^{B_+})^*\nu\big) (X)= \nu \big(e^{B_+}(X)\big)\ ,
$$ 
for all $ X \in \mathsf{\Gamma}(TM).$
The splitting of the tangent bundle $TM=L_+ \oplus L_-$ induces a
splitting of the cotangent bundle $T^*M=L_+^* \oplus L_-^*,$ thus a
1-form $\nu \in \mathsf{\Gamma}(T^*M)$ decomposes as
$$
\nu= \nu_+ + \nu_- \qquad \mbox{with} \quad \nu_+ \in
\mathsf{\Gamma}(L_+^*) \ , \ \nu_- \in \mathsf{\Gamma}(L_-^*)\ .
$$
Then the induced $B_+$-transformation on 1-forms reads
$$
\big((e^{B_+})^*\nu\big) (X)= \nu_+(X_+)+ \nu_-\big(B_+(X_+)\big)+
\nu_-(X_-) \ .
$$
Since $B_+$ is a map from $L_+$ to
$L_-,$ its transpose is a map $B^{\rm t}_+:
L_-^* \rightarrow L_+^*.$ This means
that the function $\nu_- (B_+(X_+))$ can also be written as $\nu_-
(B_+(X_+))= (B^{\rm t}_+(\nu_-))(X_+),$ thus the $B_+$-transformation of a
1-form $\nu$ can be written as
$$
\big((e^{B_+})^*\nu\big) (X)= \big(\nu_+ + B^{\rm t}_+(\nu_-)\big)(X_+) +
\nu_-(X_-) \ .
$$
This implies that $(e^{B_+})^*$ takes the same matrix form \eqref{btra} as $(e^{B_+})^{\rm t}$ in the chosen polarization.

A $B_+$-transformation induces two almost
para-complex structures from the almost para-Hermitian manifold
$(M,K,\eta),$ as we saw in Section~\ref{sec:Borntransf}.  Let us
choose the splitting $TM=L_+\oplus L_-$ induced by $K,$ with the
corresponding decompositions $X=X_+ + X_- \in \mathsf{\Gamma}(TM)$ and
$\nu=\nu_+ + \nu_- \in \mathsf{\Gamma}(T^*M).$ Recall that $K\in {\sf
  Aut}_\unit(TM),$ in this polarization, can be written as $K=
\unit_{L_+}- \unit_{L_-}.$ Then the {pullback} of $K$  by a $B_+$-transformation is
\begin{align*}
K_{B_+}(X)( \nu) &= K\big(e^{B_+}(X)\big)\big( (e^{-B_+})^*(\nu)\big) \\[4pt]
&= K\big(X_+ + B_+(X_+)+ X_- \big)\big(\nu_+ - B_+(\nu_-)+ \nu_-\big) \\[4pt]
&= X_+(\nu_+)- X_-(\nu_-)- 2\, B_+(X_+)(\nu_-) \\[4pt]
&=(K-2\,B_+)(X)( \nu)\ . 
\end{align*}
Hence $K_{B_+}=K- 2\,B_+$, which can be cast in the form 
\be \nonumber
K_{B_+}=e^{-B_+}\circ K\circ e^{B_+}=
\begin{pmatrix}
\unit & 0 \\
-2\,B_+ & -\unit
\end{pmatrix} \ .
\ee
We then have $K_{B_+}^2= \unit,$ since $B_+(K(X))=-K(B_+(X))$ and
$B_+(B_+(X))=0,$ for all $X\in \mathsf{\Gamma}(TM),$ and $K_{B_+}$ satisfies the compatibility condition  $\eta(K_{B_+}(X),K_{B_+}(Y))=-\eta(X,Y)$ with $\eta$  because of the skew-symmetry property of $B_+.$
Thus $(K_{B_+}, \eta)$ is an almost para-Hermitian structure, as expected.   
   
The fundamental 2-form $\omega_{B_+}$ of $(K_{B_+}, \eta)$ is given by  $\omega_{B_+}= (e^{B_+})^*\omega,$ as shown in Section~\ref{int:Odd}. In this case we obtain
\begin{align*}
\omega_{B_+}(X,Y)&=\omega\big(e^{B_+}(X), e^{B_+}(Y)\big) \\[4pt]
&= \omega(X_-, Y_+)+ \omega(X_+,Y_-)+ \omega\big(B_+(X_+), Y_+\big) +
  \omega\big(X_+, B_+(Y_+)\big) \\[4pt]
&= \omega(X_-, Y_+)+ \omega(X_+,Y_-)- \eta\big(B_+(X_+), Y_+\big)+
  \eta\big(X_+, B_+(Y_+)\big) \\[4pt]
&= \omega(X_-, Y_+)+ \omega(X_+,Y_-)- 2\,\eta\big(B_+(X_+), Y_+\big) \\[4pt]
&= (\omega- 2\,b_+)(X,Y) \ ,
\end{align*}
for all $X, Y \in \mathsf{\Gamma}(TM),$ where in the second equality
we used the isotropy of $L_+$ and $L_-$ with respect to $\omega$, and
in the fourth equality we used the skew-symmetry property of $B_+.$
This shows that $\omega_{B_+}= \omega - 2\,b_+$; the same result can
also be obtained by computing $\omega_{B_+}(X,Y)= \eta(K_{B_+}(X),Y).$
As a consequence, a $B_+$-transformation does not generally preserve the closure of the fundamental 2-form.
     
In a similar fashion, the {pushforward} of the almost para-complex
structure $K$ is given by 
$$
K^{B_+}(X)( \nu)=K\big(e^{-B_+}(X)\big)\big( (e^{B_+})^*(\nu)\big)\ ,
$$ 
for all $ X \in \mathsf{\Gamma}(TM)$ and $\nu \in \mathsf{\Gamma}(T^*M).$ Then with a similar computation to the case of the pullback we obtain 
$$ 
K^{B_+}= e^{B_+}\circ K\circ e^{-B_+}= K+ 2\,B_+\ ,
$$
or in matrix notation 
\be
K^{B_+}=
\bigg(\begin{matrix}
\mathds{1} & 0 \\
2\,B_+ & -\mathds{1}
\end{matrix}\bigg)
\ , \label{matrK}
\ee
with respect to the splitting $TM=L_+ \oplus L_-$. One easily has $(K^{B_+})^2=\mathds{1}$, while the skew-symmetry
property of $B_+$ implies the compatibility condition
$\eta(K^{B_+}(X),K^{B_+}(Y))=-\eta(X,Y)$. In the conventions
of~\cite{Svoboda2018}, the definition of a $B_+$-transformation is
given by the pushforward of $K$ by $e^{B_+}$; we will adhere to this convention unless otherwise stated. 
The fundamental 2-form is then given by
$$
\omega^{B_+}= \big(e^{-B_+}\big)^*\omega = \omega +
2\,b_+\ ,
$$
so that such transformations also may not preserve the closure of the
fundamental 2-form. A completely analogous discussion can be carried
out for a \emph{$B_-$-transformation}, defined by a skew-symmetric map
$B_-: L_- \rightarrow L_+$ in the
sense described in Definition~\ref{def:B+}, by interchanging the roles of the eigenbundles $L_+$ and $L_-$.

The main effect of a $B_+$-transformation is that the splitting
$TM=L_+\oplus L_-$ changes, i.e.~$e^{B_+}$ maps the polarization
$L_+\oplus L_-$ to a new polarization $L^{B_+}_+\oplus L^{B_+}_-,$ which implies that the potential Frobenius integrability of the original splitting may not be preserved in its image under $e^{B_+}$. The transformed projections are given by
\be
\Pi_+^{B_+}= \tfrac{1}{2}\,\big(\mathds{1}+K^{B_+}\big)=
\bigg(\begin{matrix}
\mathds{1} & 0 \\
B_+ & 0
\end{matrix}\bigg)
\qquad \mbox{and} \qquad
\Pi_-^{B_+}= \tfrac{1}{2}\,\big(\mathds{1}-K^{B_+}\big)=
\bigg(\begin{matrix}
0 & 0 \\
-B_+ & \mathds{1}
\end{matrix}\bigg)
\ . \nonumber
\ee
Hence decomposing any vector field with respect to the splitting associated with $K$ as $X=X_++X_- \in  \mathsf{\Gamma}(TM),$ where $X_\pm \in \mathsf{\Gamma}(L_\pm)$, the new distributions are obtained by using the transformed projections to get
$$
\Pi^{B_+}_+(X)= X_+ + B_+(X_+) \qquad \mbox{and} \qquad
\Pi^{B_+}_-(X)= X_- - B_+(X_+) \ ,
$$
where $\Pi^{B_+}_-(X) \in \mathsf{\Gamma}(L_-)$ since $B_+$ maps
$L_+$ to $L_-$, thus $L^{B_+}_-=L_-.$ On the other hand, the same reasoning applied to $\Pi^{B_+}_+(X)$ shows that it is not an element of $\mathsf{\Gamma}(L_+),$ thus $L^{B_+}_+\neq L_+.$
Therefore only the $-1$-eigenbundle is preserved by a
$B_+$-transformation, while the $+1$-eigenbundle changes; in
particular, if $L_+$ is integrable then $L^{B_+}_+$ is generally
non-integrable.

\begin{remark}\label{rem:B+}
More generally, if $(E,\eta,L_-)$ is an even rank vector bundle
endowed with a split signature metric and a choice of maximally
isotropic sub-bundle, as in Section~\ref{sec:paravector},  then the
maximally isotropic splittings of the short exact sequence
\eqref{eq:EL-} are precisely (up to isomorphism) the
$B_+$-transformations, which preserve $L_-$.
\end{remark}

To compare two different almost para-Hermitian structures on
the same manifold $M$, a weaker notion of integrability can be introduced. The main difference from the usual notion of Frobenius integrability is the replacement of the Lie bracket of vector fields with the D-bracket. This is discussed in \cite{SzMar, Svoboda2018, Freidel2017, Mori2019,Hu2019}. Changes of polarization generally induce flux
deformations of the almost para-Hermitian structure given by Lie
algebroid 3-forms, and hence may
spoil integrability of the eigenbundles (either Frobenius or with
respect to the D-bracket associated to the original para-Hermitian structure).

We conclude this section by discussing the $B$-transformations of a
compatible generalized metric of a Born geometry. A compatible
generalized metric $\cH$ of the almost para-Hermitian structure $(K,
\eta)$ transforms under a $B_+$-transformation to the compatible
generalized metric $\cH_{B_+}$ of the pullback almost
para-Hermitian structure $(K_{B_+}, \eta)$ on $M.$ Recalling that
$\cH$ takes the diagonal form \eqref{eq:diaghermmetric}, we then have 
\begin{align*}
\cH_{B_+}(X,Y)&=\big(e^{B_+}\big)^*\cH(X,Y) \\[4pt]
&= \cH\big(e^{B_+}(X), e^{B_+}(Y)\big) \\[4pt]
&= g_+(X_+,Y_+)+g_-\big(B_+(X_+),B_+(Y_+)\big)+
  g_-\big(B_+(X_+),Y_-\big) \\
& \qquad + g_-\big(X_-, B_+(Y_+)\big)+ g_-(X_-,Y_-) \ ,
\end{align*}
for all $X, Y \in \mathsf{\Gamma}(TM).$ Similarly, the
$B_+$-transformation of a compatible generalized metric with respect
to the pushforward of an almost para-Hermitian structure $(K^{B_+}, \eta)$ takes the form
\begin{align*}
\cH^{B_+}(X,Y)&=\big(e^{-B_+}\big)^*\cH(X,Y) \\[4pt]
&=\cH\big(e^{-B_+}(X), e^{-B_+}(Y)\big) \\[4pt]
&=g_+(X_+,Y_+)+g_-\big(B_+(X_+),B_+(Y_+)\big)-
  g_-\big(B_+(X_+),Y_-\big) \\
& \qquad - g_-\big(X_-, B_+(Y_+)\big)+ g_-(X_-,Y_-) \ ,
\end{align*}
for all $ X, Y \in \mathsf{\Gamma}(TM).$ This is exactly the same
expression as \eqref{gemeju} from the proof of
Proposition~\ref{gbparaherm} upon identifying $\gamma_b= B_+,$ and
we have thereby shown
\begin{proposition}
A generalized metric $V_+ \subset TM$ on an almost para-Hermitian
manifold $(M, K, \eta)$ corresponds to the choice of a Born geometry
$(K,\eta, \cH)$ and a $B_+$-transformation.
\end{proposition}

\section{Worldsheet Theory for Para-Hermitian Manifolds} \label{BSM}

In this section we will relate $\sfO(d,d)(M)$-transformations to (generalized)
T-duality from the perspective of closed string sigma-models whose target spaces are
Born manifolds. Special cases
where at least partial integrability of the para-Hermitian structure
is preserved will then allow us to derive a sigma-model on the leaf
space of the foliated para-Hermitian manifold. These sigma-models are thought of as emerging from the
Born geometry after a quotient along a foliation of a
para-Hermitian manifold in a given polarization. A T-duality
transformation will then be an $\sfO(d,d)(M)$-transformation. 

\medskip

\subsection{Born Sigma-Models} ~\\[5pt]
Our aim in the following is to define worldsheet sigma-models for para-Hermitian
manifolds using, whenever it exists, a compatible Born geometry (cf.
Section~\ref{intropara}). This will allow us to see how generalized
T-duality transformations relate different sigma-models and how, in turn,
they give a relation between sigma-models obtained by reduction
on foliated para-Hermitian manifolds. We begin by defining the
sigma-models of interest in this paper.  
\begin{definition}
A \emph{Born sigma-model} is a harmonic map 
$$
\phi: (\Sigma, h)
\longrightarrow (M,\cH) \ ,
$$ 
where $\Sigma$ is a closed oriented surface
endowed with  a (pseudo-)Riemannian metric $h$ and $(M,K, \eta, \cH)$
is an almost para-Hermitian manifold with compatible generalized metric $\cH$ in the sense of Section~\ref{sec:Born}.
\end{definition}
In other words, $\phi \in {C}^\infty(\Sigma, M)$ is the smooth map
minimizing the functional\footnote{Here and throughout upper case
  Latin indices run over the target space directions
  $I,J,\dots=1,\dots,2d$, and repeated upper and lower indices are
  implicitly summed over.}
\be \label{sigmaaction}
\cS_0[\phi]=\frac{1}{4}\, \int_{\Sigma}\, \bar\cH_{IJ} \, \de \phi^I \wedge \star\, \de \phi^J\ ,
\ee
where the Hodge duality operator $\star$ is defined with respect to
the worldsheet metric $h$, which for definiteness we take to be
Lorentzian so that $\star^2=\mathds{1}$. Here and in the following a
bar on a field on $M$ denotes its pullback to the worldsheet $\Sigma$
by the map $\phi$. The action \eqref{sigmaaction} is
invariant under Lorentz transformations of the worldsheet and rigid
${\sf O}(d,d)$-transformations of the almost para-Hermitian target
space. 

This is the Dirichlet
functional obtained by endowing the space of maps $\de \phi:
T\Sigma \rightarrow \phi^*TM$ with a norm defined by
$\cH$, regarded as a metric on the vector space of sections of the pullback 
$\phi^*TM$ of the
tangent bundle $TM$ to $\Sigma$ by $\phi,$ and the inverse metric $h^{-1}$ on $T^*\Sigma$; this gives
a well-defined norm $\|\,\underline{\de\phi}\,\|_{h,\cH}$ for sections $\underline{\de \phi} \in
\mathsf{\Gamma}(T^*\Sigma \otimes \phi^*TM)$ which allows us to write
the action functional as
\be \label{sigmanorm}
\cS_0[\phi]=\frac{1}{4}\,\int_{\Sigma}\, h^{\alpha \beta}\,
\bar\cH_{IJ} \,
\del_{\alpha} \phi^I\, \del_{\beta} \phi^J\ \de \mu(h) =:
\frac{1}{4}\,\int_\Sigma\, \|\,\underline{\de \phi}\, \|_{h,\cH}^2\ \de \mu(h)\ ,
\ee 
where the Greek indices run over local coordinates $(\sigma^0,
\sigma^1)$ on $\Sigma,$ with
$\partial_\alpha$ the derivative with respect to $\sigma^\alpha$, and
$$
\de \mu(h) = \star\,1
$$ 
is the area measure
induced by $h$. 

We will also consider a topological term of the form 
\be \label{sigmatop}
\cS_{\rm top}[\phi]=\frac{1}{2}\, \int_{\Sigma}\, \bar\omega\ ,
\ee
where $\omega$ is the fundamental 2-form of the Born manifold $M.$
For curved worldsheets, the general form of a two-dimensional sigma-model also involves a
Fradkin-Tseytlin term
\be\nonumber
\cS_{\Psi}[\phi]=\frac1{2\pi}\, \int_{\Sigma}\,R^{(2)}(h) \, \bar\Psi \ \de
\mu(h) \ , 
\ee
where the smooth function $\Psi: M \rightarrow \mathbb{R}$ is a scalar
dilaton field and $R^{(2)}(h)$ is the scalar curvature of the metric
$h$ on $\Sigma$; since the metric $h$ is conformally equivalent to a
flat metric on $\Sigma$, this term can be (classically) set to $0$ by
a conformal transformation of the worldsheet and will not be
considered any further in the ensuing analysis.

We will usually denote by $\mathcal{S}(\cH, \omega)$ a Born sigma-model given by the sum of \eqref{sigmaaction} and~\eqref{sigmatop}:
\be \label{actcomplete}
 \cS[\phi]= \cS_0[\phi]+ \cS_{\rm top}[\phi] \ .
\ee
The notation stresses that the defining data for a Born sigma-model are given by
the fundamental geometric structures of a Born manifold. Written in this way,
the Born sigma-model is a direct generalization of the sigma-models
for doubled torus fibrations that were introduced
in~\cite{Hull2007}. 

On the other hand, we
denote by $S(g, b)$ a non-linear sigma-model with target space any Riemannian manifold $(\cQ, g)$ with Kalb-Ramond field $b \in \Omega^2(\cQ)$ which is \emph{not} a Born manifold:
$$
S[\phi]=\frac{1}{2}\,\int_\Sigma\, \bar g_{ij}\, \de \phi^i \wedge
\star\,  \de \phi^j +  \int_\Sigma\, \bar b\ ,
$$
where here $\phi$ is a map from $(\Sigma,h)$ to $(\cQ,g).$
Unlike the fundamental 2-form $\omega$ on $M$, the Kalb-Ramond field $b$ is not always a globally defined 2-form on
$\cQ$. Generally $b$ is only locally defined because the sigma-model
is characterized by a topological Wess-Zumino term given by a closed 3-form
$H$ with integer cohomology class $[H] \in \mathsf{H}^3(\cQ, \IZ),$
which is geometrically the Dixmier-Douady class of a gerbe on $\cQ$
with connection $b$ so that
$H=\de b$ only locally. In this case, we introduce a
three-dimensional manifold $V$ with boundary $\partial V= \Sigma,$ and
extend the maps $\phi$ to $V$. We may then write the action functional of the sigma-model as
$$
S[\phi]=\frac{1}{2}\,\int_\Sigma\, \bar g_{ij}\, \de \phi^i \wedge
\star\,  \de \phi^j + \int_V\, \bar H\ ,
$$
where $\bar{H}$ is the pullback of the 3-form $H$ to
$V$. In the quantum theory, the contribution of the $B$-field amplitude
$\exp(2\pi\,\mathrm{i}\,\int_V\,\bar H)$ to the functional integral is well-defined,
i.e. independent of the choice of three-dimensional manifold $V$ bounded by
$\Sigma$, by virtue of the assumption that $H$ has integer periods.

It will prove convenient to work with a local form for the
Born sigma-model, following the flux formulation of Born geometry
from Section~\ref{sec:fluxes}.
We write the Born sigma-model associated with the splitting
$TM=L_+\oplus L_-$ in local coordinates $\IX^I$ on $M$ by letting the map $\phi: (\Sigma, h) \rightarrow (M, \cH)$ pull the structures $(\cH,\omega)$ back to $\Sigma.$ The action functional \eqref{actcomplete} can then be written as
\be \label{locbsm1}
\cS[\phi]=\frac{1}{4}\, \int_\Sigma\, \Big((\bar g_+)_{ij}\,
\bar\Theta^i \wedge \star\, \bar\Theta^j + (\bar g_-)^{ij}\, \bar{\tilde\Theta}_i \wedge \star\, \bar{\tilde\Theta}_j \Big)  + \frac{1}{2}\, \int_\Sigma\, \bar\Theta^i \wedge \bar{\tilde\Theta}_i \ ,
\ee
where the coframe $\{\Theta^I\}=\{\Theta^i, \, \tilde{\Theta}_i\}$ is generally
given by $C^\infty(M)$-linear combinations of the holonomic coframe
$\{\de \IX^I\},$ thus we pull all of
them back to the worldsheet
$\Sigma$ by $\phi.$ To highlight the coordinate dependence, we write the local expression for the Born sigma-model as
\be \label{locbsm2}
\cS[\phi]= \frac{1}{4}\, \int_\Sigma\, \bar\cH_{IJ}\, \de \bar
\IX^I \wedge \star\, \de \bar \IX^J + \frac{1}{4}\, \int_\Sigma\, \bar\omega_{IJ}\, \de \bar \IX^I
\wedge \de \bar \IX^J \ ,
\ee
where the smooth functions $\cH_{IJ}$ and $\omega_{IJ}$ are the
components of $\cH$ and $\omega,$ respectively, in their local
expressions in the holonomic coframe $\{\de \IX^I \}.$

It is important to note that all of the information regarding
the generalized fluxes are contained in the topological term of the Born
sigma-model. This can be seen by writing \eqref{sigmatop} as a Wess-Zumino term
$$
\cS_{\rm top}[\phi] = \frac12\,\int_V\,\bar\cK \ ,
$$
where $\bar\cK$ is the pullback of the 3-form
$\cK=\de\omega$ from \eqref{eq:cK} by an extension of
$\phi:\Sigma\to M$ to a
three-dimensional manifold $V$ with boundary $\partial V=\Sigma$.
Writing $\cK = \tfrac16\, \cK_{IJK}\, \de \IX^I\wedge\de \IX^J\wedge\de \IX^K$ for its local expression in the holonomic coframe, the local form of the Born sigma-model can then be expressed as
\be\nonumber
\cS[\phi]= \frac{1}{4}\, \int_\Sigma\, \bar\cH_{IJ}\, \de\bar
\IX^I \wedge \star\, \de\bar \IX^J + \frac1{12}\,
\int_V\, \bar\cK_{IJK} \, \de\bar \IX^I\wedge\de\bar
\IX^J\wedge\de \bar \IX^K \ .
\ee
This form of the action functional shows that our Born sigma-model is
an immediate
generalization of the doubled sigma-models that were introduced for
the doubled twisted torus
in~\cite{Hull2009}; however, in our formalism the 2-form $\omega$ is
globally defined and the action functional can be defined without
resorting to an extension of the two-dimensional field theory to three
dimensions, as in~\cite{Hull2007}. 

The generalized fluxes
also obstruct the harmonic property of
$\phi:\Sigma\to M$. After integrating by parts and using Stokes'
Theorem, the equations of motion $\delta\cS[\phi]/\delta\IX^I=0$ for the action functional
\eqref{locbsm2}
read
\be \nonumber
\de\star\,\bar\cH_{IJ}\,\de\bar\IX^J +
\bar\cK_{IJK}\,\de\bar\IX^J\wedge\de\bar\IX^K=0 \ .
\ee
It follows that the Born sigma-model is a theory of harmonic maps only
for almost para-K\"ahler target spaces $(M,K,\omega)$, i.e. when
$\cK=\de\omega=0$. In general, the equations of motion determine extremal surfaces $\phi(\Sigma)\subset M$ with respect to a connection with torsion determined by
the 3-form $\cK$.

Thus far we have only introduced the obvious definition of a
worldsheet sigma-model for a Born manifold $M$. In the following we will
determine when a gauging of such a sigma-model model is possible if at
least one of the distributions $L_+$ or $L_-$ is integrable,
i.e.~assuming that $M$ is a Born manifold foliated by a maximally isotropic regular foliation.  
For definiteness, we suppose that the sub-bundle $L_-$ is integrable,
i.e. $L_-=T\cF_-.$ Provided that $\cS(\cH, \omega)$ satisfies certain
{generalized isometry conditions} as spelled out in~\cite{Kotov2018},
the gauging of the Born sigma-model reduces it to a worldsheet
sigma-model $S(g,b)$ for the quotient $\cQ=M/\cF_-.$ When the Born manifold admits
a maximally isotropic integrable distribution, we shall discuss the situations under
which these compatibility conditions are met; these include the
standard constraints for gauging an isometry, but in general the
foliation $\cF_-$ need not be generated by the action of a Lie group, and in fact a generic Born manifold need not admit a Lie algebra of
Killing vector fields. This will lead to a
description of T-duality for Born sigma-models and their reduced sigma-models for the quotient spaces~$\cQ.$

\medskip

\subsection{Gauged Sigma-Models for Foliated Manifolds} \label{isononiso} ~\\[5pt]
We shall start by briefly reviewing the formalism for the
gauging of a sigma-model. Then we will analyze in more detail gaugings
which involve regular foliations of the target space, following the general
treatment of~\cite{Kotov2018, Kotov2014} which also applies to
singular foliations.

\medskip

\subsubsection{The Standard Isometric Gauging} \label{sec:standardgauging} ~\\[5pt]
Let $\phi: (\Sigma, h) \rightarrow (M, \cH)$ be a
sigma-model.\footnote{For the present discussion, $(M, \cH)$ is a
  Riemannian manifold but not necessarily a Born manifold, unless
  specified explicitly.} Suppose that the Riemannian metric $\cH,$
defining the action functional $\cS_0[\phi]$ in the form \eqref{sigmaaction}, has
isometry group $\sfG$ with Lie algebra $\frg= \sfL \sfi \sfe (\sfG).$ Then there is a Lie algebra homomorphism
$$
\rho: \frg \longrightarrow \mathsf{\Gamma}(TM)
$$ 
such that the Killing vector $X=\rho(\sfx)$ corresponding to every $\sfx \in \frg$ satisfies
\be \label{killingeq}
\pounds_X \cH =0\ .
\ee
This homomorphism can also be regarded as a bundle map 
$$
\rho: M \times \frg \longrightarrow TM \ , \quad (p,\xi)\longmapsto \rho(\xi)\big|_p\in T_pM
$$ 
of constant rank covering the identity. The left action of $\sfG$ on $M$ by isometries is a rigid symmetry of the sigma-model $\cS_0[\phi]$. Gauging this symmetry reduces the sigma-model for $M$ to a sigma-model for the quotient $M/\sfG$. For this, 
we construct a new gauged action functional $\cS_0[\phi,A]$ by considering the trivial
principal $\sfG$-bundle\footnote{This discussion extends to any choice of principal $\sfG$-bundle on $\Sigma$, but we
  work with the trivial $\sfG$-bundle to simplify the presentation.} $\Sigma \times \sfG\to\Sigma$
and choosing a $\sfG$-connection on it. This gives a $\frg$-valued connection
1-form $A \in \mathsf{\Gamma}(T^*\Sigma \, \otimes \,\frg)$ which can
be used as the gauge field in 
$\cS_0[\phi,A].$ We incorporate this worldsheet gauge field by minimal coupling, i.e. by ``covariantizing'' the map $\de \phi: T\Sigma \rightarrow \phi^*TM$ in the following way.   

The connection 1-form $A,$ together with $\phi,$ can be regarded as a
bundle map $\bar{A}:T\Sigma \rightarrow M \times \frg$ covering
$\phi$, i.e. $\bar A$ is the map defined by the commutative diagram
\begin{center}
\begin{tikzcd}
T\Sigma \arrow[r, "\bar{A}"] \arrow[d]
& M\times \frg \arrow[d] \\
\Sigma \arrow[r, swap, "\phi "]
& M
\end{tikzcd}
\end{center}
The trivial vector bundle $M \times \frg\to M$ has a natural Lie algebroid structure with $\rho$ as anchor map, given by the action algebroid associated with
the action of $\frg$ on $M$: The sections of $M\times\frg$ are naturally identified with smooth $\frg$-valued functions on $M$, and given $f,g\in C^\infty(M,\frg)$ their Lie bracket is defined by
$$
[f,g](p) = \big[f(p),g(p)\big]_\frg + \rho\big(f(p)\big)\big|_p\,g - \rho\big(g(p)\big)\big|_p\,f 
$$
for all $p\in M$, where $[\,\cdot\,,\,\cdot\,]_\frg$ is the Lie bracket of $\frg$.
Then there is a composition of bundle maps 
$$
T\Sigma \xlongrightarrow{\bar{A}} M \times \frg
\xlongrightarrow{\rho} TM 
$$
and the image ${\rm Im}(\rho \, \circ \, \bar{A})$ is a vector sub-bundle of $TM.$ We consider the pullback bundle $\phi^* {\rm Im}(\rho \, \circ \, \bar{A})$ on $\Sigma,$ which is a sub-bundle of $\phi^* TM,$ so that $\rho \, \circ \, \bar{A}$ induces a bundle map 
$\overline{\rho(A)}: T\Sigma \rightarrow \phi^* {\rm Im}(\rho \, \circ \, \bar{A})$ covering the identity.

Now we define a new vector bundle map by 
$$
\mathrm{D}^A\phi= \de \phi -\overline{\rho(A)}: T\Sigma \longrightarrow
\phi^* TM\ ,
$$ 
which we regard as the ``covariantization'' of $\de \phi.$ The
bundle map $ \overline{\rho(A)}$ gives a tensor
$\bar{\rho}(A) \in \mathsf{\Gamma}(T^*\Sigma \otimes \phi^* TM).$ Thus $\De^A \phi$ 
is associated with an element $\underline{\De^A \phi} \in
\mathsf{\Gamma}(T^*\Sigma \otimes \phi^* TM)$. Recalling that the
vector space of sections $\mathsf{\Gamma}(T^*\Sigma \otimes \phi^* TM)$ is
endowed with a norm induced by the metric $\cH,$ we can thereby write the gauged action functional
\be\label{eq:gaugedaction}
\cS_0[\phi,A]=\frac{1}{4}\, \int_{\Sigma}\, \bar\cH_{IJ} \, \De^A
\phi^I \wedge \star\, \De^A \phi^J=\frac{1}{4}\,\int_\Sigma\, \|\,\underline{\De^A \phi}\,
\|_{h,\cH}^2 \ \de \mu(h) \ ,
\ee
where the norm $\|\,\underline{\De^A \phi}\, \|_{h,\cH}$ is defined as in
\eqref{sigmanorm}. This two-dimensional field theory is invariant
under infinitesimal gauge transformations in $C^\infty(\Sigma,\frg)$, which is the Lie algebra of sections of the pullback of the action algebroid $(M\times\frg,\rho)$ by $\phi$. A Lie groupoid integrating the action algebroid is given by the action groupoid $\sfG\ltimes M\rightrightarrows M$ associated with the smooth left $\sfG$-action on $M$, whose orbit space is $M/\sfG$, and sections of its pullback form the
group of gauge transformations $C^\infty(\Sigma,
\sfG)$ leaving \eqref{eq:gaugedaction} invariant.

The gauging of the action \eqref{actcomplete} with topological term \eqref{sigmatop} is completely analogous, provided that the 2-form $\omega$ satisfies the condition 
$$
\pounds_X\, \omega=0\ ,
$$
where $X=\rho(\sfx),$ for all $\sfx \in \frg.$
The gauged topological term then reads
$$
\cS_{\rm top}[\phi,A]= \frac{1}{4}\,\int_{\Sigma}\,
\bar\omega_{IJ}\,  \De^A \phi^I \wedge \De^A \phi^J\ ,
$$
where $\omega_{IJ}$ are local smooth functions on $M$ given by the
components of the 2-form $\omega$, and subsequently pulled back to $\Sigma$ by $\phi.$ 

\medskip

\subsubsection{The Kotov-Strobl Gauging} ~\\[5pt]
The isometric gauging may be generalized by replacing
the action algebroid $M \times \frg$ with any Lie algebroid $\sfA$
over $M$ and replacing the worldsheet gauge field $A$ with an
$\sfA$-valued $1$-form. This generalization is discussed in \cite{Kotov2018, Kotov2014}. 
In this instance we would like to construct the covariantization of
the map $\de \phi$ coming from a Lie algebroid connection and the
generalization of the isometry conditions which allow such
gauging. Let $\sfA\rightarrow M$ be a Lie algebroid with anchor map
$\rho: \sfA \rightarrow TM$ which is endowed with an connection 
$$
\nabla: TM \times \sfA \longrightarrow \sfA \ .
$$
By definition it satisfies the Leibniz rule
$$
\nabla_X(f\,\sfa) = \pounds_X(f)\,\sfa+f\,\nabla_X\sfa
$$
for all $X\in\mathsf{\Gamma}(TM)$, $f\in C^\infty(M)$ and $\sfa\in\mathsf{\Gamma}(\sfA)$.
Consider the short exact sequence of vector bundles
$$
0 \longrightarrow T^*M \otimes \sfA \longrightarrow J^1(\sfA)
\longrightarrow \sfA \longrightarrow 0
$$
where $J^1(\sfA)$ is the first jet bundle of smooth sections of
$\sfA.$ Then connections $\nabla$ are in one-to-one
correspondence with splittings 
$$
s: \sfA \longrightarrow J^1(\sfA)
$$
of this
short exact sequence, and for every $ j^1(\sfa) \in
\mathsf{\Gamma}(J^1(\sfA))$ with $\sfa \in \mathsf{\Gamma}(\sfA)$ one has 
$$
j^1(\sfa)= s(\sfa)- \nabla \sfa\ ,
$$ 
where $\nabla \sfa \in \mathsf{\Gamma}(T^*M  \otimes \sfA).$ We may therefore consider the $\sfA$-valued 1-form $\ccA \in
\mathsf{\Gamma}(T^*M \otimes \sfA)$ defined by $\nabla$, which
together with $\phi$ gives a bundle map
$\overline{\ccA}:T\Sigma\to\sfA$ covering $\phi: \Sigma \rightarrow M$
which generalizes the map $\bar{A}$ discussed in Section~\ref{sec:standardgauging}.

To extend the covariantization of $\de \phi,$ we consider
the case in which $\rho$ has constant rank. Then there is a composition of vector bundle maps
$$
T\Sigma \xlongrightarrow{\phi^* \overline{\ccA}} \sfA
\xlongrightarrow{\rho} TM 
$$
where $\phi^* \overline{\ccA}$ is obtained from $\ccA$ by pulling back the $T^*M$ factor
to $\Sigma$, so that the image ${\rm Im}(\rho \, \circ \,
\phi^* \overline{\ccA}\,)$ is a vector sub-bundle of $TM.$ We pull it
back to $\Sigma$ and obtain the vector sub-bundle $\phi^*{\rm Im}(\rho \, \circ \, \phi^* \overline{\ccA}\,)$ of $\phi^*TM.$ Thus $\rho \, \circ\, \phi^* \overline{\ccA}$ induces a bundle map
$\overline{\rho(\ccA)}: T\Sigma \rightarrow \phi^*{\rm Im}(\rho \, \circ \, \phi^* \overline{\ccA}\,)$ covering the identity, which in turn can be regarded as a tensor $\bar{\rho}(\ccA) \in \mathsf{\Gamma}(T^*\Sigma \otimes \phi^* TM).$ Hence the map $\de \phi$ is covariantized by considering
$$
\De^\nabla \phi= \de \phi - \overline{\rho(\ccA)}
$$
and a gauged action of the form \eqref{sigmanorm} can be written by
regarding $\De^\nabla \phi$ as a tensor $\underline{\De^\nabla \phi}
\in \mathsf{\Gamma}(T^*\Sigma \otimes \phi^* TM)$, which has a well-defined
norm $\|\,\underline{\De^\nabla \phi}\, \|_{h,\cH}$ induced by the metric $\cH.$

Following \cite{Kotov2018}, we shall next discuss the generalization of the isometry conditions by introducing the induced connection 
$$
\prescript{\tau}{}{\nabla}: \mathsf{\Gamma}(\sfA)\times \mathsf{\Gamma}(TM)\longrightarrow \mathsf{\Gamma}(TM)
$$ 
on the tangent bundle $TM$ by 
\be \label{indcon}
\prescript{\tau}{}{\nabla}_\sfa X=\big[\rho(\sfa),X\big]+\rho(\nabla_X \sfa)\ ,
\ee
where $X \in \mathsf{\Gamma}(TM)$ and $\sfa \in \mathsf{\Gamma}({\sf A}).$
The superscript refers to the canonical representation $\tau$ of the jet bundle $J^1(\sfA)$ on $TM$ which, when combined with the connection $\nabla,$ gives the $\sfA$-connection $ \prescript{\tau}{}{\nabla}$ (see \cite{Kotov2018, Fernandes2002,Blaom2006} for details). This provides the last ingredient we need to generalize the isometric gauging construction.

\theoremstyle{definition}
\begin{definition}
Let $(\sfA,\rho)$ be a Lie algebroid over a Riemannian manifold $(M, \cH)$ endowed with an connection $\nabla.$ Then $\nabla$ and $\cH$ are \emph{compatible} if 
\be \label{killie}
\prescript{\tau}{}{\nabla}\cH=0 \ ,
\ee
where $\prescript{\tau}{}{\nabla}$ is defined in \eqref{indcon}.
When this condition holds, the triple $(\sfA,\nabla, \cH)$ is a \emph{Killing Lie algebroid}.
\end{definition}
If we set $\sfA=M \times \frg,$ the action algebroid on $(M, \cH)$ for the Lie algebra $\frg$ of Killing vector fields of the metric $\cH,$ then \eqref{killie} is exactly the Killing equation \eqref{killingeq}. 

We can recast the compatibility condition \eqref{killie} in the following form. Let  $\{ a_i \}$ be a local basis of $\mathsf{\Gamma}(\sfA)$ with $i=1,\dots, \mathrm{rank}(\sfA).$ Then \eqref{killie} can be written as
\be \label{lockillie}
\pounds_{\rho(a_i)} \cH= \Omega^j_i \odot \iota_{\rho(a_j)}\cH \ , 
\ee  
where $\Omega^j_i$ are defined by the action of the connection
$\nabla$ on basis sections, $\nabla a_i = \Omega^j_i \otimes a_j\in
\mathsf{\Gamma}(T^*M\otimes \sfA),$ so that the connection
coefficients $\Omega^j_i$ are 1-forms on $M$. As we now explain, \eqref{lockillie} is an immediate generalization of the Killing equation \eqref{killingeq}.
The image of the anchor map $\rho: {\sf A} \rightarrow TM$ defines a
generalized distribution on $TM$ which is integrable in the sense of
the Stefan-Sussman Theorem, therefore $M$ is foliated by a singular
foliation $\cF.$ The main idea behind \eqref{lockillie} is to give a
condition stating that the components of the Riemannian metric $\cH$
transverse to the leaves, which induces a metric on the leaf space,
must be constant along the leaves. We will elaborate on this a bit
further supposing that the rank of $\rho$ is constant, i.e. the
foliation $\cF$ is regular, as our main interest in this paper concerns gaugings with Lie algebroids characterized by injective anchor maps, and we will be particularly interested in this case for our applications to Born sigma-models later on.

In the case of a regular foliation $\cF$ of $M,$ there is a short exact sequence of vector bundles 
$$
0 \longrightarrow T\cF \xlongrightarrow{i} TM  \xlongrightarrow{\bar{q}} N\cF \longrightarrow 0
$$
where $N\cF= TM / T\cF$ is the normal bundle of $\cF$ and $\bar{q}: TM
\rightarrow N\cF$ is the quotient map. We can always choose an
orthogonal splitting 
$$
s_\perp: N\cF \longrightarrow TM
$$
of this exact
sequence with respect to $\cH$, so that ${\rm Im}(s_\perp)=T\cF^\perp$
and 
$$
TM=T\cF \oplus T\cF^\perp \ ,
$$ 
where $T\cF^\perp \simeq N\cF$ is the orthogonal complement of $T\cF$ in the metric $\cH.$ 
This also induces a splitting of the cotangent bundle $T^*M= T^*\cF \oplus (T\cF^\perp)^*,$ where $(T\cF^\perp)^* \simeq N^*\cF$ with $N^*\cF$ the conormal bundle of $\cF.$ Then we can decompose the metric $\cH$ as
$$
\cH= g_\parallel + g_\perp\ ,
$$
where $g_\parallel$ is a fiberwise metric on the sub-bundle $T\cF$ and
$g_\perp$ is a fiberwise metric on $T\cF^\perp.$ This allows us to
rephrase the condition \eqref{lockillie} by saying that the Lie
derivative $\pounds_{X_\parallel} \cH,$ for every vector field $
X_\parallel \in \mathsf{\Gamma}(T\cF),$ can only have components in
$\mathsf{\Gamma}(T^*\cF \otimes T^*\cF),$ $\mathsf{\Gamma}(T^*\cF
\otimes (T\cF^\perp)^*)$ and $\mathsf{\Gamma}((T\cF^\perp)^* \otimes
T^*\cF),$ i.e. $(\pounds_{X_\parallel} \cH)_\perp =0$ in
$\mathsf{\Gamma}((T\cF^\perp)^* \otimes (T\cF^\perp)^*)$. It is easy to see that this constraint is equivalent to
\be \label{transvinvariance}
\pounds_{X_\parallel} g_\perp=0\ , 
\ee
for all $X_\parallel \in \mathsf{\Gamma}(T\cF).$ This implies that the component $g_\perp$ of the metric $\cH$ is transverse invariant.
In other words, the fiberwise metric $g_\perp$ induced by $\cH$
satisfies ${\rm Ker}(g_\perp)= T\cF$ and also \eqref{transvinvariance}
whenever the gauging is possible. These gauging constraints are simply
the defining conditions for $(M, g_\perp, \cF)$ to be a Riemannian foliation \cite{Molino}. 

We will now show that this statement about the gauging of sigma-models
is equivalent to saying that $(M, \cH, \cF)$ is a foliated manifold
with a bundle-like metric \cite{Reinhart1959,Hermann1960} whenever the foliation
is regular. 
\begin{definition}
A Riemannian metric $\cH$ on a
foliated manifold
$(M,\cF)$ is (\emph{totally geodesics}) \emph{bundle-like} if 
\be\nonumber
\pounds_{X_\parallel}\cH(Y_\perp, Z_\perp)=0\ ,
\ee
for all $ X_\parallel \in \mathsf{\Gamma}(T\cF)$ and $Y_\perp, Z_\perp \in \mathsf{\Gamma}(T\cF^\perp).$
\end{definition}
The leaf holonomy invariance of bundle-like metrics is discussed in \cite{Hermann1960}. The motivating example for this structure is given by the compatible generalized metrics of Example \ref{fiberborn}.

\begin{theorem}\label{thm:bundlelike}
Let $M$ be a manifold endowed with a regular foliation $\cF$ and a
Riemannian metric $\cH.$ Then the gauging condition \eqref{killie} holds if and only if $\cH$ is bundle-like.
\end{theorem}
\begin{proof}
Suppose that $\cH$ is a bundle-like metric. Let ${\sf A}=T\cF$ be the
Lie algebroid with anchor map given by the inclusion $i: T\cF
\hookrightarrow TM.$ We choose the Lie algebroid connection $\nabla$
to be the unique Bott connection\footnote{See \cite{Tondeur1988,
    Baudoin-Lect} for background on Bott connections on foliated Riemannian manifolds. In our case we use the Bott connection on the tangent bundle rather than on the normal bundle.} $\nabla^{\tt B}$ defined by $(M, \cH,
\cF)$ restricted to $T\cF$: 
$$
\nabla^{\tt B}: \mathsf{\Gamma}(TM)
\times \mathsf{\Gamma}(T\cF)\longrightarrow \mathsf{\Gamma}(T\cF)
$$ 
given
by $\nabla^{\tt B}_X Y_\parallel \in \mathsf{\Gamma}(T\cF),$ for all $
X \in \mathsf{\Gamma}(TM)$ and $ X_\parallel \in
\mathsf{\Gamma}(T\cF)$. The restriction is always well-defined as a
Lie algebroid connection since this reflects one of the properties of
the Bott connection. 
To show that \eqref{killie} holds, we give the
representation $\prescript{\tau}{}{\nabla}^{\tt B}$ for this
restriction of the Bott connection by using \eqref{indcon}:
\be \label{reprbott}
\prescript{\tau}{}{\nabla}^{\tt B}_{Y_\parallel} X=[Y_\parallel ,X] + \nabla^{\tt B}_X Y_\parallel \ .
\ee
The only non-zero component of the torsion tensor $T^{\tt B}$ of a Bott
connection is $T^{\tt B}(X_\perp,
Y_\perp) \in \mathsf{\Gamma}(T\cF),$ i.e. $T^{\tt B}(X_\perp,
Y_\parallel)=T^{\tt B}(X_\parallel, Y_\parallel)=0.$ This implies
$$
\nabla^{\tt B}_X Y_\parallel= \nabla^{\tt B}_{Y_\parallel} X + [X,
Y_\parallel ]\ ,
$$ 
which when substituted in \eqref{reprbott} gives 
$$ 
\prescript{\tau}{}{\nabla}^{\tt B}_{Y_\parallel} X= \nabla^{\tt
  B}_{Y_\parallel} X \ .
$$
Since $\nabla^{\tt B}$ is a metric connection for the bundle-like
metric $\cH,$ the compatibility condition \eqref{killie} follows.

Conversely, we have already shown above that \eqref{killie} is equivalent to the condition \eqref{transvinvariance}, which is the defining condition of a bundle-like metric. 
\end{proof}

To understand when the condition \eqref{killie} allows us to recover a
sigma-model for the quotient space of a foliation, we need an additional concept.
\begin{definition}
Let $(M, \cH)$ and $(\cQ, g)$ be Riemannian manifolds. Let ${\mathit\Pi}: M
\rightarrow \cQ$ be a surjective submersion, so that the orthogonal complement ${\rm Ker}(\de {\mathit\Pi})^\perp \subset TM$ defines the horizontal sub-bundle complementary to ${\rm Ker}(\de {\mathit\Pi}).$ Then ${\mathit\Pi}$ is a \emph{Riemannian submersion} if the isomorphism $\de {\mathit\Pi}:{\rm Ker}(\de {\mathit\Pi})^\perp \rightarrow T\cQ$ is an isometry:
$$ 
\cH(X_{\tt h}, Y_{\tt h})= g\big(\de {\mathit\Pi} (X_{\tt h}), \de {\mathit\Pi}
(Y_{\tt h})\big)\ , 
$$  
for all $X_{\tt h}, Y_{\tt h} \in \mathsf{\Gamma}({\rm Ker}(\de \phi)^\perp)$.
\end{definition}

Whenever the quotient $\cQ=M/\cF$ of a foliated manifold
$(M,\cF)$ equipped with a bundle-like metric $\cH$ is smooth, we can
identify $(T\cF^\perp)^* \simeq N^*\cF$ with $T^*\cQ$ and find that
$g_\perp$ induces a metric $g$ on the quotient manifold $\cQ.$ This
describes a Riemannian submersion of $(M, \cH)$ onto $(\cQ,g).$ We
interpret the Riemannian submersion $(M, \cH) \rightarrow (\cQ,g)$ as
a way to relate a sigma-model $\cS(\cH)$ for a foliated Riemannian
manifold $(M,\cH)$ to a sigma-model $S(g)$ for the leaf space $\cQ$
endowed with a Riemannian metric $g$ obtained from $\cH$ through the
gauging. From this point of view, the constraint \eqref{transvinvariance} is simply a condition for the metric $g_\perp$ to be well-defined on the leaf space $\cQ.$ 

Under these conditions, the gauged sigma-model thus constructed is
invariant under the Lie algebra of sections of the pullback to $\Sigma$ of the Lie algebroid $(T\cF,i)$ on $M$, which is a Lie subalgebroid of the tangent algebroid $(TM,\unit_{TM})$. An integrating Lie groupoid for $TM$ is given by the pair groupoid $M\times M\rightrightarrows M$ of the manifold $M$, with orbit space $M$, while $T\cF$ is integrated by the Lie subgroupoid of $M\times M$ given by the graph of the equivalence relation on $M$ defined by the surjective submersion $M\to\cQ$~\cite{Mackenzie1995}, with orbit space $M/\cF=\cQ$. The pullback of this to $\Sigma$ determines the groupoid of gauge transformations which leaves the resulting gauged sigma-model invariant.

\medskip

\subsubsection{Incorporating the $B$-Field} ~\\[5pt]
Finally, we wish to include a topological term in the action functional of the sigma-model, i.e. a $B$-field. This extension of the present formalism is also discussed in \cite{Kotov2018}.

\theoremstyle{definition}
\begin{definition}
Let $(\sfA, \rho)$ be a Lie algebroid over $(M, \Phi),$ where $\Phi \in \mathsf{\Gamma}(T^*M \otimes T^*M)$ is a non-degenerate bilinear form. Endow $\sfA$ with a connection  $\nabla$ and let $\psi \in \mathsf{\Gamma}(T^*M \otimes {\sf End}(\sfA)).$ Then $(\sfA,\nabla, \psi)$ and $(M, \Phi)$ are \emph{compatible} if 
\be \label{compatop}
(\prescript{\tau}{}{\nabla}^{+}\otimes \unit + \unit \otimes \prescript{\tau}{}{\nabla}^{-})(\Phi)=0\ , 
\ee
where 
$$
\nabla^{\pm}=\nabla \pm \psi
$$ 
and $\prescript{\tau}{}{\nabla}$
is given in \eqref{indcon}.
\end{definition}
To apply this to the case at hand, we also describe the local
expression of \eqref{compatop}. Let us write ${\sf Sym}\{\Phi\}=\cH$
and ${\sf Alt}\{\Phi\}=\omega$ for the symmetric and skew-symmetric
parts of the $(0,2)$-tensor $\Phi$, and let $\{a_i\}$ be a local basis of $\mathsf{\Gamma}(\sfA)$ such that 
$$
\nabla a_i = \Omega^j_i \otimes a_j \qquad \mbox{and} \qquad \psi a_i
= \psi^j_i \otimes a_j \ ,
$$
where $\Omega^j_i$ and $\psi^j_i$ are 1-forms on $M$.
Then \eqref{compatop} can be cast in the form
\begin{align}
\pounds_{\rho(a_i)}\cH &= \Omega^j_i \odot \iota_{\rho(a_j)} \cH + \psi^j_i \odot \iota_{\rho(a_j)}\omega \label{compa1} \ , \\[4pt]
\pounds_{\rho(a_i)}\omega &= \Omega^j_i \wedge \iota_{\rho(a_j)} \omega + \psi^j_i \wedge \iota_{\rho(a_j)}\cH \ . \label{compa2}
\end{align}
These equations represent the generalization of \eqref{lockillie},
thus the condition expressed by \eqref{compatop} is usually referred
to as a \emph{generalized isometry} or as the condition for a \emph{Lie algebroid gauging}. 

A relevant example for our purposes is given by a sigma-model on a foliated
manifold. Let $(M,\Phi)$ be a manifold foliated by $\cF.$ We then
naturally consider as Lie algebroid $\sfA= T\cF.$ The gauging of the
corresponding sigma-model is possible if \eqref{compatop} is satisfied
and the bilinear form $\Phi$ induces a bilinear form $\Phi^{\cQ}$ on
$N^*\cF.$ Whenever the quotient $\cQ=M/\cF$ is smooth, the bilinear
form $\Phi^\cQ$ defines a section of $T^*\cQ \otimes T^*\cQ.$ 

In this paper we are mostly concerned with the conditions that the
geometry of the target space must satisfy in order to admit a gauging
of the generalized isometry. For a discussion involving the structure
of the pullback gauge Lie algebroid on the worldsheet $\Sigma$, see \cite{Wright2019}. 

\medskip

\subsection{Gauging the Born Sigma-Model} \label{BSMGau} ~\\[5pt]
The discussion of Section~\ref{isononiso} can be applied to a Born manifold $(M,K,\omega, \cH)$ simply by setting 
$$
\Phi= \cH+ \omega
$$
in $\mathsf{\Gamma}(T^*M \otimes T^*M),$ which is non-degenerate by construction. By assuming that one of the eigenbundles $L_{\pm}$ of $K$ is integrable, we obtain a foliation $\cF_{\pm}$ of $M.$ 
Then we can regard the map $(M, \Phi) \rightarrow (\cQ, \Phi^{\cQ}),$
where $\cQ=M/ \cF_{\pm}$ and 
$$
\Phi^\cQ= g+b \ ,
$$ 
as a reduction of a Born
sigma-model $\cS(\cH, \omega)$ to a sigma-model $S(g, b)$ for the leaf
space $\cQ.$ The specialization of the general formalism for the
gauging for a foliated manifold comes from the compatibility conditions that the Riemannian metric $\cH$ now satisfies with the para-Hermitian structure.

Here we shall focus on the local construction of the gauging of a Born
sigma-model. From now on $\cH$ will be the compatible generalized metric on a Born manifold $M$ and $\omega$ will be its fundamental 2-form. To apply the theory of gauged sigma-models with generalized isometry
reviewed in Section~\ref{isononiso}, we need a foliated target space
$M.$ Since the target space of a Born sigma-model is an almost
para-Hermitian manifold $(M, K, \eta),$ the tangent bundle is given by
the splitting $TM=L_+\oplus L_-$ as described in Section
\ref{intropara}. We require
one of the eigenbundles $L_\pm$ of $K$ to be integrable; let us take
it to be $L_-$ for definiteness. Then $M$ is foliated by the leaves of
$\cF_-$ such that $T\cF_-=L_-.$ From a local perspective, this amounts
to the assumption that any frame for $TM$ closes a Lie algebra of the
form \eqref{fluxes} with $R^{mnk}=0$; we shall return to the case when
$R^{mnk}\neq0$ in Section~\ref{specT}.
The bundle $L_-$ is a Lie subalgebroid of the tangent algebroid
$TM$ with the anchor map given by the inclusion $ i:L_-
\hookrightarrow TM$. When the space of leaves is a smooth manifold, $L_-$ is integrated by the Lie subgroupoid of the pair groupoid $M\times M\rightrightarrows M$ given by the graph of the equivalence relation on $M$ defined by the surjective submersion $M\to\cQ=M/\cF_-$; we shall discuss the case when the leaf space is not smooth in Section~\ref{specT}. It is natural to use this Lie algebroid in the
application of the formalism of Section~\ref{isononiso}, and
investigate under which geometric constraints the generalized isometry
conditions \eqref{compa1} and \eqref{compa2} are satisfied for a Born
sigma-model with a foliated para-Hermitian target space. For this, we
require a Lie algebroid connection on $L_-$; as discussed in the proof
of Theorem~\ref{thm:bundlelike}, a natural candidate is the Bott
connection $\nabla^{\tt B}$. 

The traditional approach to the gauging of a
generalized isometry involves applying a non-standard variation of the
dynamical
fields $\phi$ and $A$ of the gauged Born sigma-model, where $A$ is the
$L_-$-valued connection 1-form on $M$. Since the anchor map is the
inclusion $i$ of $L_-$ in $TM,$ the Lie algebroid bracket is locally given by 
\be \label{eq:tildeZQ}
[\tilde{Z}^m, \tilde{Z}^n]=Q_k{}^{mn}\, \tilde{Z}^k \ .
\ee
The map $\de \phi$ is covariantized upon introducing 
\be \label{gaude}
\De^A \phi^I = \de \phi^I - \bar{i}^{Ij}\, \bar A_j \ ,
\ee
where $\bar{i}$ is the pullback of $i$ along $\phi.$ 
The gauged sigma-model with topological term is thus written as 
\be \label{gauact}
\cS[\phi,A]=\frac{1}{4}\,\int_{\Sigma}\,\bar\cH_{IJ}\, \De^A \phi^I \wedge \star\, \De^A \phi^J + \frac{1}{4}\,\int_{\Sigma}\, \bar\omega_{IJ}\, \De^A \phi^I \wedge \De^A \phi^J \ ,
\ee
and the variations of the fields $\phi$ and $A$ under the infinitesimal gauge transformations generated by the vector fields $\tilde{Z}^i$ read
\begin{align}
\delta_{\epsilon}\phi^I &= \bar{i}^{Ij}\,\epsilon_j \ , \label{varia1} \\[4pt]
\delta_{\epsilon}\bar A_i &=\de \epsilon_i + \bar Q_i{}^{jk}\,
                            \epsilon_k\, \bar A_j + \bar\Omega_{iJ}^{k}\,\epsilon_k\, \De^A \phi^J + \bar\psi_{iJ}^{k}\, \epsilon_k\, \star\, \De^A \phi^J \ , \label{varia2}
\end{align}
where $\epsilon_i \in C^{\infty}(\Sigma).$ 
It can be verified that the gauged action functional \eqref{gauact} is invariant under these transformations:
\be \nonumber
\delta_{\epsilon}\cS[\phi,A]=0 \ ,
\ee
if and only if the conditions \eqref{compa1} and \eqref{compa2} are
satisfied~\cite{Deser2016b}. 

We also give the gauging of the Born sigma-model in local coordinates 
$$
\IX^I=(x^i, \tilde{x}_i)
$$ 
on $M,$ where the foliation $\cF_-$ has leaves whose adapted coordinates are given by $\tilde{x}_i$ for $i=1,\dots,d$.
Then the gauged Born sigma-model coming from the local expression
\eqref{locbsm2} is obtained by replacing the pulled back 1-forms $ \de
\bar x^i$ and $\de \bar {\tilde{x}}{}_i$ by the covariantized maps\footnote{Here and in the following we suppress the pullback of the inclusion $i.$}
\be \label{covdiff}
\De^A \bar x^i= \de \bar x^i \qquad \mbox{and} \qquad \De^A
\bar{\tilde{x}}{}_i = \de \bar{\tilde{x}}{}_i - \bar A_i \ .
\ee
In the quantum theory, this minimal coupling allows $\de
\bar{\tilde{x}}{}_i$ to be absorbed into a shift of the worldsheet gauge fields
$\bar A_i$.
The gauged Born sigma-model in local coordinates thus reads
\be \label{locgaubsm}
\cS[\phi, A]= \frac{1}{4}\, \int_\Sigma\, \bar\cH_{IJ}\, \De^A
\bar \IX^I \wedge \star\, \De^A \bar \IX^J + \frac{1}{4}\, \int_\Sigma\, \bar\omega_{IJ}\, \De^A\bar \IX^I \wedge \De^A\bar \IX^J \ ,
\ee
and the variations \eqref{varia1} and \eqref{varia2} under which the gauged sigma-model is invariant are written as
\begin{align*}
\delta_\epsilon\bar x^i&=0 \ , \\[4pt]
\delta_{\epsilon}\bar{\tilde{x}}{}_i&=\epsilon_i \ , \\[4pt]
\delta_{\epsilon}\bar A_i&=\de \epsilon_i + \bar Q_i{}^{jk}\,
                           \epsilon_k\, \bar A_j + \bar\Omega_{iJ}^{k}\,\epsilon_k\, \De^A\bar \IX^J + \bar\psi_{iJ}^{k}\, \epsilon_k\, \star\, \De^A\bar \IX^J \ .
\end{align*}

In the Born sigma-model, the almost para-Hermitian structure appears
solely in the topological term. In the sigma-model without topological
term, the generalized isometry condition for $\cH$ only involves the
assumption that $M$ must be foliated and does not capture deeper
information about the almost para-Hermitian manifold; in~\cite{Hull2005} the
topological term was introduced for doubled torus bundles in order to maintain invariance under
large gauge transformations in the corresponding gauged sigma-model
and found to play an important role in the quantum
theory~\cite{Hull2007}. To understand how the generalized isometry conditions \eqref{compa1} and \eqref{compa2} specialise to a Born sigma-model, let us discuss the local reduction of a Born sigma-model $\cS(\cH, \omega)$ to a sigma-model $S(g, b)$ for the leaf space. 

In the coordinates $\IX^I=(x^i, \tilde{x}_i)$ adapted to the foliation
$\cF_-,$ where $\tilde{x}_i$ are the leaf coordinates, the local frame
and dual coframe which respectively span $\mathsf{\Gamma}(TM)$ and
$\mathsf{\Gamma}(T^*M)$, and which diagonalize the almost para-complex structure $K$, are given in the form
$$ 
Z_i= \frac{\partial}{\partial x^i} + N_{ij}\, \frac{\partial}{\partial \tilde{x}_j} \qquad \mbox{and} \qquad \tilde{Z}^i= \frac{\partial}{\partial \tilde{x}_i} \ ,
$$
for local functions $N_{ij}$ on $M$, since there always exists a local completion $\{ \tilde{Z}^i \}$ of the holonomic frame for $\mathsf{\Gamma}(T\cF_-),$ and
$$
\Theta^i= \de x^i \qquad \mbox{and} \qquad \tilde{\Theta}_i= \de
  \tilde{x}_i - N_{ji} \,\de x^j\ ,
$$ 
which form a local basis for $L_+^*$ and $L_-^*$ respectively. The split signature metric $\eta$ assumes the form
\be\label{eq:etagen}
\eta=\eta^i_j\, \bigl((\de \tilde{x}_i - N_{ki}\, \de x^k )\otimes \de x^j + \de x^j \otimes (\de \tilde x_j - N_{kj}\, \de x^k) \bigl)\ ,  
\ee
while the fundamental 2-form $\omega$ reads
$$
\omega= \eta^{j}_i\, \de x^i \wedge \de \tilde{x}_j - \eta^k_j\, N_{ik}\, \de x^i \wedge \de x^j \ .
$$ 
Finally, a compatible generalized metric $\cH$ on $(M, K, \eta)$ is equivalent to specifying a fiberwise metric $g_+$ on $L_+,$ which locally reads 
$$
g_+= (g_+)_{ij}\, \de x^i \otimes \de x^j \ .
$$
Then the complete local expression for $\cH$ is given by
\begin{align}
\cH&= \bigl( (g_+)_{ij}+(g_-)^{kl}\,N_{ik}\,N_{jl} \bigr)\,
     \de x^i \otimes \de x^j + (g_-)^{ij}\, \de \tilde{x}_i \otimes
     \de \tilde{x}_j \nonumber \\
& \qquad - (g_-)^{jk}\,N_{ik}\, \de x^i \otimes \de
  \tilde{x}_j -(g_-)^{ik}\,N_{jk}\, \de \tilde{x}_i \otimes
  \de x^j \ ,
\label{eq:cHgen}\end{align}
where $g_-$ is defined in \eqref{eq:g-}.

The Born sigma-model $\cS(\cH, \omega)$ in the coordinates adapted to the foliation is written as in \eqref{locbsm2} and its gauging is obtained upon introducing the covariant derivatives \eqref{covdiff} to get the action functional \eqref{locgaubsm}. Then the gauged Born sigma-model reads
\begin{align}
\cS[\phi, A] &= \frac{1}{4}\,\int_\Sigma\, \Big( \bigl( (\bar g_+)_{ij}+
                                            \bar
                                            N_{ik}\,(\bar g_-)^{kl}\,\bar N_{jl}
                                            \bigr)\, \de\bar x^i
                                            \wedge \star\, \de\bar x^j
                                            \nonumber \\ & \hspace{2cm} - 2\,(\bar
                                            g_-)^{ik}\,\bar N_{jk}\, \De^A
                                            \bar{\tilde{x}}{}_i \wedge
                                            \star\, \de \bar x^j+(\bar g_-)^{ij}\, \De^A \bar{\tilde{x}}{}_i \wedge \star\, \De^A
   \bar{\tilde{x}}{}_j \Big) \nonumber \\ & \quad \, + \frac{1}{2}\, \int_\Sigma\, \Big(\bar \eta^{j}_i\, \de\bar x^i \wedge \De^A
   \bar{\tilde{x}}{}_j - \bar\eta^k_j\, \bar N_{ik}\, \de\bar x^i
   \wedge \de\bar x^j \Big) \ .
   \label{foliagau}\end{align}
To recover a reduced sigma-model on the leaf space, we impose the constraint obtained from the equation of motion for the worldsheet gauge field $A,$ $\delta\cS[\phi, A]/\delta A_i=0,$ which appears quadratically as an auxiliary field. It is given explicitly by
\be \label{selfduality}
\De^A\bar{\tilde{x}}{}_i= \bar N_{ki} \, \de\bar x^k + (\bar
g_-^{-1})_{il}\, \bar\eta^l_k\, \star\, \de\bar x^k\ .
\ee
By using \eqref{eq:etagen} and \eqref{eq:cHgen} we can write \eqref{selfduality} in a more covariant form as
$$
\De^A\bar\IX^I = \bar\eta^{IJ}\,\bar\cH_{JK} \, \star \, \de\bar\IX^K \ ,
$$
which for $A=0$ is the immediate generalization of the self-duality constraint written in~\cite{Hull2005,Hull2007,Hull2009}.
By substituting the constraint \eqref{selfduality} in \eqref{foliagau}, we obtain the local expression
\be\label{eq:sigmareduced}
S[\phi]= \frac{1}{2}\,\int_\Sigma\, (\bar g_+)_{ij}\, \de\bar
x^i \wedge \star\, \de\bar x^j + \int_\Sigma \, \bar N_{ik}\,\bar\eta^k_j\,
\de\bar x^i \wedge \de\bar x^j \ . 
\ee
In the quantum theory, integrating out $A_i$ in the functional integral
formally generates a determinant involving $\det(\bar g_-)=\det(\bar g_+)^{-1}$ which contributes a
Fradkin-Tseytlin term with dilaton field 
$$
\Psi = -\tfrac12\log\det(g_+)
$$
 in the sigma-model action functional; this gives the required
 generalized T-duality invariant correction to the dilaton.

It follows that the reduced sigma-model \eqref{eq:sigmareduced} is well-defined on the leaf space if the condition
$$
\pounds_{X_-} g_+=0
$$
holds for all $ X_- \in \mathsf{\Gamma}(L_-),$ so that $\cH$ is a
bundle-like metric; in other words, $\pounds_{X_-}\cH \in
\mathsf{\Gamma}(T^*M \otimes L_-^*),$ and $(M,g_+,\cF_-)$ is a
Riemannian foliation. If $\psi_i^j \in
\mathsf{\Gamma}(L_-^*),$ then the Lie derivative of $\cH$ along any
vector field in $\mathsf{\Gamma}(L_-)$ is still an element of
$\mathsf{\Gamma}(T^*M \otimes L_-^*)$  and hence the condition
\eqref{compa1} still holds, since $\iota_{X_-}\omega \in
\mathsf{\Gamma}(L_+^*),$ for all $ X_- \in \mathsf{\Gamma}(L_-).$ 

Let us focus now on the topological term of the reduced
sigma-model \eqref{eq:sigmareduced}. We would like to find conditions under which the 2-form
$$
b=N_{ik}\,\eta^k_j\, \de x^i \wedge \de x^j
$$ 
is at least locally
well-defined on the leaf space $\cQ= M/ \cF_-.$ This 2-form arises
from a local frame spanning the sub-bundle $L_+,$ so it involves the
locally defined functions $N_{ik}$ which characterize the frame for
$L_+.$ To obtain a condition involving these functions, we will
consider the transverse components to the foliation of the Lie derivative of $\omega$ along vector fields parallel to the foliation.
Following~\cite{Kotov2018}, we require the condition
$$
{}^\tau\nabla\omega=0 \ ,
$$
which implies transversal invariance of the fundamental 2-form, i.e. for all $X_-\in \mathsf{\Gamma}(L_-)$:
\be \label{topcondition}
\pounds_{X_-}\omega(Y_+, Z_+)=0\ ,  
\ee
for all $ Y_+, Z_+ \in \mathsf{\Gamma}(L_+).$ This implies that the split signature metric satisfies the invariance condition $\pounds_{X_-}\eta=0$, and the frame spanning $L_+$ which diagonalizes $K$ is composed of (local) projectable vector fields. If we write any section of $L_+$ in the local
basis diagonalizing $K$, then this condition imposes the local constraint 
$$
\pounds_{X_-}N_{ik}=0 \ .
$$ 
Then the metric \eqref{eq:etagen} coincides with the pp-wave type
split signature metrics proposed by~\cite{Cederwall:2014kxa}, proving here that
these exhaust (locally) the allowed non-constant split signature metrics
for double field theory.
It follows that then the local 2-form $b=N_{ik}\,\eta^k_j\, \de x^i \wedge \de x^j$ is well-defined on the leaf space.

Because of the isotropy of $L_-$ with respect to $\omega$ and its involutivity, it is also easy to show that
$$
\pounds_{X_-} \omega (Y_-, Z_-)=0 \ , 
$$ 
for all $Y_-, Z_- \in \mathsf{\Gamma}(L_-) .$ Together with
\eqref{topcondition} this implies that $\pounds_{X_-}\omega$, like
$\omega$, is an element of $\mathsf{\Gamma}(L_+^*\wedge L_-^*).$ It
follows that the Lie derivative of $\omega$ along any vector field
from $\mathsf{\Gamma}(L_-)$ satisfies \eqref{compa2} if the connection
coefficients of $\nabla$ satisfy $\Omega_i^j \in
\mathsf{\Gamma}(L_-^*).$ Combining this constraint with the constraint
obtained from the generalized isometry condition for $\cH$, we may
still consider the Bott connection $\nabla^{\tt B}$ as the Lie
algebroid connection on $T\cF_-$ and restrict it further to get a map
$$
\nabla^{\tt B}: \mathsf{\Gamma}(T\cF_-) \times
\mathsf{\Gamma}(T\cF_-) \longrightarrow \mathsf{\Gamma}(T\cF_-) \ ,
$$ 
so that
$\nabla^{\tt B}\pm \psi$ with $\psi \in \mathsf{\Gamma}(L_-^* \otimes
{\sf End}(T\cF_-))$ give well-defined connections on $T\cF_-.$ This
restriction of $\nabla^{\tt B}$ is well-defined because of the
properties of Bott connections. 

The 2-form $b$ on the leaf space
emerging from this description is not always globally defined; a
global construction involving the 3-form $\cK = \de \omega$ should, in principle, resemble the construction implemented in \cite{Severa2019}. A first step towards extending this gauging procedure to further include open string sigma-models can be performed following the formalism of Hamiltonian Lie algebroids~\cite{Ikeda2019}.

\medskip

\subsection{Generalized T-Duality and Non-Geometric Backgrounds} \label{specT} ~\\[5pt]
We will now discuss the role of $ {\sf O}(d,d)(M)$-transformations for Born sigma-models and how they relate to their gauging. 
We saw in Section~\ref{ContTduality} that, given an almost para-Hermitian manifold $(M, \eta, K)$ with compatible generalized metric $\cH,$ the Born structure $(\eta, K, \cH)$ is mapped into another Born structure $(\eta, K_\vartheta, \cH_\vartheta)$ on $M$ by $\vartheta \in \sfO(d,d)(M).$ Thus starting from a Born geometry which defines a Born sigma-model $\cS(\cH, \omega)$ with the target space $M,$ an $\sfO(d,d)(M)$-transformation gives another Born sigma-model $\cS(\cH_\vartheta, \omega_\vartheta)$ with the new Born structure on the same target space. This is a generalized T-duality transformation which relates two Born sigma-models.

To interpret generalized T-duality in the context of gauged Born sigma-models, we first focus on $\sfO(d,d)(M)$-transformations which relate Born structures $(\eta, K, \omega, \cH)$ and $(\eta, K_\vartheta, \omega_\vartheta, \cH_\vartheta)$ for which both $K$ and $K_\vartheta$ have at least one integrable eigenbundle, which without loss of generality we may assume corresponds to the same eigenvalue $-1$. We write $L_-$ for the integrable sub-bundle of $K$ and $L^\vartheta_-$ for the integrable sub-bundle of $K_\vartheta.$ Then the Born manifold $(M, \eta, K, \cH)$ is foliated by $\cF_-$ such that $L_-=T \cF_-$ and $(M, \eta, K_\vartheta, \cH_\vartheta)$ is foliated by $\cF^\vartheta_-$ such that $L^\vartheta_-=T\cF^\vartheta_-.$
Whenever the backgrounds defining both Born sigma-models $\cS(\cH,
\omega)$ and $\cS(\cH_\vartheta, \omega_\vartheta)$ for $M$ satisfy
the generalized isometry conditions, we can gauge both sigma-models to
reduce them to two distinct non-linear sigma-models $S(g, b)$ and
$S(g_\vartheta, b_\vartheta)$ for the leaf spaces $M/ \cF_-$ and $M/
\cF^\vartheta_-$ respectively. The reduced sigma-models are defined,
respectively, by the Riemannian metric $g$ and the 2-form $b$ induced
by $(\cH, \omega)$ on $M/ \cF_-$, and the metric $g_\vartheta$ and
2-form $b_\vartheta$ induced by $(\cH_\vartheta, \omega_\vartheta)$ on
$M /\cF^\vartheta_-.$  We say that the non-linear sigma-models $S(g,
b)$ and $S(g_\vartheta, b_\vartheta)$ recovered in this way are
\emph{T-dual} to each other. 

We may picture this prescription through the diagram
\begin{center}
\begin{tikzcd}
(M, \cH, \omega) \arrow[r, "\vartheta "] \arrow[d, "{\mathit\Pi} "'] & (M, \cH_\vartheta, \omega_\vartheta) \arrow[d, "{\mathit\Pi}_\vartheta "] \\
(M/\cF_-, g, b) \arrow[r, dashed, swap, "\ccT "] & (M/\cF_-^\vartheta, g_\vartheta, b_\vartheta)
\end{tikzcd}
\end{center}
where $\vartheta \in \sfO (d,d)(M)$, and the dashed arrow $\ccT$ is
not a map but rather the generalized T-duality relation between the
sigma-models $S(g, b)$ and $S(g_\vartheta, b_\vartheta)$ defined by
the backgrounds on the respective leaf spaces. The vertical arrows are
the Riemannian submersions ${\mathit\Pi}$ of $(M, \cH, \omega)$ onto
$(M/\cF_-, g, b),$ which is physically defined by imposing the
dynamical self-duality constraint $\delta\cS(\cH, \omega)/\delta A_i=0$,
and ${\mathit\Pi}_\vartheta$ of $(M, \cH_\vartheta, \omega_\vartheta)$ onto
$(M/ \cF^\vartheta_-, g_\vartheta, b_\vartheta)$ which is similarly
defined by the self-duality constraint $\delta \cS(\cH_\vartheta,
\omega_\vartheta)/\delta A_i=0.$ These constraints relate derivatives of
the pullback of the leaf coordinates $\tilde x_i$ to derivatives of
the pullback of the physical coordinates $x^i$ on the space of leaves as in \eqref{selfduality}, and together with the generalized isometry conditions they constitute the generalization of the strong constraint of double field theory.

The reduced sigma-models may also be used to geometrically
characterize the choice of polarization. When the leaf space $\cQ=
M/\cF_-$ is a smooth manifold and the reduced sigma-model $S(g, b)$
involves well-defined background fields on $\cQ,$ it corresponds to a
\emph{geometric background} and the corresponding polarization is
called a \emph{geometric polarization}. Otherwise, if the leaf space
$\cQ$ does not admit a smooth structure but the background fields
$(g,b)$ are still well-defined on $\cQ$, i.e. they come from a
background $(\cH, \omega)$ satisfying the generalized isometry
conditions for the Born sigma-model, we call it a \emph{locally
  geometric background}. Using the terminology of~\cite{Hull2005}, we
will refer to such leaf spaces as \emph{T-folds}, and the
corresponding polarization defining the Born sigma-model which leads
to a T-fold will be called a \emph{T-fold polarization}. In contrast
to common lore, initiated by~\cite{Shelton2005}, it is possible for
both geometric and non-geometric backgrounds in this sense to have
non-vanishing `$Q$-flux', as \eqref{eq:tildeZQ} and the general
analysis of Section~\ref{BSMGau} shows. That $Q$-flux is not
necessarily an obstruction to global geometry was also highlighted by~\cite{Schulz:2011ye}.

In a T-fold
polarization, the foliation $\cF_-$ defines a Lie subalgebroid of the
tangent algebroid $(TM,\unit_{TM})$ that is naturally integrated by
the holonomy groupoid of $\cF_-$ presenting the
space of leaves, which however is no longer a Lie
subgroupoid of the pair groupoid $M\times M\rightrightarrows
M$~\cite{Mackenzie1995}. The (singular) quotient $\cQ=M/\cF_-$ can
also be presented in a more invariant way as a
  \emph{smooth} stack, even for singular foliations
  $\cF_-$, and generalized T-duality can be realized as a morphism of
  stacks. This perspective was developed by~\cite{Bytsenko:2015uxa} in
  the more general context of stratified spaces, which include the
  orbifolds and symmetric spaces that appear in the following, while topological T-duality and T-folds are described in such a
  geometric framework by~\cite{Nikolaus2018} for the 
  polarizations obtained from torus bundles with NS--NS $H$-flux. An interesting special class of T-folds which admit a precise geometric description are the foliated Born manifolds whose leaves are compact. Since the generalized isometry condition for the compatible generalized metric 
$
\cH= g_+ + g_-
$
implies that $(M, g_+, \cF_-)$ is a Riemannian foliation, in these
cases the leaf space admits the structure of an orbifold with isotropy
group given by the leaf holonomy group, and ${\mathit\Pi}:M\to\cQ$ is an
orbifold submersion, see \cite{Molino, Boyer2007, Mrcun2003} for
further details; in this case, a more invariant description of the
orbifold $\cQ$ is as a smooth real Deligne-Mumford stack. 

The other scenario which can arise is when a generalized T-duality ${\vartheta} \in {\sf O}(d,d)(M)$ maps a Born manifold $(M,\eta, K, \cH)$ with an integrable eigenbundle $L_- \subset TM$ of $K$ into another Born manifold $(M, \eta,K_{{\vartheta}}, \cH_{{\vartheta}})$ with eigenbundle $L^{{\vartheta}}_-\subset TM$ of $K_\vartheta$ which is no longer integrable; this is the case of non-vanishing `$R$-flux' $R^{mnk}\neq0$ in \eqref{fluxes}.
In this case, there is still a well-defined Born sigma-model $\cS(\cH_{{\vartheta}}, \omega_{{\vartheta}})$ which is related to the Born sigma-model $\cS(\cH, \omega)$ 
by an ${\sf O}(d,d)(M)$-transformation. However, even if a frame
spanning $\mathsf{\Gamma}(L^{{\vartheta}}_-)$ and
$\cS(\cH_{{\vartheta}}, \omega_{{\vartheta}})$ satisfy the generalized
isometry conditions, a gauging of $\cS(\cH_{{\vartheta}},
\omega_{{\vartheta}})$ which recovers a conventional spacetime
description is not possible since $M$ is no longer a foliated manifold: there is no submersion from $(M,\cH_{{\vartheta}}, \omega_{\vartheta})$ because there is no leaf space in this case. This situation is summarized by the diagram
\begin{center}
\begin{tikzcd}
(M, \cH, \omega) \arrow[r, "\vartheta "] \arrow[d, "{\mathit\Pi} "'] & (M,\cH_{{\vartheta}}, \omega_{{\vartheta}})\arrow[d, dashed, " "] \\
(M/\cF_-, g, b) \arrow[r, dashed, swap,"\ccT "] & (\ \cdot\ ,\ \cdot\ ,\ \cdot\ ) 
\end{tikzcd}
\end{center}
The vertical dashed arrow here indicates the impossibility of recovering any conventional background, even locally. 

In this instance one could try to implement a similar version of the gauging of the Born sigma-model $\cS(\cH_{{\vartheta}}, \omega_{{\vartheta}})$ upon introducing the analogue of a covariantized map $\De^\cA \phi$ which is defined by the bundle maps
$$ 
T\Sigma \xlongrightarrow{\bar{\cA}} L^{{\vartheta}}_- \xlongrightarrow{i} TM
$$ 
where $\bar{\cA}: T\Sigma \rightarrow L^{{\vartheta}}_-$ is a bundle
map which is \emph{not} generally induced by the pullback of a Lie
algebroid connection, since $L^{{\vartheta}}_-$ is not naturally a Lie algebroid in this case. 
There is again the pullback bundle $\phi^* {\rm Im}(i \, \circ \, \bar{\cA}\,) \subset \phi^* TM$ and the induced map $\overline{i(\cA)}: T\Sigma \rightarrow \phi^* TM$ which permits us to write the analogue of the covariant derivative 
$$
\De^\cA \phi = \de \phi - \overline{i(\cA)}\ .
$$
We can associate to this map the
tensor $\underline{\De^\cA \phi} \in \mathsf{\Gamma}(T^*\Sigma \otimes \phi^*
TM).$ A ``gauged'' sigma-model can still be defined with this
map since the norm  $\|\,\underline{\De^\cA \phi}\, \|_{h,\cH}$ is
well-defined on the vector space of sections
$\mathsf{\Gamma}(T^*\Sigma \otimes \phi^* TM).$ This construction
depends on the choice of bundle map~$\bar\cA$.

To give a physical meaning to this construction, we pass to the
local picture and introduce the analogue of a covariant derivative for
only half of the coordinates; however these coordinates now have no
particular geometric significance. We can write down a self-duality
constraint $\delta\cS(\cH_{{\vartheta}},
\omega_{{\vartheta}})/\delta\cA_i=0,$ but the solution to this constraint does not eliminate the dependence of the
background fields on the ``gauged coordinates'', and is moreover expected to involve a non-local expression. This means that it is
not possible to find even a locally defined conventional background on some open subset of $M.$
Thus there is no reduced sigma-model that can be recovered, since
there is no well-defined quotient and hence no physical spacetime to
serve as a target for a reduced sigma-model. The polarization in which
this happens can thus only be described in the full doubled formalism
based on the Born manifold $(M,\cH_\vartheta,\omega_\vartheta)$; using
the terminology of~\cite{Hull:2019iuy}, we say that this polarization
is associated with an \emph{essentially doubled space}, and call
the corresponding polarization an \emph{essentially doubled polarization}.

\medskip

\subsection{Weakly versus Strongly T-Dual Sigma-Models} \label{naive} ~\\[5pt]
We will now discuss how to distinguish T-dual sigma-models based on
the geometry of the underlying foliations. For this, we stress a
further distinction amongst generalized T-duality transformations. We
may apply an ${\sf O}(d,d)(M)$-transformation $\vartheta$
preserving the foliation induced by the almost para-complex structure
and preserving the transverse metric $g_+$ to the foliation, i.e. $(M,
\cF_-, \cH, \omega)$ and $(M, \cF_-, \cH_\vartheta, \omega_\vartheta)$
are both Riemannian foliations with respect to the same metric $g_+$
and have the same leaf space $\cQ=M/\cF_-.$ In this case, the only
difference between the reduced sigma-models is given by the
topological term for the leaf space $\cQ,$ i.e. the T-dual
sigma-models are given by the backgrounds $(\cQ, g_+, b)$ and $(\cQ,
g_+, b_\vartheta).$ These sigma-models thus have the same dynamical
content, since the same transverse metric $g_+$ appears in the
background of each. In this sense, such sigma-models are \emph{weakly}
T-dual to each other; weak T-duality acts on an exact Courant algebroid,
and is physically a manifest symmetry of the low-energy effective target space supergravity theory on $\cQ$.
In contrast, \emph{strongly} T-dual sigma-models arise when applying $\sfO(d,d)(M)$-transformations which map a Riemannian foliation into a different Riemannian foliation, each associated with a different Born structure.

It follows from Remark~\ref{rem:B+} that a weak generalized T-duality transformation
is exactly a $B$-transformation. Let $(M, \cF_-, \eta, \omega, \cH)$ be a
foliated Born manifold with splitting of its tangent bundle 
$$
TM=T\cF_- \oplus L_+
$$
induced by the almost para-complex structure $K.$ The compatible generalized metric is 
$$
\cH= g_++g_-
$$ 
in this splitting, where $g_+$ and $g_-$ are fiberwise metrics on
$L_+$ and $T\cF_-$, respectively, which are related by \eqref{eq:g-}. We may think of this splitting as the
bundle map 
$$
s: N\cF_- \longrightarrow TM
$$ 
that splits the short exact sequence of vector bundles 
\be \label{shortfoli}
0 \longrightarrow T\cF_- \longrightarrow TM \longrightarrow  N\cF_- \longrightarrow 0
\ee
which is maximally isotropic with respect to the split signature metric $\eta$ and orthogonal to $T\cF_-$ in the compatible generalized metric $\cH.$ We assume that $(M, g_+, \cF_-)$ is a Riemannian foliation, therefore the associated Born sigma-model $\cS(\cH, \omega)$ can be reduced to a conventional non-linear sigma-model $S(g_+, \omega\rvert_{L_+})$ for the leaf space $\cQ,$ i.e. there exists a Riemannian submersion 
$$
{\mathit\Pi}: (M, \cH, \omega) \longrightarrow (\cQ, g_+, \omega\rvert_{L_+}) \ .
$$

A weakly T-dual sigma-model is obtained by applying a
$B_+$-transformation, which preserves $T\cF_-$ and is generated by a
basic 2-form $b_+ \in \mathsf{\Gamma}(\midwedge^2 L_+^*)$, as
discussed in Section~\ref{sec:Btransformations}. The $B_+$-transformed
Born structure is given by $(K^{B_+}, \eta, \cH^{B_+}, \omega^{B_+}),$
where $\omega^{B_+}= \omega+ 2\, b_+$ and $K^{B_+}=K+2\,B_+$ with
$B_+$ a bundle map from $L_+$ to $T\cF_-.$ It induces the splitting 
$$
TM=T\cF_-
\oplus L_+^{B_+}
$$
that can be regarded as a different choice of
splitting 
$$
s^{B_+}: N\cF_- \longrightarrow TM
$$
of the short exact sequence
\eqref{shortfoli} such that ${\rm Im}(s^{B_+})$ is still maximally
isotropic with respect to $\eta,$ which naturally follows from the fact that $s$ changes by a $B_+$-transformation.

Let us concentrate on the action of the foliation-preserving
$B_+$-transformation on the compatible generalized metric. We can
easily show that an $\cF_-$-preserving ${B_+}$-transformation is an
isometry of the fiberwise metric $g_+$ on $L_+$: any section $X_+ \in
\mathsf{\Gamma}(L_+)$ transforms as 
$$
X_+^{B_+}= e^{B_+}(X_+)= X_+ + {B_+}(X_+) \ ,
$$ 
so that
$$
g_+\big(X^{B_+}_+, Y^{B_+}_+\big)=g_+(X_+,Y_+)\ .
$$
Thus the eigenbundle $L_+^{B_+}$ of $K^{B_+}$, whose sections are of
the form  $X^{B_+}_+$, inherits the fiberwise Riemannian metric $g_+,$
and the structure of $(M, g_+, \cF_-)$ as a Riemannian foliation is
preserved. We can also show that the only effect of the ${B_+}$-transformation on the compatible generalized metric $\cH$ is to introduce a new fiberwise metric on $T\cF_-.$ In other words, the fiberwise metric $g_-$ on $T\cF_-$ does not have $L_+^{B_+}$ as its kernel:
$$
g_-\big(X_+^{B_+}, Y_+^{B_+}\big)= g_-\big({B_+}(X_+), {B_+}(Y_+)\big) \ , 
$$
and so the change of sub-bundle $L_+^{B_+}$ is associated with a change of fiberwise metric on $T\cF_-,$ now given by $g_-^{B_+}$ such that ${\rm Ker}(g_-^{B_+})= L_+^{B_+}.$
The ${B_+}$-transformed compatible generalized metric thus reads
$$
\cH^{B_+}=g_+ + g_-^{B_+}
$$
in the polarization $TM= T\cF_- \oplus L_+^{B_+}$ defined by $K^{B_+},$ i.e. by the splitting $s^{B_+}$ of \eqref{shortfoli} which is now orthogonal to $T\cF_-$ with respect to $\cH^{B_+}.$

Since the ${B_+}$-transformation preserves the Riemannian foliation and the 2-form $b_+$ is basic, we can still obtain a well-defined Riemannian submersion from $M$ to the leaf space $\cQ$ given by
$$
{\mathit\Pi}^{B_+} : (M, \cH^{B_+}, \omega^{B_+}) \longrightarrow (\cQ, g_+, \omega\rvert_{L_+}+2\,b_+) \ .
$$
Then the sigma-models $S(g_+, \omega \rvert_{L_+})$ and $S(g_+, \omega\rvert_{L_+}+2\,b_+),$ each defined with the same leaf space $\cQ$ as target space, are T-dual. They have the same dynamical content, as they are defined by the same metric $g_+$, whereas the topological term changes. In this sense they are {weakly} T-dual to each other: Classically they have the same local degrees of freedom and differ only in their global structure (which can lead to differences in the quantum theory). It follows that weakly T-dual sigma-models are classified by the cohomology of basic 2-forms. 

On the other hand, a strongly T-dual sigma-model of $S(g_+, \omega\rvert_{L_+})$ may be thought of as induced by a maximally isotropic splitting 
$$
s^\vartheta: N \cF^\vartheta_- \longrightarrow TM
$$ 
with respect to $\eta$ of the short exact sequence
$$
0 \longrightarrow T\cF^\vartheta_- \longrightarrow TM \longrightarrow  N\cF^\vartheta_- \longrightarrow 0 
$$
for $\vartheta\in{\sf O}(d,d)(M)$, which corresponds to an almost
para-Hermitian structure $(K^\vartheta, \eta, \omega^\vartheta)$ on
$M$. The compatible generalized metric $\cH^\vartheta$ decomposes as 
$$
\cH^\vartheta= g^\vartheta_++ g^\vartheta_- \ , 
$$
where $g^\vartheta_+$ is a fiberwise metric on $L^\vartheta_+$ such that $(M, g^\vartheta_+, \cF^\vartheta_-)$ is a Riemannian foliation. 

In the remainder of this paper we will illustrate the constructions of this section through several explicit examples. 

\section{Born Sigma-Models for Phase Spaces}\label{sigmacot}
A large class of examples which are well-suited to explicit
realization of the formalism of Section~\ref{BSM} come from Born
structures on fiber bundles. These examples naturally supply
Riemannian submersions from the total space $M$ to the base space
$\cQ$, which can be regarded as the smooth quotient $M / \cF$ of the
total space with respect to the foliation of the bundle given by the
fibers $\cF.$ We will also consider a different quotient, which in
Section~\ref{sec:BornDTT} will be used to give a geometric interpretation of
the prototypical T-folds in this framework. Our working example will be the cotangent bundle $T^*\cQ$, which can be thought of as the phase space for a closed string with target space $\cQ$. This nicely ties our worldsheet formalism from Section~\ref{BSM} with the old sigma-models for duality-symmetric string theory based on phase space targets~\cite{Tseytlin:1990nb} and with more recent discussions of phase spaces as instances of doubled geometry~\cite{Freidel2014,Deser2015,Aschieri2015,Lee:2015xga,Blumenhagen:2016vpb,
Samann2018,Heller2016,Aschieri2017,Jonke2018,Osten:2019ayq}.

\medskip

\subsection{Para-K\"ahler Structure on the Cotangent Bundle} \label{parakcot} ~\\[5pt]
We first recall how to define a Born structure on the cotangent bundle of any smooth manifold by specializing the general discussion of Examples~\ref{fiberexample} and~\ref{fiberborn}. Let $\cQ$ be a smooth manifold with $\mathrm{dim}(\cQ)=d.$ Its cotangent bundle is the vector bundle
\be \label{exseq}
\pi: T^*\cQ \longrightarrow \cQ \ ,
\ee
where $\pi$ is the canonical projection and the typical fiber $\cF$ is diffeomorphic to $\IR^d.$ 
Since $\pi$ is a surjective submersion, there is a short exact sequence of vector bundles
\be \label{cotshort}
0 \longrightarrow L_{\tt v}(T^*\cQ)\xlongrightarrow{i} T(T^*\cQ) \xlongrightarrow{\hat{\pi}} \pi^*(T \cQ)\longrightarrow 0 
\ee
where $L_{\tt v}(T^*\cQ)= \mathrm{Ker}(\pi_*)=\mathrm{Ker}(\hat{\pi})$ is the vertical sub-bundle defined by the differential of the projection and $\pi^*(T \cQ)$ is the pullback bundle of $T\cQ$ over $T^*\cQ$ along $\pi.$ The map $i: L_{\tt v}(T^*\cQ) \rightarrow T(T^*\cQ)$ is the canonical inclusion of the vertical vector sub-bundle into $T(T^*\cQ).$ The vertical sub-bundle is integrable and can be regarded as a tangent bundle $L_{\tt v}(T^*\cQ) \simeq T \IR^d.$ The bundle map $\hat{\pi}: T(T^*\cQ) \rightarrow \pi^*(T\cQ)$ is surjective and covers $\pi,$ since there is also a surjective submersion of $\pi^*(T\cQ)$ onto $\cQ.$

We define an almost para-complex structure on $T^*\cQ$ by choosing a splitting of the short exact sequence \eqref{cotshort}, i.e. we fix a right inverse $C$ of $\hat{\pi}$:
\be \label{spl}
%\[
\begin{tikzcd} 
T(T^*\cQ) \arrow [r,"\hat{\pi}"']
& \pi^*(T\cQ) \arrow[l,bend right,"C"'] 
\end{tikzcd}
%\]
\ee
so that $T(T^*\cQ)=\mathrm{Im}(C) \oplus \mathrm{Ker}(\hat{\pi}).$ The
sub-bundle $\mathrm{Im}(C)=L^C_{\tt h}(T^*\cQ)$ is one of the
possible choices of horizontal distribution: the map $C$ is usually
understood as a horizontal lift of sections of $T\cQ$ to $T(T^*\cQ).$
This defines an Ehresmann connection on $T^*\cQ,$ with 
$$
T(T^*\cQ)= L^C_{\tt h}(T^*\cQ) \oplus L_{\tt v}(T^*\cQ) \ .
$$
The horizontal lift $C$ thus defines a vector sub-bundle $L^C_{\tt
  h}(T^*\cQ)$ of $T(T^*\cQ),$ which is generally not involutive.   
A splitting \eqref{spl} of $T(T^*\cQ)$ is equivalent to a choice of an almost para-complex structure on $T^*\cQ$: we define the almost para-complex structure $K_C \in {\sf Aut}_\unit(T(T^*\cQ))$ by 
$$
K_C\big|_{L^C_{\tt h}(T^*\cQ)}=\unit_{L^C_{\tt h}(T^*\cQ)}\qquad
\mbox{and} \qquad K_C \big|_{L_{\tt v}(T^*\cQ)}= -\unit_{L_{\tt
    v}(T^*\cQ)} \ .
$$

The phase space $T^*\cQ$ is endowed with a canonical symplectic 2-form
$\omega_0$ with respect to which $L_{\tt v}(T^* \cQ)$ is maximally
isotropic. We may then ask whether the almost para-complex structure
$K_C$ and the canonical symplectic 2-form $\omega_0$ satisfy a
compatibility condition such that they induce a split signature metric on
$T^*\cQ.$ In other words, we may ask for conditions ensuring existence
of a split signature metric compatible with $K_C$, in the sense of almost para-Hermitian structures, such that $\omega_0$ is the corresponding fundamental 2-form. Recall that the requisite compatibility condition is
\be \label{omegacomp}
\omega_0\big(K_C(X),Y\big)+ \omega_0\big(X, K_C(Y)\big)=0\ , 
\ee
for all $X, Y \in \mathsf{\Gamma}(T(T^*\cQ)).$ It is straightforward to check that \eqref{omegacomp} holds if and only if the chosen splitting is isotropic with respect to $\omega_0$. Since $L_{\tt v}(T^*\cQ)$ is maximally isotropic because it is in the kernel of the tautological 1-form on $T^*\cQ,$ this means we have to choose $C$ such that $L^C_{\tt h}(T^*\cQ)$ is isotropic with respect to $\omega_0.$ We denote the split signature metric given by such a choice by $\eta_C$:
$$
\eta_C(X,Y)= \omega_0\big(K_C(X),Y\big)\ , 
$$ 
for all $X , Y \in \mathsf{\Gamma}(T(T^*\cQ)).$ 
We then obtain an almost para-K\"ahler structure on $T^*\cQ$ given by $(K_C, \eta_C, \omega_0).$

\begin{remark}
This construction generalizes to any fiber bundle $\pi: M \rightarrow \cQ,$ with ${\rm dim}(M)= 2 \, {\rm dim}(\cQ),$ which is endowed with a Liouville 1-form~\cite{Alekseevsky1994}. For this, consider again the short exact sequence of vector bundles \eqref{bunshort} from Example~\ref{fiberexample}. A \emph{Liouville 1-form} $\alpha \in \mathsf{\Gamma}(T^*M)$ is a horizontal 1-form on $M,$ i.e. $\iota_{X_{\tt v}}\alpha=0,$ for all $X_{\tt v} \in \mathsf{\Gamma}(L_{\tt v}(M)).$ Then the foliation given by the fibers of $M$ is Lagrangian with respect to the symplectic 2-form $\omega=\de \alpha$ 
associated with $\alpha,$ since ${\rm Ker}(\alpha)=L_{\tt v}(M).$
Any choice of an isotropic splitting $s: \pi^*(T\cQ) \rightarrow TM$ of \eqref{bunshort} with respect to $\omega$ defines an almost para-K\"ahler structure on $M.$ These structures are all diffeomorphic to those defined on the cotangent bundle of the base manifold $\cQ$ by considering the tautological 1-form as a Liouville 1-form. 
\end{remark}

\begin{example}
Let $\cQ$ be the configuration space of a dynamical system and
consider the tangent bundle $\pi: T\cQ \rightarrow \cQ$ as a carrier
space of the dynamics. The equations of motion of the system is thus defined by a
second order vector field ${\mathit\Sigma} \in
\mathsf{\Gamma}(T(T\cQ))$. A regular Lagrangian $\ccL \in
C^{\infty}(T\cQ)$ for the dynamical system, when it exists,
defines a  Liouville 1-form in the following way.

The \emph{vertical lift} $X_{\tt v}\in \mathsf{\Gamma}(L_{\tt
  v}(T\cQ))$ of a vector field $X \in \mathsf{\Gamma}(T\cQ)$ to
$\mathsf{\Gamma}(T(T\cQ))$ is the infinitesimal generator of
translations along the fibers, i.e. of the one-parameter group of
diffeomorphisms defined by $\mathbb{R} \ni t \mapsto (q, t\, X
\rvert_q) \in T_q\cQ$; this induces a map $\bar{{\tt v}}:  T\cQ \rightarrow L_{\tt v}(T\cQ).$ 
The \emph{vertical endomorphism} $\upsilon \in \mathsf{End}(T(T\cQ))$ is the bundle map given by the composition of the vertical lift and the tangent projection: $\upsilon= \bar{{\tt v}}\, \circ \, \pi_*,$ or equivalently the endomorphism of $T(T\cQ)$ which makes the diagram
\begin{center}
\begin{tikzcd}
T(T\cQ) \arrow{r}{\pi_*}  \arrow{rd}{\upsilon} 
  & T\cQ \arrow{d}{\bar{{\tt v}}} \\
    & T(T\cQ)
\end{tikzcd}
\end{center}
commute.
Then the Liouville 1-form of the Lagrangian dynamics is given
by\footnote{In this context, $\alpha_\ccL$ is usually called the
  `Cartan 1-form'.} $\alpha_\ccL= \upsilon(\de \ccL)$, since ${\rm
  Ker}(\alpha_\ccL)= {\rm Ker}(\upsilon)= L_{\tt v}(T \cQ).$ This gives
the Lagrangian symplectic 2-form $\omega_\ccL= \de \alpha_\ccL.$

The dynamical vector field ${\mathit\Sigma} \in \mathsf{\Gamma}(T(T\cQ))$
induces an isotropic splitting, with respect to $\omega_\ccL$, of the
canonical short exact sequence \eqref{bunshort} of vector bundles from
Example~\ref{fiberexample} with $M=T\cQ.$ Hence $(T\cQ , \ccL)$ admits
a para-K\"ahler structure. One can also show that the Lagrangian
$\ccL$ induces a compatible generalized metric on $T\cQ$~\cite{SzMar}. 
The symplectomorphism induced by a regular Lagrangian, given by the
Legendre transform from $T\cQ$ to $T^*\cQ,$ is a vector bundle
isomorphism covering the identity which induces a map from the dynamical para-K\"ahler structure to an isotropic splitting of \eqref{cotshort}; see~\cite{SzMar, Ferrario1990} for further details.
\end{example}

As discussed in~\cite{SzMar}, in a local description we may describe
the horizontal lift of a holonomic frame $\big\{
\frac{\partial}{\partial q^i}\big\}$ of $T\cQ,$ where $q^i$ are local
coordinates on $\cQ$ (pulled back from $\cQ$ to $T^*\cQ$ by the
projection $\pi$), by the
vector fields
$$
{\tt h}_i=C\Big(\frac\partial{\partial q^i}\Big)=
\frac{\partial}{\partial q^i} + C_{ij}\,\frac{\partial}{\partial p_j}
\ \in \ \mathsf{\Gamma}\big(L^C_{\tt h}(T^*\cQ)\big)
$$
where $(q^i,p_i)$ are local Darboux coordinates on $T^*\cQ,$ and $C_{ij}$ are smooth functions on the chosen open subset of $T^*\cQ$ defining the Darboux chart.  
This gives a local basis of sections of the horizontal
sub-bundle. Then it is straightforward to see that $L^C_{\tt
  h}(T^*\cQ)$ is maximally isotropic with respect to 
$$
\omega_0=\de
q^i\wedge\de p_i
$$ 
if and
only if $C_{ij}=C_{ji}$ is symmetric.

To define a Born sigma-model for $T^*\cQ$ we need to define a
compatible generalized metric $\cH_C$ for the almost para-K\"ahler
structure $(K_C, \eta_C, \omega_0).$ Following
Example~\ref{fiberborn}, we endow the base manifold $\cQ$ with a
Riemannian metric $g$. Then a fiberwise metric $g_+$ on $L^C_{\tt
  h}(T^*\cQ)$ is given by the pullback $g_+= \pi^* g$:
\be \label{cotglift}
g_+\big(X_{\tt h}, Y_{\tt h}\big)= g(X,Y) \ ,
\ee
where $X_{\tt h}= C(X)$ and $Y_{\tt h}=C(Y),$ for all $X, Y \in
\mathsf{\Gamma}(T\cQ),$ i.e. $X_{\tt h}$ and $ Y_{\tt h}$ are arbitrary
horizontal lifts. This gives a compatible generalized metric for the
almost para-K\"ahler structure $(K_C, \eta_C, \omega_0)$ which takes the
form 
$$
\cH_C= g_++ g_- 
$$
in the splitting induced by $K_C,$ with 
$$
g_-(X_{\tt v},Y_{\tt v})= g_+^{-1}\big(\eta_C{}^\flat(X_{\tt v}),
  \eta_C{}^\flat(Y_{\tt v})\big)
$$ 
for all
$X_{\tt v},Y_{\tt v}\in \mathsf{\Gamma}(L_{\tt
  v}(T^*\cQ))$.

\medskip

\subsection{Phase Space Born Sigma-Model and its Gauging} \label{smCTQ} ~\\[5pt] 
We now have all of the ingredients needed to write down a Born sigma-model for $T^*\cQ.$
The cotangent bundle Born sigma-model $\cS(\cH_C, \omega_0)$ is given by 
\begin{equation*}
\cS_0[\phi]=\frac{1}{4}\,\int_{\Sigma}\, \bigl( (\bar g_+)_{ij} \, \de
\bar q^i \wedge \star \, \de \bar q^j + (\bar g_-)^{ij} \, 
\bar\zeta_i \wedge \star \, \bar\zeta_j \bigr)
\end{equation*}
and 
\begin{equation*}
\cS_{\rm top}[\phi]= \frac{1}{2}\,\int_{\Sigma}\, \de\bar q^i \wedge \de\bar p_i \ ,
\end{equation*}
where $\phi$ is a map from the closed string worldsheet $\Sigma$ to the
phase space $T^*\cQ$. Here we wrote the compatible generalized metric $\cH_C$ as 
$$
\cH_C= (g_+)_{ij} \, \de q^i \otimes \de q^j + (g_-)^{ij}\, \zeta_i
\otimes \zeta_j \ ,
$$
with $\de q^i$ and
$$
\zeta_i=\de p_i - C_{ij}\, \de q^j
$$
dual 1-forms to ${\tt h}_i$ and $\frac{\partial}{\partial p_i}$
respectively, and $(g_+)_{ij}=g_{ij}$; we also used $\de
q^i\wedge\zeta_i=\de q^i\wedge\de p_i$ in writing the topological term. The topological term is defined
by the symplectic 2-form $\omega_0=- \de \alpha,$ where $\alpha$ is
the tautological 1-form on $T^*\cQ$ (in a Darboux chart,
$\alpha=p_i\,\de q^i$). Since $\omega_0$ is exact and we
assume that $\Sigma$ is closed, the topological term
vanishes. However, even in the case when $\Sigma$ has a non-empty
boundary, since $\omega_0$ does not have a component in
$\mathsf{\Gamma}(\midwedge^2 L_{{\tt v}}(T^*\cQ)^*)$ we do not expect
any topological term to arise in the reduced sigma-model. We will keep
the topological term explicit in the following to show that
this is indeed the case.

This sigma-model can be gauged, as discussed in Section \ref{BSM}, by
considering the vertical distribution $L_{\tt v}(T^*\cQ)$ as a Lie
algebroid. In this case $L_{\tt v}(T^*\cQ)$ is the Lie algebroid
of symmetries of the Born sigma-model $\cS(\cH_C,\omega_0),$ since 
$$
\pounds_{Z_{\tt v}}\,g_+=0 \ , 
$$ 
for all $ Z_{\tt v} \in \mathsf{\Gamma}(L_{\tt v}(T^*\cQ)),$ because
$\pounds_{Z_{\tt v}}X_{\tt h} \in \mathsf{\Gamma}(L_{\tt h}^C
(T^*\cQ))$ with $X_{\tt h}$ the horizontal lift of a vector field
$X\in \mathsf{\Gamma}(T\cQ),$ and because \eqref{cotglift} holds. The
gauging is also possible since $\pounds_{Z_{\tt v}}\, \omega_0$ has vanishing component in $\mathsf{\Gamma}(\midwedge^2 L_{\tt h}^C(T^*\cQ)^*).$
We introduce the connection 1-form $ A$ on $T^*\cQ$ obtained from the
Lie algebroid of generalized isometries of $\cH_C$ and $\omega_0.$
As discussed in Section \ref{BSM}, we define the covariant derivatives
$$
\De^A \bar q^i= \de\bar q^i \qquad \mbox{and} \qquad \De^A\bar p_i=
\de\bar p_i - \bar A_i\ , 
$$ 
where we covariantize only the pullback of the differential of the
leaf coordinates $p_i.$ Here we work with a Darboux chart $( q^i, p_i)$, so that $(p_i)$ are coordinates adapted to the leaves of $T^*\cQ,$ as discussed in \cite{Kotov2014, Deser2017}.

The action functional $\cS[\phi,A]$ of the resulting gauged Born sigma-model has two terms and is
given by
\begin{align*}
\cS[\phi,A]&=\frac{1}{4}\,\int_{\Sigma}\, \Big(\big((\bar g_+)_{ij}+\bar
                                              C_{im}\,(\bar
                                              g_-)^{mn}\,\bar
                                              C_{jn}\big)\,\de\bar q^i
                                              \wedge \star\, \de\bar
                                              q^j \\ & \hspace{2cm} + (\bar g_-)^{ij}\, \De^A\bar p_i \wedge \star\, \De^A\bar p_j -2\,(\bar g_-)^{ik}\,\bar C_{jk}\, \De^A\bar p_i \wedge \star\, \de\bar q^j \Big) \\ & \quad \, + \frac{1}{2}\,\int_{\Sigma}\, \de\bar q^i \wedge \De^A\bar p_i \ .
\end{align*}
Then the self-duality constraints $\delta \cS[\phi,A]/\delta A_i=0$ are given by
$$
\star\, (\bar g_-)^{ij}\,\De^A\bar p_j - (\bar g_-)^{ij}\,\bar C_{kj}\,
\star\, \de\bar q^k - \de\bar q^i=0 \ .
$$
By imposing this constraint we obtain a sigma-model for the quotient
$T^*\cQ / \IR^d \simeq \cQ$ with background given by the Riemannian
metric $g$ in which the holonomic basis of $T\cQ$ is orthonormal, i.e. the reduced sigma-model $S(g, b)$ is given by
$$
S[\phi]= \frac12\,\int_{\Sigma}\, (\bar g_+)_{ij}\,\de\bar q^i \wedge
\star\, \de\bar q^j \ ,
$$
where here $\phi$ is the harmonic map with its image projected to the leaf space $\cQ$ and $(g_+)_{ij}=g_{ij}$.
Thus the sigma-model $S(g, b)$ for $\cQ$ is characterized
by the metric $g=g_{ij}\, \de q^i \otimes \de q^j$ on $\cQ,$ which is
not surprising since the compatible generalized metric $\cH_C$ is
defined as a horizontal lift of $g$ to $T^*\cQ$. However, what was not
obvious from the start is that the reduced sigma-model has vanishing
Kalb-Ramond field $b=0,$ i.e. even starting with a topological term,
in the almost para-K\"ahler case the reduced sigma-models involve a
background with vanishing $B$-field on $\cQ$. This also means that the Riemannian submersion $(T^* \cQ, \cH_C,
\omega_0)\rightarrow (\cQ, g, b=0)$ is simply given by the bundle projection $\pi: T^* \cQ \rightarrow \cQ,$ as expected. 

The properties of the class of examples described in this section
extend more generally to arbitrary choices of compatible generalized
metric, given by a fiberwise metric $g_+ \in
\mathsf{\Gamma}(\midodot^2 L_{\tt h}^C(T^*\cQ)^*)$ such that 
$\pounds_{Z_{\tt v}}\, g_+=0,$ for all
$Z_{\tt v} \in \mathsf{\Gamma}(L_{\tt v}(T^*\cQ)).$
Although this simple class of examples gives the obvious result, it
aids in understanding how to deal with the gaugings in
general. Furthermore, we can still use it to understand how to obtain
a background for the reduced sigma-model in this case with a non-trivial $B$-field. 

\medskip

\subsection{Weakly T-Dual Sigma-Model with $B$-Field} ~\\[5pt]
In order to describe T-dual sigma-models for the background $(\cQ, g, b=0),$ we give the first simple example of the construction discussed in Section~\ref{naive} by considering a $B_+$-transformation of the Born structure 
introduced in Section~\ref{parakcot}, which is a pushforward of $(K_C, \eta_C, \cH_C)$ by $e^{B_+}.$ For this, we recall that $e^{B_+} \in \sfO(d,d)(T^*\cQ)$ is generated by a skew map $B_+:L^C_{\tt h}(T^*\cQ) \rightarrow L_{\tt v}(T^*\cQ)$ such that the new almost para-complex structure $K_C^{B_+}$ has $L_{\tt v}(T^*\cQ)$ and $L^C_{B_+}(T^*\cQ)$ as its eigenbundles, where $L^C_{B_+}(T^*\cQ)$ is the sub-bundle obtained from $L^C_{\tt h}(T^*\cQ)$ after the $B_+$-transformation as discussed in Section~\ref{sec:Btransformations}:
$$
L_{B_+}^C(T^*\cQ)= \big\{ X_{\tt h}+ B_+(X_{\tt h}) \ \big| \ X_{\tt h} \in L^C_{\tt h}(T^*\cQ) \big\} \ .
$$
Since a $B_+$-transformation preserves the vertical sub-bundle, we may think of the sub-bundle $L^C_{B_+}(T^*\cQ)$ as the horizontal distribution defining a new splitting \eqref{spl} of the short exact sequence of vector bundles \eqref{cotshort}: 
$$
T(T^*\cQ)= L^C_{B_+}(T^*\cQ)\oplus L_{\tt v}(T^*\cQ) \ , 
$$
which represents a different Ehresmann connection 
$$
C_{B_+}:\pi^* (T\cQ)\longrightarrow T(T^*\cQ) \ .
$$
Thus given an isotropic splitting of the short exact sequence \eqref{cotshort} with respect to $\omega_0$, we can obtain non-isotropic splittings by acting with $B_+$-transformations, which preserve the vertical distribution. The splittings obtained in this way are maximally isotropic with respect to the split signature metric $\eta_C.$
The associated almost para-Hermitian structure $(K^{B_+}_C, \eta_C, \omega^{B_+}_0)$ is obtained as a $B_+$-transformation of the almost para-K\"ahler structure $(K_C, \eta_C, \omega_0),$
whose fundamental 2-form is no longer the canonical symplectic 2-form on~$T^*\cQ$.

It is important to stress that $L^C_{B_+}(T^*\cQ)$ is not isotropic with respect to $\omega_0,$ therefore the splitting $C_{B_+}$ does not induce an almost para-complex structure which is compatible with $\omega_0.$ In fact, the fundamental 2-form of the $B_+$-transformed structure is not symplectic in general.  
Recall that, since $e^{B_+} \in \sfO(d,d)(T^*\cQ),$ the $B_+$-transformation preserves the metric $\eta_C,$ but not the fundamental 2-form which transforms to 
$$
\omega^{B_+}_0= \omega_0 + 2\, b_+ \ ,
$$ 
where $b_+$ is a horizontal 2-form,\footnote{$b_+$ is the pullback of a 2-form in $\Omega^2(\cQ)$ since the map $B_+$ is constant along the fibers.} i.e. $\iota_{X_{\tt v}}\,b_+=0,$ for all $X_{\tt v}\in \mathsf{\Gamma}(L_{\tt v}(T^*\cQ))$; hence $\omega^{B_+}_0$ is closed if and only if $b_+$ is closed. Thus $(T^*\cQ, K_C^{B_+}, \eta_C)$ is generally only an almost para-Hermitian manifold. 

In a local description, we may regard $L^C_{B_+}(T^*\cQ)$ as the sub-bundle locally spanned by sections 
$$
{\tt h}^{B_+}_i=C_{B_+}\Big(\frac\partial{\partial q^i}\Big)=\frac{\partial}{\partial q^i} + \big(C_{ij}+(b_+)_{ij}\big)\,\frac{\partial}{\partial p_j} \ ,
$$
where $C_{ij}+(b_+)_{ij}$ is not symmetric.
The map $B_+$ can be regarded as a tensor $\underline{B_+}\in \mathsf{\Gamma}(L^C_{\tt h}(T^*\cQ)^* \otimes L_{\tt v}(T^*\cQ))$ such that 
$$
\underline{B_+}=(b_+)_{ij}\, \de q^i \otimes \frac{\partial}{\partial p_j} \ ,
$$ 
where $\{\de q^i\}$ is the local coframe that spans $\mathsf{\Gamma}(L^C_{\tt h}(T^*\cQ)^*),$ and so is dual to $\big\{C(\frac\partial{\partial q^i})\big\}.$ Thus the horizontal 2-form $b_+$ reads 
$$
b_+= \tfrac12\,(b_+)_{ij}\, \de q^i \wedge \de q^j \ .
$$
In summary, we consider a $B_+$-transformation, which preserves the vertical sub-bundle so that it can still be regarded as the Lie algebroid of generalized isometries of the new structure. 

A $B_+$-transformation of the compatible generalized metric $\cH_C$ gives rise to a new compatible generalized metric $\cH_C^{B_+}$ as in \eqref{gemeju}. In the setting of Section~\ref{parakcot}, the horizontal lift $g_+$ of a Riemannian metric $g$ on $\cQ$ to $T^*\cQ$ is mapped into another horizontal lift $g_+^{B_+}$ of the same metric $g$ given by 
$$
g_+^{B_+}\big(C_{B_+}(X), C_{B_+}(Y)\big)= g(X,Y)\ , 
$$
for all $X, Y \in \mathsf{\Gamma}(T\cQ).$ 
In this case we can write  
$$
\cH_C^{B_+}=\big(g_+^{B_+}\big)_{ij} \, \de q^i \otimes \de q^j + \big(g_-^{B_+}\big)^{ij}\, \zeta^{B_+}_i \otimes \zeta^{B_+}_j \, ,  $$
where $(g_+^{B_+})_{ij}=g_{ij}$ and $g_-^{B_+}= \eta_C \,\circ \, (g_+^{B_+})^{-1}\, \circ \, \eta_C,$ while
$$
\zeta^{B_+}_i = \de p_i -\big(C_{ij}-(b_+)_{ij}\big)\,\de q^j
$$ 
is the local coframe spanning $\mathsf{\Gamma}(L_{\tt v}(T^*\cQ)^*),$ i.e. it is dual to $\frac\partial{\partial p_i}$ in the $B_+$-transformed polarization.

We thus obtain a new Born sigma-model $\cS\big(\cH_C^{B_+}, \omega^{B_+}_0\big)$ for $T^*\cQ$ given by
$$
\cS^{B_+}[\phi]=\frac{1}{4}\, \int_{\Sigma}\,\Big( \big(\bar g_+^{B_+}\big)_{ij}\, \de \bar q^i \wedge \star \, \de \bar q^j +\big(\bar g_-^{B_+}\big)^{ij}\, \bar\zeta^{B_+}_i \wedge \star \, \bar\zeta^{B_+}_j\Big) + \frac12\,\int_\Sigma\,\de \bar q^i \wedge \bar\zeta^{B_+}_i
$$
and it can still be gauged with respect to the generalized isometries generated by the vertical distribution. Following the same steps as in Section~\ref{smCTQ}, the generalized isometry condition
$$
\pounds_{Z_{\tt v}}\,g^{B_+}_+=0
$$
holds for all $Z_{\tt v} \in \mathsf{\Gamma}(L_{\tt v}(T^*\cQ))$ since $g_+^{B_+}$ is the horizontal lift of a Riemannian metric $g$ on $\cQ$, and the component $b_+$ of the fundamental 2-form $\omega^{B_+}_0$ must satisfy the condition
$$
\pounds_{Z_{\tt v}}\, b_+ =0\ ,
$$ 
for all $Z_{\tt v} \in \mathsf{\Gamma}(L_{\tt v}(T^*\cQ)),$ which
follows here since $b_+$ is the pullback of a 2-form on $\cQ.$ In the
case of a Born structure arising from a $B_+$-transformation one
obtains a global expression for this component of the fundamental
2-form. This is not always the case, since generally it has a local
characterization in terms of the local expression of a splitting and
therefore the induced 2-form on the leaf space, even when it is smooth, is not necessarily globally defined.  

The gauging is analogous to the gauging for the Born sigma-model
$\cS(\cH_C, \omega_0)$ described in Section~\ref{smCTQ}. In the
present case, the self-duality constraint $\delta \cS^{B_+}[\phi,A]/\delta A_i=0$ from gauging reads
$$
\star \, \big(\bar g_-^{B_+}\big)^{ij}\, \De^A\bar p_j - \star \, \big(\bar g_-^{B_+}\big)^{ij}\,\big(\bar C_{kj}+(\bar b_+)_{kj}\big)\, \de\bar q^k - \de\bar q^i=0 \ . 
$$
Then the sigma-model for $\cQ$ obtained by imposing these self-duality constraints is 
$$
S^{B_+}[\phi]= \frac{1}{2}\,\int_{\Sigma}\, \big(\bar g_+^{B_+}\big)_{ij}\, \de\bar q^i \wedge \star \, \de\bar q^j+\int_{\Sigma}\, \bar b_+ \ .
$$
Hence the reduced sigma-model $S(g, b_+)$ for the quotient $T^*\cQ /
\IR^d \simeq \cQ$ is defined by the same Riemannian metric $g=g_{ij}\,
\de q^i \otimes \de q^j$ as in the previous gauging, and Kalb-Ramond
field given by $b_+=\frac12\, (b_+)_{ij}\, \de q^i \wedge \de q^j.$ It follows
that the only effect of a $B_+$-transformation, which leaves unchanged
the integrable sub-bundle generating the generalized isometries, is to
give a new topological term for this class of sigma-models.
In summary, we have shown that the sigma-models $S(g,0)$ and $S(g, b_+)$ for $\cQ$ can be considered as weakly T-dual sigma-models in the sense of Section~\ref{naive}, i.e. they are related by a weak generalized T-duality transformation. This is our working example of weakly T-dual sigma-models.

\medskip

\subsection{Strongly T-Dual Sigma-Models} \label{sec:cotstrongTdual} ~\\[5pt]
An example of how a strongly T-dual sigma-model can be constructed is obtained by considering a foliation $\cF^\vartheta$ of $T^*\cQ$ of codimension ${\rm dim}(\cQ)$ such that $T\cF^\vartheta$ is maximally isotropic with respect to $\eta_C,$ for $\vartheta\in\sfO(d,d)(T^*\cQ)$. We then obtain the short exact sequence 
\be \label{cotfoli}
0 \longrightarrow T\cF^\vartheta \longrightarrow T(T^*\cQ) \longrightarrow N\cF^\vartheta \longrightarrow 0
\ee
and an isotropic splitting 
$$
s^\vartheta: N\cF^\vartheta \longrightarrow T(T^*\cQ)
$$ 
with respect to $\eta_C$ defines an almost para-Hermitian structure with different fundamental 2-form $\omega_0^\vartheta$ which is not the canonical symplectic 2-form. Whenever the compatible generalized metric $\cH_C^\vartheta$ induces a Riemannian foliation, we then obtain a reduced sigma-model for the leaf space $T^*\cQ/ \cF^\vartheta$ which is strongly T-dual to the natural sigma-model constructed in Section~\ref{smCTQ}.

This discussion is a natural prelude to describing Born geometries associated with Lagrangian foliations of $T^*\cQ$ with respect to $\omega_0.$ Consider a foliation $\cF^\vartheta$ of $T^*\cQ$ such that $T\cF^\vartheta$ is maximally isotropic with respect to $\omega_0.$ 
A maximally isotropic splitting of the short exact sequence \eqref{cotfoli} with respect to $\omega_0$ gives an almost para-K\"ahler structure $(K^\vartheta, \eta^\vartheta, \omega_0).$ Such a splitting has split signature metric $\eta^\vartheta$ which is in general different from $\eta_C.$ Therefore this structure cannot be obtained via an $\sfO(d,d)(T^*\cQ)$-transformation of the canonical almost para-K\"ahler structure discussed in Section~\ref{parakcot}. Moreover, any compatible generalized metric $\cH^\vartheta$ induced by a fiberwise metric $g^\vartheta_+$ giving rise to a Riemannian foliation defines a Born sigma-model which is not T-dual to the canonical Born sigma-model of Section~\ref{smCTQ}. In this sense, the distinct T-duality orbits of phase space Born sigma-models, giving rise to distinct T-duality chains, are classified by Lagrangian foliations and their allowed Riemannian foliation structures.

\medskip

\subsection{Worldsheet Description of Essentially Doubled Backgrounds} \label{Rfluxcot} ~\\[5pt]
We shall now discuss the case in which the gauging is not possible, i.e. it does not lead to any submersion from $T^*\cQ$ to any orbit space. For a phase space Born sigma-model, this happens whenever we apply a $B_-$-transformation. For this, let $(K_C, \eta_C)$ be an almost para-K\"ahler structure on $T^*\cQ$ which is compatible with the canonical symplectic 2-form $\omega_0,$ obtained as discussed in Section~\ref{parakcot}. Consider the automorphism $e^{B_-} \in {\sf O}(d,d)(T^*\cQ)$ covering the identity that is generated by a skew map $B_-: L_{\tt v}(T^*\cQ) \rightarrow L^C_{\tt h}(T^*\cQ),$ which can be regarded as a tensor $\underline{B_-} \in \mathsf{\Gamma}(L_{\tt v}(T^*\cQ)^* \otimes L^C_{\tt h}(T^*\cQ)),$ as discussed in Section~\ref{sec:Btransformations}. As a section of this tensor bundle, in a local coordinate chart $\underline{B_-}$ takes the form
$$
\underline{B_-} = (\beta_-)^{ij}\, \zeta_i \otimes C\Big(\frac\partial{\partial q^j}\Big) \ ,
$$ 
where $\{ \zeta_i \}$ is a local basis of vertical 1-forms and $\big\{ C(\frac\partial{\partial q^i}) \big\}$ is a local basis of horizontal vector fields, as described in Sections~\ref{parakcot} and~\ref{smCTQ}. The $B_-$-transformation determines a horizontal bivector $\beta_-$ which in local coordinate form reads
$$
\beta_- = \frac12\,(\beta_-)^{ij}\, C\Big(\frac\partial{\partial q^i}\Big)\wedge C\Big(\frac\partial{\partial q^j}\Big) \ .
$$

We then obtain a new almost para-Hermitian structure by pushing forward $(K_C, \eta_C)$ by $e^{B_-}.$ The new almost para-complex structure is given by 
$$
K_C^{B_-}= e^{B_-}\circ K_C\circ e^{-B_-}= K_C + 2\,B_- \ ,
$$ 
while the split signature metric $\eta_C$ is preserved. The
fundamental 2-form is no longer closed in general and becomes the
2-form $\omega_0^{B_-}$, where
\begin{align}\label{eq:omega0B-}
\omega_0^{B_-}(X,Y) = \omega_0(X,Y)+2\, \beta_-\big(\eta_C^\flat(X), \eta_C^\flat(Y)\big)
\end{align}
for all $X,Y\in\mathsf{\Gamma}(T(T^*\cQ))$.
The eigenbundles of $K_C^{B_-}$ are given by $L^C_{\tt h}(T^*\cQ),$ i.e. the horizontal sub-bundle remains unchanged, and $L_{B_-}(T^*\cQ)$ as a deformation of the vertical distribution, which is no longer a vertical sub-bundle. In fact, we may regard (local) sections of $L_{B_-}(T^*\cQ)$ as spanned by vector fields of the form 
$$
P^i = \frac{\partial}{\partial p_i}+ (\beta_-)^{ij}\,C_{jk}\, \frac{\partial}{\partial p_k} + (\beta_-)^{ij}\,\frac{\partial}{\partial q^j}\ ,
$$ 
obtained from the local span of the vertical distribution $L_{\tt
  v}(T^*\cQ)$ via $e^{B_-}.$ The transformation $e^{B_-}$ induces a
new splitting of $T(T^*\cQ),$ but it is no longer given by a choice of
an Ehresmann connection on $T^*\cQ,$ i.e. by a splitting of the short
exact sequence \eqref{cotshort}. 

Despite the fact that the eigenbundle $L^C_{\tt h}(T^*\cQ)$ remains
the same, the new complementary sub-bundle (which is maximally isotropic with respect to $\eta_C$) is no longer vertical. This can be easily seen by considering the pullback $\pi^* f \in C^{\infty}(T^*\cQ)$ of any function $f \in C^\infty(\cQ),$ and noticing that  
$$
\pounds_X\,\pi^*f \neq 0\ ,
$$ 
in general for $ X \in \mathsf{\Gamma}(L_{B_-}(T^*\cQ)).$ Therefore
there is still an eigenbundle decomposition
$$
T(T^*\cQ)=L^C_{\tt h}(T^*\cQ)\oplus L_{B_-}(T^*\cQ) \ ,
$$ 
but this splitting does not arise from the fiber bundle structure of
$T^*\cQ$; in other words, the transformation $e^{B_-}$ does not map a
horizontal distribution into another horizontal distribution, so it
does not preserve the vertical distribution. Hence in this case it no
longer makes sense to distinguish the eigenbundles of $K^{B_-}_C$ as
vertical and horizontal. Furthermore, neither of the eigenbundles is
integrable in general, so $T^*\cQ$ does not generally admit any foliation associated with the almost para-Hermitian structure $(K_C^{B_-}, \eta_C).$ This is the reason for refering to this polarization as an `essentially doubled polarization': It exhibits an unavoidable obstruction to obtaining a reduced sigma-model associated with the Born sigma-model in this polarization. 

To describe the Born sigma-model associated with this polarization, we consider the generalized metric $\cH_C^{B_-}$ which is compatible with the pullback almost para-Hermitian structure $(K_C^{B_-}, \eta_C)$ obtained via $e^{B_-}$ from the generalized metric $\cH_C.$ We then obtain 
$$
\cH_C^{B_-}= \big(g_+^{B_-}\big)_{ij}\, \theta^i \otimes \theta^j + \big(g_-^{B_-}\big)^{ij}\,\lambda_i \otimes \lambda_j \ ,
$$
where 
$$
\lambda_i = \de p_i - C_{ik}\, \de q^k
$$ 
is the coframe dual to $\{P^i\}$ and 
$$
\theta^i=(\beta_-)^{ij}\, \de p_j +\big(\delta^i{}_k - (\beta_-)^{ij}\,C_{jk}\big)\, \de q^k
$$ 
is the dual coframe to $\big\{C(\frac\partial{\partial q^i})\big\}$ in the splitting given by $K_C^{B_-}.$ In this coframe, the fundamental 2-form \eqref{eq:omega0B-} can be written as 
$$
\omega_0^{B_-}=  \theta^i \wedge \lambda_i \ .
$$
We now have all the data needed to write down the Born sigma-model action functional in the essentially doubled polarization:
\be \label{rfluxsigma}
\cS^{B_-}[\phi]=\frac{1}{4}\,\int_{\Sigma}\,\Big(\big(\bar
g_+^{B_-}\big)_{ij}\, \bar\theta^i \wedge \star \, \bar\theta^j +
\big(\bar g_-^{B_-}\big)^{ij}\, \bar\lambda_i \wedge \star \,
\bar\lambda_j\Big) + \frac{1}{2}\,\int_{\Sigma}\, \bar\theta^i \wedge\bar \lambda_i \ ,
\ee
where $\phi$ is a map from the closed string worldsheet $\Sigma$ to the phase space $T^*\cQ.$ 
It is clear that we cannot obtain from the sigma-model $\cS(\cH_C^{B_-}, \omega_0^{B_-})$ any reduced sigma-model associated with the almost para-Hermitian structure $(K_C^{B_-},\eta_C)$, since this structure does not yield any foliation. 

Neverthess, let us still attempt to follow the discussion of Section \ref{BSM}, and in particular our interpretation of generalized T-duality from Section~\ref{specT}. We may try to force the gauging with respect to the $-1$-eigenbundle $L_{B_-}(T^*\cQ)$ of $K_C^{B_-},$ which cannot be regarded as a Lie algebroid. Despite the fact that the Lie derivatives $\pounds_{P^k}\, \cH^{B_-}_C$ and $\pounds_{P^k}\, \omega_0^{B_-}$ can still be cast in the form of the conditions for a generalized isometry, there are no leaf coordinates whose pullback differentials can be covariantized. Thus in order to introduce some kind of gauging, we might try to covariantize $\de\bar p_i$ by minimal coupling to a 1-form $\cA$ on $\Sigma$ which is neither 
valued in a Lie algebra nor in a Lie algebroid. In principle, $\cA$ should be valued in the vector sub-bundle $L_{B_-}(T^*\cQ)$ and might be interpreted as induced by a vector bundle morphism $\bar{\cA}: T\Sigma \rightarrow L_{B_-}(T^*\cQ)$ which covers $\phi:\Sigma\to T^*\cQ$ giving the pullback bundle $\phi^*{\rm Im}(i \, \circ \, \bar{\cA})$ on $\Sigma,$ a vector sub-bundle of $\phi^*T(T^*\cQ),$ and an associated tensor $\underline{\bar\cA} \in \mathsf{\Gamma}(T^*\Sigma \, \otimes \, \phi^*T(T^*\cQ)).$ 
Here the choice of $p_i$ as ``gauged'' coordinates is arbitrary, since they do not have a geometric meaning in this polarization as leaf coordinates. Despite this, we introduce the ``covariant derivatives''
\be\label{rfluxsigmacot}
\De^\cA\bar q^i=  \de\bar q^i \qquad \mbox{and} \qquad \De^\cA\bar p_i = \de\bar p_i - \bar\cA_i\ ,
\ee
and write down the ``gauged'' sigma-model $\cS^{B_-}[\phi, \cA]$ in the usual way by replacing $\de\bar p_i$ and $\de\bar q^i$ with the maps in \eqref{rfluxsigmacot}.

We obtain the self-duality constraint as the equation of motion for $\cA$ by imposing $\delta\cS^{B_-}[\phi, \cA]/\delta\cA_i=0,$ which reads
\begin{align*}
&\Big(\big(\bar g_+^{B_-}\big)_{ij}\, (\bar\beta_-)^{im}\,(\bar\beta_-)^{jn}+ \big(\bar g_-^{B_-}\big)^{mn}\Big)\,\star\, \De^\cA\bar p_n \\ & \hspace{3cm} + \Big(\big(\bar g_+^{B_-}\big)_{ij}\,(\bar\beta_-)^{im}\,\big(\delta^j{}_l- (\bar\beta_-)^{jk}\,\bar C_{kl}\big)-\big(\bar g_-^{B_-}\big)^{mj}\,\bar C_{jl}\Big)\,\star\, \de\bar q^l \\
& \hspace{6cm} - (\bar\beta_-)^{mn}\,\De^\cA\bar p_n + \big(2\,(\bar\beta_-)^{mi}\,\bar C_{il} - \delta^m{}_l\big)\,\de\bar q^l=0 \ .
\end{align*}
These equations are solved formally by the non-local expression
\begin{align*}
\De^\cA \bar p_k = \bigg(\frac1{\bar g_-^{B_-}-\bar\beta_-\,\bar g_+^{B_-}\,\bar\beta_- - \bar\beta_-\,\star}\bigg)_{km} \, \Big(& \big(\bar g_+^{B_-}\big)_{ij}\,(\bar\beta_-)^{im}\,\big(\delta^j{}_l- (\bar\beta_-)^{jk}\,\bar C_{kl}\big) \, \de\bar q^l \\ & - \big(\bar g_-^{B_-}\big)^{mj}\,\bar C_{jl} \, \de\bar q^l+\big(2\,(\bar\beta_-)^{mi}\,\bar C_{il} - \delta^m{}_l\big)\,\star\, \de\bar q^l \Big)
\end{align*}
and substitution into the gauged extension of \eqref{rfluxsigma} gives an action functional with the $\de\bar p_n$ dependence removed. However,
even in the simplest instances where $\beta_-$ and $C$ guarantee that the sub-bundle $L_{B_-}(T^*\cQ)$ is non-integrable, the ``reduced'' local action functional still involves both sets of coordinates $q^i$ and $p_i.$ In other words, the worldsheet formulation for an essentially doubled polarization does not permit the writing of any reduced sigma-model, not even in a local form, which has dependence on only half of the coordinates. In this sense the Born sigma-model itself is needed to describe string theory on the essentially doubled background.

\section{Born Sigma-Models for Doubled Groups} \label{DoubledGroups}
Another broad class of examples of Born geometries come from Lie
groups which can be endowed with an almost para-Hermitian structure,
and their discrete quotients. In particular, cotangent bundles of Lie groups
(and their discrete quotients) furnish natural examples of doubled groups which
can be nicely combined with the phase space formalism of
Section~\ref{sigmacot}. Worldsheet theories for these types of doubled
geometries are discussed
in~\cite{Hull:2007jy,DallAgata:2008ycz,Hull2009,ReidEdwards2009,Avramis:2009xi,
  edwards:nonpub,
  Schulz:2011ye,Klimcik:2015gba,Lust:2018jsx,mpv,
  Hassler:2019wvn,Chaemjumrus:2019fsw}. In particular, doubled groups
provide concrete examples where both the exact Courant algebroid and
doubled geometry descriptions of the string background are understood,
and the connections between them were described
by~\cite{edwards:nonpub} in a similar spirit to the framework of the
present paper. These doubled sigma-models are defined using the
natural left-invariant metric and 3-form on the group manifold, so that
quotients by left-acting discrete subgroups of the doubled group can be treated using the
standard isometric gauging techniques reviewed in
Section~\ref{sec:standardgauging}. Particular changes of polarization
of doubled groups are related to non-abelian T-duality and some
aspects of Poisson-Lie T-duality; double field
theory on these sorts of extended spacetimes was formulated
in~\cite{Blumenhagen:2015zma,Bosque:2015jda}, and in this context Poisson-Lie
T-duality for Drinfel'd double groups was studied in~\cite{Hassler:2017yza,Demulder:2018lmj,Sakatani:2019jgu}. In this section we shall re-examine
gauged sigma-models for doubled groups (and their discrete quotients) from the
perspective of the Lie algebroid gauging of Born sigma-models
developed in Section~\ref{BSM}, and hence provide a more intrinsic
geometric description of them. 

\medskip

\subsection{Invariant Para-Hermitian Structures on Lie Groups} \label{sec:parahermLie} ~\\[5pt]
We begin by describing the para-Hermitian geometry on a Lie group which is invariant under the action of the group on itself. We refer to such groups as `doubled groups'.
\begin{definition}
A Lie group $\sfD$ of even dimension $2d$ is a \emph{doubled group} if
it is endowed with a left-invariant almost para-Hermitian
structure\footnote{One can also define a doubled group with a
  right-invariant almost para-Hermitian structure.} $(K^{\textrm{\tiny
    L}}, \eta^{\textrm{\tiny L}}),$
so that
$$
{\tt L}_\gamma^*\eta^{\textrm{\tiny L}} = \eta^{\textrm{\tiny L}} \qquad \mbox{and} \qquad {{\tt L}_\gamma}_*\circ K^{\textrm{\tiny L}} = K^{\textrm{\tiny L}}\circ
{{\tt L}_\gamma}_*
$$ 
for all $\gamma\in\sfD$, where ${\tt L}_\gamma:\sfD\to\sfD$ is
the map induced by left multiplication of elements of $\sfD$ with $\gamma$.
\end{definition}

Let $T_M$, with $M=1,\dots,2d$, be generators of the Lie algebra $\frd={\sf
  Lie}(\sfD)$, with the brackets
\be \label{eq:sfDLiealg}
[T_M,T_N] = t_{MN}{}^P\, T_P \ .
\ee
This preserves a constant ${\sf O}(d,d)$-invariant
metric defined by $\eta^{\textrm{\tiny L}}$, and so a doubled group is
a $2d$-dimensional subgroup of ${\sf
  O}(d,d)$. The polarization defined by the left-invariant almost para-complex
structure $K^{\textrm{\tiny L}}$ splits the generators $T_M$ into two
sets $T_m$ and $\tilde T^m$, with $m=1,\dots,d$, such that the Lie
algebra \eqref{eq:sfDLiealg} takes the form
\begin{align*}
[T_m,T_n] &= f_{mn}{}^p\,T_p + H_{mnp}\,\tilde T^p \ , \\[4pt]
[\tilde T^m,T_n] &= f_{np}{}^m\,\tilde T^p - Q_n{}^{mp}\, T_p \ ,
  \\[4pt]
[\tilde T^m,\tilde T^n] &= Q_p{}^{mn}\, \tilde T^p + R^{mnp}\, T_p \ ,
\end{align*}
with constant fluxes $H_{mnp}$, $f_{mn}{}^p$, $Q_p{}^{mn}$ and
$R^{mnp}$. The Jacobi identity for the Lie brackets
\eqref{eq:sfDLiealg} implies a set of algebraic Bianchi identities for
the generalized fluxes which can be found in e.g.~\cite{SzMar}.

Corresponding to $T_M$ there is a global frame of left-invariant vector
fields $Z_M$ on $\sfD$ which trivialize the tangent bundle
$T\sfD\simeq\sfD\times\IR^{2d}$ and generate the right action of
$\sfD$ on itself; they generate the Lie algebra \eqref{eq:sfDLiealg}
with respect to the Lie bracket of vector fields. The left-invariant
Maurer-Cartan one-forms $\Theta^M$, dual to the
left-invariant vector fields $Z_M$, form a global coframe trivializing
the cotangent bundle $T^*\sfD$ which satisfy the Maurer-Cartan equations
$$
\de\Theta^M+\tfrac12\,t_{NP}{}^M\, \Theta^N\wedge\Theta^P=0 \ .
$$
The polarization selects a splitting of these bases as $Z_M=(Z_m,\tilde Z^m)$
and $\Theta^M=(\Theta^m,\tilde\Theta_m)$. The left-invariant almost
para-Hermitian structure $(K^{\textrm{\tiny L}}, \eta^{\textrm{\tiny
    L}})$ can then be expressed in terms of this global frame and
coframe as
$$
\underline{K^{\textrm{\tiny L}}} = \Theta^m\otimes Z_m-\tilde\Theta_m\otimes\tilde
Z^n \qquad \mbox{and} \qquad \eta^{\textrm{\tiny L}} =
\Theta^m\otimes\tilde\Theta_m+\tilde\Theta_m\otimes\Theta^m \ ,
$$
and the corresponding fundamental 2-form is
\be \label{eq:2formdoubledgroup}
\omega^{\textrm{\tiny L}}= \Theta^m\wedge\tilde\Theta_m \ ,
\ee
with field strength
\begin{align}
\cK^{\textrm{\tiny L}} = \de\omega^{\textrm{\tiny L}} = \tfrac12\,\big( & H_{mnp}\,\Theta^m\wedge\Theta^n\wedge\Theta^p
            + f_{mn}{}^p\,\Theta^m\wedge\Theta^n\wedge\tilde\Theta_p
\nonumber \\ &  -
       Q_m{}^{np}\,\Theta^m\wedge\tilde\Theta_n\wedge\tilde\Theta_p +
       R^{mnp}\,
       \tilde\Theta_m\wedge\tilde\Theta_n\wedge\tilde\Theta_p \big) \ .
\label{eq:cKgroup} \end{align}

We will now introduce a suitable notion of left-invariant generalized
metric.
\begin{definition}
A \emph{generalized metric} on a doubled group $\sfD$ is an
automorphism $I^{\textrm{\tiny L}}\in{\sf Aut}_\unit(T\sfD)$ such that
$(I^{\textrm{\tiny L}})^2=\unit$, $I^{\textrm{\tiny L}}\neq\pm\,\unit$
and $I^{\textrm{\tiny L}}\circ{{\tt L}_\gamma}_*={{\tt L}_\gamma}_*\circ
I^{\textrm{\tiny L}}$ for all $\gamma\in\sfD$, which defines a
left-invariant Riemannian metric $\cH^{\textrm{\tiny L}}$ by
$$
\cH^{\textrm{\tiny L}}(X,Y) = \eta^{\textrm{\tiny
    L}}\big(I^{\textrm{\tiny L}}(X),Y\big) 
$$
for all $X,Y\in\mathsf{\Gamma}(T\sfD)$.
\end{definition}
An example is provided by the generalized metric \eqref{eq:sl2Cmetric}
on $\sfD={\sf SL}(2,\IC)$ from Example~\ref{ex:sl2c}; see
also~\cite{Hassler:2019wvn} for a two-parameter family of almost
para-Hermitian structures in this case which arise in the context of
integrable deformations of the principal chiral model.

A left-invariant Born geometry is a compatible generalized metric on
$\sfD$ which is determined in the usual way by choosing a
left-invariant fiberwise metric $g_+$ on $L_+$ (or
$g_-$ on $L_-$), where $T\sfD=L_+\oplus L_-$ is the splitting induced by
$K^{\textrm{\tiny L}}$. We will often work with the simplest example
of a Born metric on
$\sfD$ which is constructed from the left-invariant 1-forms as
\be \label{eq:genmetricdoubledgroup}
\cH^{\textrm{\tiny L}} = \delta_{MN}\, 
\Theta^M\otimes\Theta^N = \delta_{mn}\,\Theta^m\otimes\Theta^n +
\delta^{mn}\,\tilde \Theta_m\otimes\tilde \Theta_n \ .
\ee
This is the unique left-invariant Riemannian metric on $\sfD$ in which
the selected frame $\{Z_M\}$ is orthonormal.

\medskip

\subsubsection{Matrix Lie Groups} \label{sec:matrixgroups} ~\\[5pt]
To make contact between our framework and previous treatments of the geometry of
doubled groups in the literature, as well as to work out some explicit examples, we will
now specialise to the case that $\sfD$ is a matrix group. Then the
Maurer-Cartan 1-forms are given by
$$
\Theta = \gamma^{-1}\, \de\gamma = \Theta^M\,T_M
$$
where $\gamma\in\sfD$.
In a neighbourhood of the identity, we can introduce local coordinates
$$
\IX^M=(x^m,\tilde x_m)
$$ 
on the group manifold of $\sfD$ by using the
polarization to write a general group element $\gamma\in\sfD$ through
the exponential parameterization
$$
\gamma(\IX) = \tilde \sigma(\tilde x) \, \sigma(x)
$$
where
$$
\sigma(x)=\exp\big(x^m\, T_m\big) \qquad \mbox{and} \qquad \tilde\sigma(\tilde
x)=\exp\big(\tilde x_m\,\tilde T^m\big) \ .
$$
So far we have not said anything about the integrability of the
eigenbundles of $K^{\textrm{\tiny L}}$, and this splitting of
coordinates can always naturally be made for a doubled group. With it
we can express the global Maurer-Cartan 1-forms $\Theta^M$ on $\sfD$
as local $C^\infty(\sfD)$-linear combinations of the holonomic basis
$\de\IX^M$. 

The Born geometry of the doubled group $\sfD$ may then be
expressed in this parameterization by
following~\cite{Hull2009,Schulz:2011ye} to introduce the $\frd$-valued
1-form
\be \label{eq:XiTheta}
\Xi = \sigma\,\Theta\,\sigma^{-1} =
\tilde\sigma^{-1}\,\de(\tilde\sigma\,\sigma)\,\sigma^{-1} =
\de\sigma\, \sigma^{-1} + \tilde\sigma^{-1}\, \de\tilde\sigma \ .
\ee
Using the polarization we can expand the $\frd$-valued 1-forms on the
right-hand side of \eqref{eq:XiTheta} as
$$
\de\sigma\,\sigma^{-1} = \varrho^m\,T_m  + \varrho_m\,\tilde T^m
\qquad \mbox{and} \qquad \tilde\sigma^{-1}\,\de\tilde\sigma =
\tilde\ell^m\,T_m + \tilde\ell_m\,\tilde T^m \ .
$$
The component form $\Xi=\Xi^M\,T_M$ is thus given by
$$
\Xi^M = (p^m,\tilde q_m)
$$
where
$$
p^m = \varrho^m + \tilde\ell^m \qquad \mbox{and} \qquad \tilde
q_m = \varrho_m + \tilde\ell_m \ .
$$
The inverse of this change of coframe, $\Theta=\sigma^{-1}\,\Xi\,\sigma$, is given by
$$
\Theta^M = \cE_N{}^M(x)\,\Xi^N \ ,
$$
where $\cE_N{}^M(x)$ depends only on the local coordinates $x^m$ 
and is given by the adjoint action of $\sigma^{-1}(x)$ on the Lie
algebra $\frd$. The adjoint action preserves the split signature metric
$\eta^{\textrm{\tiny L}}$ and so $\cE(x)\in\sfO(d,d)$ for each
$x$. Similarly to the discussion in Remark~\ref{rem:SOdd}, we may
parameterize it with respect to the splitting of $T^*\sfD$ associated to the
almost para-complex structure $K^{\textrm{\tiny L}}$ as
$$
\cE = \begin{pmatrix}
e & e\, \beta \\
e^{-1}\, b & e^{-1}
\end{pmatrix}
$$
where $e(x)\in{\sf GL}(d,\IR)$ while $b(x)$ and $\beta(x)$ are skew-symmetric $d\times d$
matrices which depend only on the local coordinates $x^m$. 

The fundamental 2-form \eqref{eq:2formdoubledgroup} can then expressed in the coframe $\Xi^M$ as
$$
\omega^{\textrm{\tiny L}} = \tfrac12\,\hat\omega_{MN}(x) \, \Xi^M\wedge\Xi^N
$$
where 
\be \label{eq:hatomega}
\hat\omega = \begin{pmatrix}
2\,b & \mathds{1}+b\,\beta \\
-\mathds{1}-b\,\beta & -2\,\beta
\end{pmatrix} \ ,
\ee
while the compatible generalized metric \eqref{eq:genmetricdoubledgroup} can be written as
$$
\cH^{\textrm{\tiny L}} = \hat\cH_{MN}(x) \, \Xi^M\otimes\Xi^N
$$
where
\be \label{eq:hatcH}
\hat\cH = \begin{pmatrix}
g-b\,g^{-1}\,b & g\,\beta-b\,g^{-1} \\
-\beta\,g+g^{-1}\,b & g^{-1}-\beta\,g\,\beta 
\end{pmatrix} \ .
\ee
Here $g(x)=e(x)^{\rm t}\,e(x)$ is a symmetric non-degenerate $d\times d$ matrix depending only on the local coordinates $x^m$.

If the $R$-flux $R^{mnp}$ vanishes, then $\tilde T^m$ generate a $d$-dimensional
subgroup $\sfG\subset\sfD$ with the Lie algebra
$$
[\tilde T^m,\tilde T^n] = Q_p{}^{mn}\, \tilde T^p \ .
$$
In this case $\tilde\ell^m=0$, so that $p^m=\varrho^m$, and $\tilde\sigma^{-1}\,\de\tilde\sigma$ gives the left-invariant Maurer-Cartan 1-forms $\tilde\ell_m$ on $\sfG$. 
The Lie group $\sfG$ gives a maximally isotropic foliation of the doubled group $\sfD$
and we can analyse the generalized isometry conditions which enable
the gauging of the corresponding Born sigma-model. This reduces to a
non-linear sigma-model for a conventional geometric background
$\sfD/\sfG$ with local coordinates $x^m$, and will be studied in
detail in Sections~\ref{sec:Manin} and~\ref{sec:GroupBorn}. 

If the
$R$-flux is non-zero, then the generators $\tilde T^m$ do not close a
Lie subalgebra of $\frd$. In this instance there is no foliation and any ``gauging''
of the Born sigma-model will yield a reduced sigma-model description
that depends explicitly on both sets of local coordinates $x^m$ and
$\tilde x_m$, so that there is no interpretation in terms of a
conventional $d$-dimensional spacetime. The resulting background is
therefore essentially doubled.

\medskip

\subsubsection{Quotienting by a Discrete Group} ~\\[5pt]
When a Lie group $\sfG$ is non-compact, a Scherk-Schwarz dimensional
reduction~\cite{Scherk:1979zr} on $\sfG$ does not give a proper compactification. In order to
lift it to string theory one should introduce a discrete cocompact
subgroup $\sfG(\IZ)\subset\sfG$ and consider instead compactification
on the compact space $\sfG/\sfG(\IZ)$~\cite{Hull2005a}. If $\sfG$ foliates a doubled
group $\sfD$, then taking the quotient by a discrete cocompact
subgroup $\sfD(\IZ)\subset\sfD$ gives a compact manifold
$M=\sfD/\sfD(\IZ)$, where $\sfG(\IZ)\subset\sfD(\IZ)$ acts only on
$\sfG$ and leaves the leaf space $\sfD/\sfG$ invariant. Thus the doubled group construction in string theory is
restricted to Lie groups which admit a discrete cocompact
subgroup. A widely studied class of examples are the nilpotent Lie
groups which can be defined over the rationals, and
taking the quotient by a discrete cocompact subgroup gives a compact
nilmanifold; we will study in detail an example from this class in Section~\ref{sec:BornDTT}. Generally, writing the left-invariant 1-forms in the
holonomic basis as 
$$
\Theta^I=E^I{}_J\,\de\IX^J
$$ 
identifies the
Scherk-Schwarz twist matrix $E=\big(E^I{}_J \big)\in{\sf GL}(2d,\IR)$ in this formalism.

Taking the subgroup $\sfD(\IZ)$ to have a left action
$L_\xi:\sfD\to\sfD$ for all $\xi\in\sfD(\IZ)$, then the left-invariant almost para-Hermitian
structure $(K^{\textrm{\tiny L}}, \eta^{\textrm{\tiny L}})$ and
compatible generalized metric $\cH^{\textrm{\tiny L}}$ descend to a
well-defined almost para-Hermitian
structure $(K, \eta)$ and
compatible generalized metric $\cH$ on the quotient
$M=\sfD/\sfD(\IZ)$. The group of large diffeomorphisms ${\sf
  Diff}(M;\IZ)$ is the automorphism group ${\sf
  Aut}\big(\sfD(\IZ)\big)$ of the lattice $\sfD(\IZ)$, and in the quantum theory physical T-duality transformations will then
live in a subgroup of the discrete group
\be \label{eq:sfOAut}
\sfO(d,d)(\sfD)\,\cap\,{\sf Aut}\big(\sfD(\IZ)\big)
\ee
of automorphisms of the doubled group $\sfD$ that preserve $\sfD(\IZ)$ and the split signature metric $\eta^{\textrm{\tiny L}}$.
For example, when $\sfD=\IR^{2d}$ then $\sfD(\IZ)=\IZ^{2d}$ with
$M=\sfD/\sfD(\IZ)=\IT^{2d}$ and ${\sf
  Aut}(\IZ^{2d})={\sf GL}(2d,\IZ)$, so that \eqref{eq:sfOAut} is the
T-duality group $\sfO(d,d)\cap{\sf GL}(2d,\IZ)=\sfO(d,d;\IZ)$ of string
  theory on a $d$-dimensional toroidal compactification.

Gauging the generalized isometry generated by the
vector fields
$\tilde Z^m$ in the corresponding Born sigma-model then gives a
conventional reduced sigma-model for the quotient space
$M/\sfG$. However, in contrast to the Born structure on the doubled group
manifold $\sfD$, where the quotient $\sfD/\sfG$ always yields a
geometric background, the geometric nature of the quotient $M/\sfG$ depends on
the way in which the subgroups $\sfG$ and $\sfD(\IZ)$ are embedded
into $\sfD$~\cite{Hull2009,ReidEdwards2009}. If the subgroup $\sfG\subset\sfD$ commutes with the
action of $\sfD(\IZ)$, so that
$$
\xi\,\sfG\,\subseteq\,\sfG\,\xi \ ,
$$
for $\xi\in\sfD(\IZ)$, then the quotient space $M/\sfG$ is smooth and
describes a conventional geometric background. On the other hand, if
the subgroup $\sfG$ does not commute with $\sfD(\IZ)$, then the quotient
space $M/\sfG$ is not smooth and the resulting background is a T-fold. In
Section~\ref{sec:BornDTT} we shall study a concrete example which
illustrates all of these general features explicitly from a different
geometric point of view.

\medskip

\subsection{Manin Pairs and Drinfel\rq{}d Doubles} \label{sec:Manin} ~\\[5pt]
We will now consider some general Lie algebraic
structures that naturally lead to doubled groups. We demonstrate how doubled groups arise from Manin pairs,
and also Manin triples such as Drinfel\rq{}d doubles which generalize
Example~\ref{ex:sl2c}. Let us start by providing some definitions
which will be central to the rest of this paper, following
\cite{Jurco2018}.
\begin{definition}
A \emph{Manin pair} $({\frd}, {\frg})$ is a $2d$-dimensional Lie algebra ${\frd}$ endowed with an invariant symmetric non-degenerate pairing $\braket{\, \cdot \, , \, \cdot \,}$ of signature $(d,d)$, together with a Lie subalgebra ${\frg}\subset\frd$ which is maximally isotropic with respect to $\braket{\, \cdot \, , \, \cdot \,}.$
\end{definition}
A short exact sequence of vector spaces is naturally associated with any Manin pair:
\be \label{Maninpair}
0 \longrightarrow \frg \xlongrightarrow{i} \frd \xlongrightarrow{i^*} \frg^*\longrightarrow 0
\ee
where $i: \frg \hookrightarrow \frd$ is the inclusion map, $\frg^*$ is the dual vector space of $\frg$, and the map $i^*$ is defined by 
$$ 
\braket{\,i(\sfx)\,,\, \sfw\,}= \braket{\,\sfx\,,\, i^*(\sfw)\,} \ , 
$$ 
for all $\sfx \in \frg$ and $\sfw \in \frd.$
We can always choose an isotropic splitting of the short exact
sequence \eqref{Maninpair}, which is an injective map 
$$
j: \frg^* \longrightarrow \frd  \qquad \mbox{with} \quad i^* \circ j =
\unit_{\frg^*} \ .
$$ 
In this case
\be \label{eq:frdsplit}
\frd= \frm \oplus \frg \ ,
\ee
where $\frm = {\rm Im}(j)$ is a maximally isotropic subspace with respect to the pairing $\braket{\, \cdot \, , \, \cdot \,},$ but not generally a Lie subalgebra of $\frd$. We call $(\frd, \frg; j)$ a \emph{split Manin pair}.

The choice of an isotropic splitting of the short exact sequence \eqref{Maninpair} defines an almost para-Hermitian structure on the Lie algebra $\frd$. It is given by an almost para-complex structure $\kappa \in {\sf Aut}(\frd)$ such that 
\be \label{eq:kappadef}
\kappa\big(j(\sfx)+ i({\tilde\sfx})\big)= j(\sfx)- i({\tilde\sfx})\ ,
\ee
for all $\sfx \in \sfg^*$ and $\tilde{\sfx} \in \sfg,$ and the symmetric non-degenerate pairing $\braket{\, \cdot \, , \, \cdot \,}.$ The almost para-complex structure $\kappa$ is compatible with the pairing $\braket{\, \cdot \, , \, \cdot \,}$ by construction. Then the fundamental 2-form $\cW \in \midwedge^2\, \frd^*$ induced by $\kappa$ and the pairing is 
$$ 
\cW( \sfw, \sfz)=  \braket{\,\kappa(\sfw)\,,\, \sfz\,}\ , 
$$ 
for all $\sfw, \sfz \in \frd,$ which by using isotropy of $\frm$ and $\frg$ with respect to the pairing reads
$$
\cW\big(j(\sfx)+ i(\tilde{\sfx})\,,\, j(\sfy)+ i(\tilde{\sfy})\big)=
\braket{\,j({\sfx})\,,\, i(\tilde\sfy)\,} - \braket{\,i(\tilde\sfx)\,,\,j({\sfy})\,} \ ,
$$
for all ${\sfx}, {\sfy} \in \frm$ and $\tilde\sfx, \tilde\sfy \in
\frg.$ Thus the subspaces $\frm$ and $\frg$ are also maximally
isotropic with respect to $\cW,$ so that $\cW \in \frm^* \wedge \frg^*.$  

There is a notion of $B$-transformations in this case preserving the Lie
subalgebra $\frg$ which are generated by bivectors $\Lambda \in \midwedge^2\, \frg.$ Once a splitting $j$ is fixed, we may then obtain a new subspace $\frm_\Lambda={\rm Im}(j_\Lambda)$, where 
$$
j_\Lambda({\sfx})= j({\sfx})+ i\big(\Lambda(\sfx)\big)\ ,
$$ 
for all ${\sfx} \in \frg^*.$ The subspace $\frm_\Lambda$ is again isotropic, so this gives a transformation that maps an isotropic splitting $j$ into another  isotropic splitting $j_\Lambda.$ The difference between these two splittings is given by the associated almost para-complex structure, which we formally write as 
$$
\kappa_\Lambda= \kappa+ 2\,\Lambda \ . 
$$
Correspondingly, the fundamental 2-form for $j_\Lambda$ reads 
$$
\cW_\Lambda= \cW + 2\, i(\Lambda) \ .
$$
Generally, changes of polarization $\vartheta\in\sfO(d,d)(\frd)$ which map a
split Manin pair $(\frd,\frg;j)$ into another split Manin pair
$(\frd,\frg_\vartheta;j_\vartheta)$ are called \emph{non-abelian
  T-duality transformations}~\cite{Lust:2018jsx}.

Suppose now that $\sfD$ is a Lie group which integrates the Lie
algebra $\frd$, i.e. $\frd={\sf Lie}(\sfD)$. The corresponding tangent
Lie group is the semi-direct product
$$
T\sfD \simeq \sfD \ltimes \frd
$$ 
by the adjoint action of $\sfD$ on $\frd\simeq\R^{2d}$ regarded as an
abelian Lie group. Thus $\sfD$ inherits an almost para-Hermitian structure
$(K^{\textrm{\tiny L}}, \eta^{\textrm{\tiny L}}, \omega^{\textrm{\tiny
    L}} )$ from $(\kappa,\braket{\, \cdot \, , \, \cdot \,}, \cW )$ by
using the isomorphism between $\frd$ and the left-invariant vector fields on $\sfD,$ which by construction is left-invariant with respect to the left action of $\sfD$ on itself. Hence $\sfD$ is a doubled group.
As a vector bundle, the tangent bundle admits a splitting into
left-invariant distributions
$$
T\sfD=L_{\frm}(\sfD) \oplus L_{\frg}(\sfD) \ ,
$$ 
which corresponds fiberwise to the vector space splitting
\eqref{eq:frdsplit}. Here $L_{\frm}(\sfD)$ is the sub-bundle of
$T\sfD$ associated with the subspace $\frm,$ and $L_{\frg}(\sfD)\simeq T\sfG$ with $\sfG$ the
Lie subgroup of $\sfD$ whose Lie algebra is $\frg,$ hence sections of
$L_{\frg}(\sfD)$ are given by left-invariant vector fields on $\sfG$. Clearly $L_{\frm}(\sfD)$ is not generally integrable, whereas $\sfG$ defines a foliation of $\sfD.$ 

It follows that a compatible generalized metric $\cH^{\textrm{\tiny L}}$ can always be defined on such an almost para-Hermitian manifold by considering a left-invariant Riemannian metric $\cG$ on $\sfG$ and setting
$$
g_-(X_-, Y_-)=\cG\big(X^{\textrm{\tiny L}}, Y^{\textrm{\tiny L}}\big)\ , 
$$  
where $X_-,Y_- \in \mathsf{\Gamma}(L_{\frg}(\sfD))$ are the sections
of $L_{\frg}(\sfD)$ corresponding to the left-invariant vector fields
$X^{\textrm{\tiny L}}, Y^{\textrm{\tiny L}} \in
\mathsf{\Gamma}(T\sfG)$ on $\sfG.$ The fiberwise metric $g_-$ on
$L_{\frg}(\sfD)$ induces a fiberwise metric $g_+$ on $L_{\frm}(\sfD)$
in the usual way by 
\be
g_+(X_+,Y_+) = g_-^{-1}\big(\eta^{\textrm{\tiny L}}{}^\flat(X_+) ,
\eta^{\textrm{\tiny L}}{}^\flat(Y_+)\big)
\label{eq:g+cG} \ee
for all $X_+,Y_+\in\mathsf{\Gamma}(L_\frm(\sfD))$. Then the left-invariant compatible generalized metric is
$$
\cH^{\textrm{\tiny L}}= g_+ + g_- \ ,
$$
which is indeed Riemannian. It is straightforward to show that this metric is the unique left-invariant Riemannian metric that can be defined on $\sfD$ for which the basis of left-invariant vector fields is orthonormal.

\begin{remark}
Whenever $\sfG$ is a closed connected Lie subgroup of the doubled
group $\sfD$, the coset space $\cQ=\sfD/\sfG$ is a smooth manifold and
the quotient map $\pi: \sfD \rightarrow \cQ$ is a principal
$\sfG$-bundle. In this case, $L_{\frg}(\sfD)$ is the induced vertical
distribution and $ L_{\frm}(\sfD)$ is the horizontal
distribution. Then an alternative way of defining a compatible
generalized metric is by lifting a Riemannian metric defined on $\cQ$
to $\sfD,$ as discussed in Example \ref{fiberborn}. 
\end{remark}

Para-Hermitian structures on Drinfel'd doubles now arise naturally from the above discussion.

\begin{definition}
Let $(\frd, \frg; j)$ be a split Manin pair. If $\tilde{\frg}={\rm
  Im}(j)$ closes a Lie subalgebra of $\frd,$ then $(\frd, \frg,
\tilde{\frg})$ is a \emph{Manin triple}. A corresponding triple of
integrating Lie groups $(\sfD,\sfG,\tilde{\sfG})$ is a \emph{Drinfel'd
  double}, and is denoted
$$
\sfD = \sfG \Join \tilde \sfG \ .
$$
\end{definition}
For further details on Drinfel'd doubles, see \cite{kossman1997}.
For a Manin triple, in addition to \eqref{Maninpair} there is also the short exact sequence of vector spaces
\be \label{Manintriple}
0 \longrightarrow \frg^* \xlongrightarrow{} \frd \xlongrightarrow{} \frg \longrightarrow 0
\ee
since a Manin triple corresponds to the Lie bialgebras $(\frg,
\tilde{\frg})$ and $(\tilde{\frg}, \frg).$ Then there is a canonical
para-Hermitian structure induced by the vector space splitting 
$$
\frd=\tilde\frg \oplus {\frg}
$$ 
and the non-degenerate symmetric pairing $\braket{\, \cdot \, , \,
  \cdot \,}.$ The subgroup of non-abelian T-duality transformations
$\vartheta\in\sfO(d,d)(\sfD)$ which map a Manin triple
$(\frd,\frg,\tilde\frg)$ into another Manin triple
$(\frd,\frg_\vartheta,\tilde\frg_\vartheta)$ captures some features of
{Poisson-Lie T-duality}~\cite{Lust:2018jsx}.

\begin{example} \label{cotadrinfeld}
Let $\sfG$ be a $d$-dimensional Lie group. Its cotangent bundle
$T^*\sfG\simeq \sfG \ltimes \R^d$ is a Drinfel'd double Lie group
$\sfD$ with $\tilde\sfG=\IR^d$. Denoting the bundle projection by $\pi: T^*\sfG \rightarrow \sfG,$ the canonical short exact sequence of vector bundles 
$$
0 \longrightarrow L_{\tt v}(T^*\sfG) \xlongrightarrow{} T(T^*\sfG)
\xlongrightarrow{} \pi^*(T\sfG) \longrightarrow 0 
$$
corresponds fiberwise to the short exact sequence of vector spaces
\eqref{Manintriple}. A left-invariant isotropic splitting with respect to the
split signature metric $\eta^{\textrm{\tiny L}},$ induced by the Drinfel'd
double structure, of this short exact sequence defines a
left-invariant para-Hermitian structure on $T^*\sfG$; note that the associated fundamental 2-form $\omega^{\textrm{\tiny L}}$ is not necessarily the canonical symplectic 2-form $\omega_0$ on $T^*\sfG$, which is not left-invariant in general. A left-invariant
compatible generalized metric $\cH^{\textrm{\tiny L}}$ on $T^*\sfG$ can be obtained by the horizontal lift of a left-invariant Riemannian metric $\cG$ on $\sfG$ which, in turn, induces a left-invariant Riemannian metric on~$T^*\sfG$. 
\end{example}

\medskip

\subsection{Doubled Group Born Sigma-Models} \label{sec:GroupBorn} ~\\[5pt]
Let $\sfD$ be a doubled group whose Lie algebra $\frd$ has the
structure of a split Manin pair, and let $(K^{\textrm{\tiny L}},
\eta^{\textrm{\tiny L}})$ be the associated almost para-Hermitian
structure on $\sfD.$ As we have seen, there is a natural compatible
generalized metric $\cH^{\textrm{\tiny L}}$ induced by a
left-invariant Riemannian metric on the Lie subgroup $\sfG\subset\sfD$,
as well as the fundamental 2-form $\omega^{\textrm{\tiny L}}$ induced by the
almost para-Hermitian structure. Thus the doubled group $\sfD$
naturally serves as the target space for a Born sigma-model
$\cS(\cH^{\textrm{\tiny L}}, \omega^{\textrm{\tiny L}}).$ Since $\sfD$
is foliated by $\sfG$, we may look for conditions under which the Born
sigma-model $\cS(\cH^{\textrm{\tiny L}}, \omega^{\textrm{\tiny L}})$
admits a gauging and thereby yields a conventional sigma-model
description of the quotient
$\sfD/ \sfG.$ We will study the problem of the existence of a gauged Born
sigma-model with target space $\sfD$ in a split Manin pair
polarization, using the general description of the Lie algebroid
gauging of Section~\ref{BSM}.

For this, we consider the generators of left-invariant vector fields
$\{ Z_I \}=\{Z_i,\tilde Z^i\}$ on $\sfD$ such that $\{ Z_i \}$ spans
the sections of $L_{\frm}(\sfD)$ and $\{ \tilde{Z}^i \}$ spans the
left-invariant vector fields on $\sfG$. This frame closes a Lie
algebra of the form
\begin{align}
[Z_m, Z_n]&={f_{mn}}^k\, Z_k + H_{mnk}\, \tilde{Z}^k \ , \nonumber \\[4pt] [ Z_m,\tilde{Z}^n]&= {f_{km}}^n \, \tilde{Z}^k + Q_m{}^{nk}\, Z_k \ , \label{mplie}\\[4pt] [\tilde{Z}^m, \tilde{Z}^n]&=Q_k{}^{mn}\, \tilde{Z}^k \  , \nonumber
\end{align}
and admits a dual left-invariant coframe $\{
\Theta^I\}=\{\Theta^i,\tilde\Theta_i \}$ such that $\{ \Theta^i \}$
spans the sections of 
$L_{\frm}^*(\sfD)$ and $\{ \tilde{\Theta}_i \}$ spans the
left-invariant 1-forms on $\sfG$. 

A left-invariant Born metric
$\cH^{\textrm{\tiny L}}$ on $\sfD$ is specified by a fiberwise left-invariant
metric $g_+$ on $L_{\frm}(\sfD)$.
To see when the Lie algebroid of left-invariant vector fields on
$\sfG$ generates the generalized isometries of $\cH^{\textrm{\tiny
    L}},$ we check when the transverse invariance condition for
$\cH^{\textrm{\tiny L}}$ is satisfied. Since the left-invariant vector fields
$\tilde Z^k$ generate the right action of $\sfG$ on $\sfD$, the
vanishing requirement
$$
\pounds_{\tilde{Z}^k}\,g_+ = 0
$$
from Section~\ref{BSMGau} implies that the metric $g_+$ is
bi-invariant for the $\sfG$-action.

We also need to check the transverse invariance of the fundamental
2-form $\omega^{\textrm{\tiny L}}$:
$$
(\pounds_{\tilde{Z}^k}\,\omega^{\textrm{\tiny L}})(X_+,Y_+)=
\omega^{\textrm{\tiny L}}\big([\tilde
Z^k,X_+],Y_+\big)+\omega^{\textrm{\tiny L}}\big(X_+,[\tilde
Z^k,Y_+]\big) = 0
$$
for all $X_+,Y_+\in \mathsf{\Gamma}(L_\frm(\sfD))$.
This holds if and only if the Lie bracketing of the subspace $\frm\subset\frd$
is given by
\be \label{eq:red}
[{\frg}, {\frm}]_\frd \subseteq \frm \ .
\ee
When $\sfG$ is connected this 
implies that the splitting \eqref{eq:frdsplit} is also invariant for
the adjoint action of $\sfG$.

Generally, the
quotient $\cQ= \sfD/ \sfG$ is a homogeneous space and the quotient
map $\sfD \rightarrow \cQ$ is a principal $\sfG$-bundle. The condition
\eqref{eq:red} then implies that $\cQ$ is a 
reductive homogeneous space: it means that there is
a natural $\sfG$-action on $\sfD$ given by right multiplication, with
$\frg$ the Lie algebra of the isotropy subgroup and $\frm$ the generators
of infinitesimal translations of $\cQ$~\cite{Helgason1978}.
In this case the gaugeable Born
sigma-models from Section~\ref{BSMGau} are in correspondence with
$\sfG$-invariant connections\footnote{Strictly speaking, for this
  correspondence we should consider a right-invariant para-Hermitian
  structure on $\sfD$, but this would not affect any of the results above.}
on the principal $\sfG$-bundle $\sfD \rightarrow \cQ$ which are
maximally isotropic with respect to $\eta^{\textrm{\tiny L}},$ the
split signature metric induced by the split Manin pair structure of $\frd=
{\sf Lie}(\sfD)$.
In particular, in a split Manin pair polarization one always
obtains a geometric background for the reduced worldsheet sigma-model $S(g,b)$
for $\cQ$, where the Riemannian metric $g$ descends from the $\sfG$-bi-invariant metric
$g_+$ and the Kalb-Ramond field $b$ is given by the transverse component
of the fundamental 2-form $\omega^{\textrm{\tiny L}}$. In this
setting, a non-abelian T-duality transformation between Born
sigma-models, as a change of split Manin pair polarization, is in the
same spirit as the Poisson-Lie T-duality of~\cite{Severa2015}. We shall
describe these sigma-models explicitly below in the special
case of matrix Lie groups.

\begin{example}\label{ex:simplestmetric}
The simplest example of a fiberwise metric $g_-$ on $L_{\frg}(\sfD)$ induced by a left-invariant metric $\cG$ on $\sfG$ can be written as 
$$ 
g_- = \delta^{ij} \, \tilde{\Theta}_i \otimes \tilde{\Theta}_j\ .
$$
Then the fiberwise metric $g_+$ on $L_{\frm}(\sfD)$ given by \eqref{eq:g+cG}
reads
$$
g_+= \delta_{ij} \, \Theta^i \otimes \Theta^j\ ,
$$
with $\cH^{\textrm{\tiny L}}= g_++ g_-.$ 
In this case from \eqref{mplie} and the Maurer-Cartan equations we find
$$
\pounds_{\tilde{Z}^k}\,g_+ = \delta_{ij}\, Q_l{}^{kj} \, \Theta^i \odot
\Theta^l\ ,
$$
which vanishes when the structure constants $Q_l{}^{kj}$ are
completely skew. This implies that the Lie group $\sfG$ is semi-simple
and $g_+$ is the lift of the metric on $\sfG$ given by the Cartan-Killing form
$$
c^{mn} = \tfrac12\, Q_p{}^{mq}\,Q_q{}^{np}
$$ 
on $\frg={\sf Lie}(\sfG).$ 

For the fundamental 2-form
$$
\omega^{\textrm{\tiny L}}= \Theta^i \wedge \tilde{\Theta}_i \ ,
$$
in this case we find
$$
\pounds_{\tilde{Z}^k}\,\omega^{\textrm{\tiny L}}= {f_{ij}}^k\, \Theta^i
\wedge \Theta^j \ .
$$
Hence $\pounds_{\tilde{Z}^k}\,\omega^{\textrm{\tiny L}}$ has only one
component which belongs to $\mathsf{\Gamma}(\midwedge^2 L_{\frm}^*(\sfD)),$ so it has to vanish identically. This implies that
$$
{f_{ij}}^k=0 \ ,
$$ 
or equivalently that the Lie bracketing of the subspace $\frm\subset\frd$
satisfies
\be \label{eq:frmbrackets}
[{\frm},
{\frm}]_\frd \subseteq {\frg}\ ,
\ee
in addition to \eqref{eq:red}. This means that the almost para-complex
structure $\kappa$ defined by \eqref{eq:kappadef} endows the splitting
\eqref{eq:frdsplit} with a $\IZ_2$-grading by assigning degree~$0$ to
elements of $\frg$ and degree~$1$ to elements of $\frm$. The remaining
fluxes in \eqref{mplie} are constrained by the Bianchi identities
$$
H_{m[np}\, Q_{l]}{}^{km} = 0 \ ,
$$
where the brackets denote skew-symmetrization of the enclosed
indices.

The extra condition
\eqref{eq:frmbrackets} implies that the reductive homogeneous space $\cQ=\sfD/\sfG$ is a 
symmetric space: the quotient $\cQ$ is invariant under
inversion about any chosen origin~\cite{Helgason1978}. Symmetric string backgrounds in this case were also
found in~\cite{Hassler:2016srl,Demulder:2018lmj} as particular explicit solutions to
the strong constraint in the target space double field theory.
\end{example}

\medskip

\subsubsection{Matrix Lie Groups} ~\\[5pt]
The Born sigma-model $\cS(\cH^{\textrm{\tiny
    L}},\omega^{\textrm{\tiny L}})$ for a general doubled group $\sfD$
which is a matrix group can be
written using the exponential parameterization from
Section~\ref{sec:matrixgroups} as
$$
\cS[\phi] = \frac14\,\int_\Sigma\,
\bar{\hat\cH}_{MN}(x) \, \bar\Xi^M \wedge \star\, \bar\Xi^N +
\frac14\,\int_\Sigma\,\bar{\hat\omega}_{MN}(x) \,
\bar\Xi^M\wedge\bar\Xi^N \ ,
$$
where the map $\phi$ embeds a closed string worldsheet $\Sigma$ into
the doubled group $\sfD$, while $\hat\cH_{MN}$ and $\hat\omega_{MN}$
are the components of the generalized compatible metric given in
\eqref{eq:hatcH} and of the fundamental 2-form given in
\eqref{eq:hatomega}. The sigma-model $\cS[\phi]$ has a rigid symmetry given by the action of $\sfD$ on itself by left multiplication. Since the $R$-flux $R^{mnp}$ vanishes in a split Manin pair
polarization, $p^m=\varrho^m$ and $\tilde\ell_m$ are the
left-invariant Maurer-Cartan 1-forms on $\sfG$, as discussed in
Section~\ref{sec:matrixgroups}. Since by Example~\ref{ex:simplestmetric} the metric torsion coefficients
$f_{mn}{}^p$ vanish in this case by the generalized isometry constraints, it
follows that
$$
e=\mathds{1} \qquad \mbox{and} \qquad \varrho^m=\de x^m \ .
$$

The Lie algebroid gauging of the Born
sigma-model in this case is achieved by the minimal coupling of the
Maurer-Cartan 1-forms $\tilde\ell_m$ to a $\sfG$-invariant
connection 1-form $\ccC_m$, giving the gauged Born sigma-model actional functional
\begin{align}
\cS[\phi,\ccC] &=
\frac14\,\int_\Sigma\, \big(\delta_{mn} + \bar
b_{mk}\,\delta^{kp}\,\bar b_{np}\big)\,\de\bar x^m\wedge \star\, \de\bar
                x^n + \frac12\,\int_\Sigma\,\bar b_{mn}\,\de\bar
  x^m\wedge\de\bar x^n \nonumber \\
& \quad \, + \frac14\,\int_\Sigma\, \big(\delta^{mn} +
                \bar\beta^{mk}\,\delta_{kp}\,\bar\beta^{np}\big)\,
                \big(\bar{\tilde q}_m+\bar\ccC_m\big)\wedge \star\,
                (\bar{\tilde q}_n+\bar\ccC_n\big) \nonumber \\
& \hspace{6cm} - \frac12\,\int_\Sigma\, \bar\beta^{mn}\,\big(\bar{\tilde
  q}_m+\bar\ccC_m\big)\wedge\big(\bar{\tilde q}_n+\bar\ccC_n\big) \label{eq:Borndoubledgroup} \\
& \quad \, -\frac12\,\int_\Sigma\,\big(\bar b_{mk}\,\delta^{kn} -
  \delta_{mk}\,\bar\beta^{kn}\big)\,\de\bar x^m\wedge \star\,
  \big(\bar{\tilde q}_n+\bar\ccC_n\big) \nonumber \\
& \hspace{6cm} + \frac12\,\int_\Sigma\, \big(\delta_m{}^n+\bar
  b_{mk}\,\bar\beta^{kn}\big)\, \de\bar x^m\wedge\big(\bar{\tilde
  q}_n+\bar\ccC_n\big) \ . \nonumber
\end{align}
Varying \eqref{eq:Borndoubledgroup} with respect to the gauge fields $\ccC_m$ leads to the self-duality constraints
\begin{align*}
& \big(\delta^{mn}+\bar\beta^{mk}\,\delta_{kp}\,\bar\beta^{np}\big) \, \star \, \big(\bar{\tilde q}_n+\bar\ccC_n\big) - 2\,\bar\beta^{mn}\, \big(\bar{\tilde q}_n+\bar\ccC_n\big) \\ & \hspace{4cm} = \big(\bar b_{nk}\,\delta^{km} -
  \delta_{nk}\,\bar\beta^{km}\big) \, \star \, \de\bar x^n - \big(\delta_n{}^m+\bar
  b_{nk}\,\bar\beta^{km}\big)\, \de\bar x^n \ ,
\end{align*}
which are formally solved by the non-local expression
\begin{align*}
\bar{\tilde q}_m+\bar\ccC_m = \bigg(\frac1{(\mathds{1}-\bar\beta\,\star)^2}\bigg)_{mn}\, \Big((\bar b-\bar\beta)^n{}_k\, \de\bar x^k - (\mathds{1}+\bar\beta\,\bar b)^n{}_k\,\star\,\de\bar x^k \Big) \ .
\end{align*}
Substituting this into the gauged action functional
\eqref{eq:Borndoubledgroup} eliminates the dependence on $\bar{\tilde
  q}_m$, giving a reduced sigma-model that depends only on the leaf
space coordinates $x^m$. The complicated non-local dependence on the
bivector $\beta$ owes to the appearence of $Q$-flux in the doubled
group background; nonetheless, the resulting physical background is
geometric. 

We conclude this section by briefly considering some explicit examples which illustrate how this reproduces well-known backgrounds in the doubled formalism.

\begin{example}[The Doubled Torus]
\label{sec:doubledtorus}
The simplest case corresponds to setting the structure constants to zero:
$$
Q_k{}^{ij}=0 \ ,
$$
in addition to $f_{ij}{}^k=0$ in \eqref{mplie},
so that $\sfD$ is a doubled group integrating a Manin pair corresponding to the abelian Lie group
$$
\sfG = \IR^d \ .
$$
In this case one finds~\cite{Hull2009}
$$
\varrho_m=\tfrac12\, H_{mnp}\,x^p\,\de x^n \ , \quad \tilde\ell_m=\de\tilde x_m \ , \quad \beta^{mn}=0 \qquad \mbox{and} \qquad b_{mn} = H_{mnp}\, x^p \ .
$$
The reduction of the Born sigma-model then yields the standard non-linear sigma-model $S(g,b)$ with flat metric
$$
g=\delta_{mn}\, \de x^m\otimes\de x^n
$$
and Kalb-Ramond field $b$, so that the spacetime is locally $\cQ=\IR^d$. The $H$-flux of the $B$-field agrees with the field strength \eqref{eq:cKgroup} of the fundamental 2-form in this case:
$$
\cK^{\textrm{\tiny L}} = \tfrac12\, H_{mnp}\, \de x^m\wedge\de x^n\wedge\de x^p \ .
$$
After taking the quotient by a cocompact discrete subgroup $\sfD(\IZ)$, the spacetime becomes a $d$-dimensional torus $\IT^d=\IR^d/\IZ^d$ with $H$-flux in this split Manin pair polarization. Thus in this case the compact space $M=\sfD/\sfD(\IZ)$ reproduces the standard doubled torus in the geometric $H$-flux polarization~\cite{Hull2005}.
\end{example}

\begin{example}[Doubled WZW Models]
\label{sec:WZW}
Setting 
$$
H_{mnk}= c_{mi}\,c_{nj}\,Q_k{}^{ij}
$$ 
in addition to $f_{ij}{}^k=0$ in \eqref{mplie} recovers the doubled
sigma-model description of the Wess-Zumino-Witten (WZW)
model discussed in~\cite{DallAgata:2008ycz,Avramis:2009xi,Schulz:2011ye,Demulder:2018lmj}. The doubled group is 
$$
\sfD=\sfG\times\sfG \ ,
$$
and if $\sfG$ is compact then $\sfD$ is embedded into the maximal compact subgroup $\sfO(d) \times
\sfO(d)$ of the generalized T-duality group $\sfO(d, d)$, where the
two copies of the semi-simple Lie group $\sfG$ are associated
to the left-moving and right-moving worldsheet sectors. In this case we can set $\beta=0$ using an $\sfO(d)\times\sfO(d)$-transformation, and embedding
$\sfG$ as the diagonal subgroup of $\sfD$, the gauged Born sigma-model
is a gauged WZW model based on the group $\sfD$ with gauge
group $\sfG$. The field strength \eqref{eq:cKgroup} yields the standard $H$-flux
$$
H=-c_{mi}\,c_{nj}\,Q_k{}^{ij}\,
\theta^m\wedge\theta^n\wedge\theta^k
$$
for the reduced WZW model at level $1$ based on
$\sfD/\sfG\simeq\sfG$, where $\theta$ is the left-invariant
Maurer-Cartan 1-form on $\sfG$. The Riemannian submersion from the
Born manifold $\sfD$ to the quotient $\sfD/\sfG$ is given by
$$
{\mathit\Pi}:\sfD\longrightarrow\sfG \ , \quad (g,g')\longmapsto g^{-1}\,g'
\ ,
$$
where $g^{-1}$ and $g'$ become the left-moving and right-moving
closed string fields after imposing the self-duality constraint
resulting from the gauged Born sigma-model.
Other natural choices of left-invariant almost para-complex structures on the doubled group
$\sfD=\sfG\times\sfG$ correspond to subgroups of $\sfD$ in the same conjugacy
class as the diagonal subgroup
$\sfG\subset\sfD$~\cite{Schulz:2011ye}. Discrete quotients of this
doubled group in the case $\sfG={\sf SU}(2)$ are studied
in~\cite{Schulz:2011ye,Hassler:2016srl} in the context of the T-duals
of the 3-sphere $\mathbb{S}^3$, viewed as a circle bundle over the
2-sphere~$\mathbb{S}^2$.
\end{example}

\begin{example}[Drinfel'd Doubles]
\label{sec:Drinfeld}
The case of a Drinfel'd double
$\sfD$ which is a matrix Lie group corresponds to setting 
$$
H_{mnk}=0
$$ 
in \eqref{mplie}. The gaugeable Born
sigma-models in a Manin triple polarization, for which $f_{ij}{}^k=0$, single out the cotangent bundles
$$
\sfD=T^*\sfG=\sfG\ltimes\IR^d
$$ 
of semi-simple Lie groups $\sfG$, which
reduce to a sigma-model description of flat space $T^*\sfG/\sfG\simeq\IR^d$. In this case one has
$$
\varrho_m=0 \ , \quad b_{mn}=0 \qquad \mbox{and} \qquad \beta^{mn}=Q_p{}^{mn}\,x^p \ , 
$$
and the reduced sigma-model action $S'(g',b')$ can be expressed in terms of a
metric $g'$ and Kalb-Ramond field $b'$ defined through
$g'+b'=(\mathds{1}+\beta)^{-1}$, or equivalently
$$
g'=(\mathds{1}-\beta^2)^{-1} \qquad \mbox{and} \qquad
b'=-(\mathds{1}+\beta)^{-1}\, \beta\, (\mathds{1}-\beta)^{-1} \ . 
$$
Alternatively, considering the change of polarization which
interchanges the roles of the Lie bialgebras $(\frg,\IR^d)$ and
$(\IR^d,\frg)$ in the Manin triple, which is the simplest example of a
Poisson-Lie T-duality
transformation, one obtains a sigma-model with
$b_{nm}=\beta^{nm}=0$~\cite{Hull2009} coinciding locally with that of
Example~\ref{sec:doubledtorus} with vanishing Kalb-Ramond field, as
expected from the general considerations of Section~\ref{smCTQ}.
\end{example}

\medskip

\section{Born Sigma-Models for Doubled Nilmanifolds} \label{sec:BornDTT} ~\\[5pt]
A broad class of compactifications of string theory come in the form
of `twisted tori', which are torus bundles that arise as
Scherk-Schwarz reductions with twist in the group of large
diffeomorphisms of the torus fibers. Examples include nilmanifolds,
which are quotients of nilpotent Lie groups by a cocompact discrete
subgroup. A more general class of examples consists of the
solvmanifolds that are discrete quotients of almost abelian solvable
Lie groups, which can be realized as torus fibrations over a
circle. Discrete quotients of the cotangent bundles of the underlying
Lie groups, which are Drinfel'd doubles, yield doubled geometries that
contain the original twisted torus as well as the correspondence space
for its geometric T-dual backgrounds, and are commonly refered to as
`doubled twisted tori'. In this section we will consider the simplest and best
studied example which doubles the compactification on the
three-dimensional Heisenberg nilmanifold. 

\medskip

\subsection{The Doubled Twisted Torus} ~\\[5pt]
We shall first recall the
construction of the doubled twisted torus as a quotient space in the
case of interest here, following~\cite{Hull2009,SzMar}; see e.g.~\cite{Hull:2019iuy} for more general cases. The doubled twisted torus is obtained from the quotient of the Drinfel'd double $\sfD_\sfH=T^*\sfH$ of the three-dimensional Heisenberg group $\sfH$ with respect to a discrete cocompact subgroup $\sfD_\sfH(\IZ).$ The nilpotent Lie algebra of $T^*\sfH=\sfH\ltimes\IR^3$ has non-vanishing brackets
\be
[T_x, T_z]= m\, T_y \ , \quad [T_x, \tilde{T}^y]= m\, \tilde T^z \qquad
\mbox{and} \qquad [T_z, \tilde{T}^y]= -m\, \tilde{T}^x \ ,  \label{lieth}
\ee
where $m \in \mathbb{Z}$. Here the Heisenberg algebra $\mathfrak{h}$ and the abelian Lie algebra $\mathbb{R}^3$ are spanned, respectively, by $\{T_m\}=\{T_x,T_y,T_z\}$ and $\{ \tilde{T}^m \}=\{\tilde T^x,\tilde T^y,\tilde T^z\},$ and together with $\mathfrak{d}_{\mathfrak{h}}=\mathfrak{h}\ltimes \mathbb{R}^3$ they form a Manin triple. Despite the fact that $\sfH$ is not semi-simple, we can still give a matrix representation for the Lie algebra of the Drinfel'd double $T^*\sfH$. This will prove useful later on for explicitly writing down the coordinate identifications defining the global structure of the doubled twisted torus.  

In local coordinates, any element $\gamma\in T^*\sfH$ may be written as
\be \nonumber
\gamma=
\begin{pmatrix}
1 & m\,x & y & 0 & 0 & \tilde{z} \\
0 & 1 & z & 0 & 0 & -\tilde{y} \\
0 & 0 & 1 & 0 & 0 & 0 \\
0 & -m\,\tilde{y} & \tilde{x}-m\,z\,\tilde{y} & 1 & m\,x & y+\frac{1}{2}\,m\,\tilde{y}^2 \\
0 & 0 & 0 & 0 & 1 & z \\
0& 0 & 0 & 0 & 0 & 1
\end{pmatrix}
\ee
where $(x,y,z)$ are coordinates on the Heisenberg group $\sfH$ and $(\tilde x,\tilde y,\tilde z)$ are coordinates on the fiber $\IR^3$. Then the left-invariant 1-forms are given by the Lie algebra components of the corresponding Maurer-Cartan 1-form
$$
\Theta=\gamma^{-1}\, \de\gamma=\Theta^n\,T_n+\tilde\Theta_n\,\tilde T^n
$$ 
as
\begin{align}
\Theta^x &= \de x \ , \quad \Theta^y=\de y- m\,x\, \de z \qquad
           \mbox{and} \qquad \Theta^z=\de z \ , \nonumber \\[4pt]
\tilde{\Theta}_x &= \de \tilde{x}-m\,z \, \de \tilde{y} \ , \quad
                   \tilde{\Theta}_y = \de \tilde{y} \qquad \mbox{and}
                   \qquad \tilde{\Theta}_z= \de \tilde{z}+ m\,x\, \de \tilde{y} \ , \label{thdeltaleft}
\end{align}
with dual left-invariant vector fields
\begin{align}
 Z_x&=\frac\partial{\partial x} \ , \quad Z_y= \frac\partial{\partial y} \qquad \mbox{and} \qquad
      Z_z= \frac\partial{\partial z} + m\,x\,\frac\partial{\partial y} \ , \label{zdist} \\[4pt]
 \tilde{Z}^x&= \frac\partial{\partial\tilde x} \ , \quad \tilde{Z}^y=
              \frac\partial{\partial\tilde y} + m\,z \,
              \frac\partial{\partial\tilde x} -m\,x\,\frac\partial{\partial\tilde z} \qquad
              \mbox{and} \qquad \tilde{Z}^z= \frac\partial{\partial\tilde z} \ . \label{tilddist}
\end{align}
It follows from \eqn{zdist} that $\{ Z_n \}$ spans an involutive distribution $L_+,$ thus it defines a foliation whose leaves are given by the Heisenberg group $\sfH.$ Similarly \eqn{tilddist} tells us that $\{ \tilde{Z}^n \}$ spans an involutive distribution $L_-$ whose foliation has leaves given by  $\mathbb{R}^3$, the fiber of the cotangent bundle 
$$
\pi:T^*\sfH\longrightarrow\sfH \ .
$$ 

Since $T^*\sfH$ is a Drinfel'd double, it is naturally endowed with a left-invariant para-Hermitian structure defined by the para-complex structure 
$$
\underline{K^{\textrm{\tiny L}}}=Z_n\otimes\Theta^n - \tilde Z^n\otimes\tilde\Theta_n 
$$ 
for which $L_+$ is its $+1$-eigenbundle and $L_-$ is its $-1$-eigenbundle. The split signature metric is given by 
$$
\eta^{\textrm{\tiny L}}= \Theta^n \otimes \tilde{\Theta}_n + \tilde{\Theta}_n \otimes \Theta^n \ , 
$$
and the fundamental 2-form is 
\be \label{eq:omegadtt}
\omega^{\textrm{\tiny L}}=\Theta^n \wedge \tilde\Theta_n
\ee
with field strength 
$$
\cK^{\textrm{\tiny L}}=\de \omega^{\textrm{\tiny L}}= -m\, \de x \wedge \de z \wedge \de \tilde{y} \ .
$$ 
Comparing with \eqref{eq:cKgroup} shows that the only non-vanishing flux
in this polarization is the metric flux $f_{xz}{}^y=-f_{zx}{}^y=-m$.

There further exists a unique left-invariant Riemannian metric $\cH^{\textrm{\tiny L}}$ on $T^* \sfH$ induced by the horizontal lift $g_+=\pi^*g$ of the left-invariant Riemannian metric $g$ on the Heisenberg group $\sfH$ given by 
\be \label{eq:gsfH}
g=\delta_{ij}\, \Theta^i \otimes \Theta^j \ ,
\ee
which can be written as
\be \nonumber
g=
\begin{pmatrix}
1 & 0 & 0 \\
0 & 1 & -m\,x \\
0 & -m\,x & 1+ (m \, x)^2
\end{pmatrix}
\ee
in the holonomic coframe $\{ \de x, \de y, \de z\}$.\footnote{Here we slightly abuse notation as before and identify the coordinates on the Heisenberg group with the pulled back coordinates to $T^*\sfH.$ The left-invariant 1-forms on $\sfH$ are identified with the left-invariant 1-forms $\Theta^i$ in \eqref{thdeltaleft}.}
Then the basis $\{Z_n, \tilde{Z}^n \}$ of left-invariant vector fields on $T^*\sfH$ is orthonormal with respect to $\cH^{\textrm{\tiny L}}$. The fiberwise metric on $L_-$ is given by
$$
g_-(X_-,Y_-)=g_+^{-1}\big(\eta^{\textrm{\tiny L}}{}^\flat(X_-), \eta^{\textrm{\tiny L}}{}^\flat(Y_-)\big)
$$ 
for all $X_-,Y_-\in\mathsf{\Gamma}(L_-)$. Then the compatible generalized metric induced by the horizontal lift of $g$ can be written as 
\be
\cH^{\textrm{\tiny L}} = \delta_{mn}\, \Theta^m \otimes \Theta^n + \delta^{mn}\, \tilde{\Theta}_m \otimes \tilde{\Theta}_n \ , \label{gendtt}
\ee
and it is easy to show that, together with $(K^{\textrm{\tiny L}}, \eta^{\textrm{\tiny L}}, \omega^{\textrm{\tiny L}}),$ it defines a left-invariant Born geometry on $T^* \sfH.$

The coordinate identifications defining the global structure of the doubled twisted
torus are obtained via the left action of a discrete cocompact subgroup $\sfD_\sfH(\IZ)$ of $\sfD_\sfH=T^*\sfH.$ Hence the left-invariant para-Hermitian structure of $T^*\sfH$ remains well-defined on the doubled twisted torus 
$$
M_\sfH= T^*\sfH\, \big\slash\, \sfD_\sfH(\IZ) \ .
$$ 
A generic element $\xi \in \sfD_\sfH(\IZ)$ is given by
\be
\xi=
\begin{pmatrix}
1 & m\,\alpha & \beta & 0 & 0 & \tilde{\delta} \\
0 & 1 & \delta & 0 & 0 & -\tilde{\beta} \\
0 & 0 & 1 & 0 & 0 & 0 \\
0 & -m\,\tilde{\beta} & \tilde{\alpha}-m\,\delta\,\tilde{\beta} & 1 & m\,\alpha & \beta+\frac{1}{2}\,m\,\tilde{\beta}^2 \\
0 & 0 & 0 & 0 & 1 & \delta \\
0& 0 & 0 & 0 & 0 & 1
\end{pmatrix}
 \ , \nonumber
\ee
where $\alpha, \beta, \delta, \tilde{\alpha}, \tilde{\beta},
\tilde{\delta}\in \mathbb{Z}$. The group
action on coordinates induced by the equivalence relation $\gamma\sim
\xi\, \gamma$, which defines the quotient $M_\sfH=T^*\sfH / \sfD_\sfH(\IZ)$, leads to the simultaneous identifications
\begin{align}
x & \sim x+ \alpha \ , \quad y \sim y+m\,\alpha\, z + \beta \qquad
    \mbox{and} \qquad z \sim z+ \delta \ , \nonumber \\[4pt]
\tilde{x}&\sim \tilde{x} + m\,\delta\, \tilde{y} + \tilde{\alpha} \ , \quad
           \tilde{y} \sim \tilde{y}+ \tilde{\beta} \qquad \mbox{and}
           \qquad \tilde{z} \sim \tilde{z} -m\,\alpha\, \tilde{y}
           +\tilde{\delta} \ . \label{iden}
\end{align}
This identifies $M_\sfH$ as a $\IT^2\times \IT^2$-bundle over
$\IS^1\times \IS^1$, with base coordinates $(x,\tilde x)$. The left-invariant 1-forms \eqn{thdeltaleft}, together with the left-invariant vector fields \eqn{zdist} and \eqn{tilddist}, are invariant under the
identifications \eqn{iden}, hence they globally descend to the quotient $M_\sfH=T^*\sfH/\sfD_\sfH(\IZ).$ This also means that the left-invariant para-Hermitian structure on $T^*\sfH$ descends to a para-Hermitian structure on $M_\sfH$,  which we denote by $(K, \eta)$. Hence the corresponding eigenbundles $L_+^{\IZ}$ and $L_-^{\IZ}$ of $K$
are both integrable, since their local generators satisfy the Lie
bracket relations \eqn{lieth}; their integral foliations are
characterized, respectively, by the Heisenberg nilmanifold
$\IT_\sfH=\sfH/\sfH(\IZ)$ and the $3$-torus $\IT^3=\mathbb{R}^3/\IZ^3$
as leaves. This is an example of a transversely parallelizable
foliation~\cite{Mrcun2003}, which implies that the leaf holonomy is
trivial for all leaves, and hence $M_\sfH$ admits the structure of a
Riemannian foliation.

Thus the Drinfel'd double structure here, in the polarization given by a Manin triple, induces a para-Hermitian structure $(K, \eta)$ on $M_\sfH$. Furthermore, the left-invariant Riemannian metric on $T^* \sfH$ descends to a Riemannian metric on $M_\sfH$, that we denote by $\cH$, which is still compatible with the para-Hermitian structure induced by the Drinfel'd double. We call $(K, \eta, \cH)$ the \emph{induced} Born geometry on the doubled twisted torus $M_\sfH.$ In the following we demonstrate how to recover the well-known conventional sigma-model descriptions, in the framework of the Lie algebroid gaugings of Born sigma-models, from the various polarizations of the doubled twisted torus, thus reproducing the results of~\cite{Hull2009} from a different geometric perspective.

\medskip

\subsection{Nilmanifold Polarization} \label{nilmanifold} ~\\[5pt]
We will first describe the Born structure defining the polarization whose leaf space is the Heisenberg nilmanifold $\IT_\sfH$. For this, we note that the doubled twisted torus admits the structure of a principal torus bundle, which is inherited from the vector bundle structure of the cotangent bundle $\pi:T^*\sfH\to\sfH$: the typical fiber is $\IT^3$ acting freely on $M_\sfH,$ giving as base space $M_\sfH / \IT^3 \simeq \IT_\sfH$, i.e. there is a principal $\IT^3$-bundle~\cite{Hull2009}
$$
\bar{\pi}: M_\sfH \longrightarrow \IT_\sfH \ .
$$
The almost para-Hermitian structure induced by this $\IT^3$-action is given by an isotropic splitting 
$$
\bar s: \bar{\pi}^*(T\IT_\sfH) \longrightarrow TM_\sfH
$$ 
with respect to the split signature metric $\eta$ and the fundamental
2-form $\omega$ of the short exact sequence of vector bundles
$$
0 \longrightarrow L_{\tt v}(M_\sfH) \longrightarrow TM_\sfH \longrightarrow \bar{\pi}^*(T \IT_\sfH) \longrightarrow 0 \ . $$
A compatible generalized metric $\cH=g_++g_-$ is induced by the horizontal lift of the natural Riemannian metric $g$ on the Heisenberg nilmanifold $\IT_\sfH$ which descends from \eqref{eq:gsfH}.
Having introduced the Born structure on $M_\sfH$ associated with this
nilmanifold polarization, there is a straightforward definition of a
corresponding Born sigma-model, as discussed in Section~\ref{BSM}, by considering a harmonic map $\phi: \Sigma \rightarrow M_\sfH$ from a closed string worldsheet $\Sigma$ to the doubled twisted torus.
This sigma-model is characterized by the pair $(\cH, \omega)$ with $\cH$ given by \eqref{gendtt} and $\omega$ given by~\eqref{eq:omegadtt}.

The principal bundle structure of $M_{\sfH}$ is crucial for describing the generalized isometry for the gauging of the Born sigma-model $\cS(\cH, \omega)$. The principal bundle 
$\bar{\pi}: M_{\sfH}\rightarrow \IT_\sfH$ with fiber $\IT^3$ admits a bundle-like metric given by \eqref{gendtt}.
Therefore we choose the vertical distribution $L_-^{\IZ}=L_{\tt v}(M_\sfH)$ to be the Lie algebroid on $M_\sfH$ of generalized isometries of our sigma-model. With this choice, it is easy to see that the generalized isometry conditions for a Born sigma-model are satisfied:
$$ 
\pounds_{{X_-}}\, g_+=0 \qquad \mbox{and} \qquad
\pounds_{{X_-}}\, \omega =0 \ , 
$$ 
for all ${X_-}\in \mathsf{\Gamma}(L_-^{\IZ}).$ 

The Lie algebroid gauging discussed in Section \ref{BSMGau} can then be applied to this case. 
The Born sigma-model has a rigid symmetry under the left action of the doubled group $\sfD_\sfH$ on the coset $M_\sfH$. For this polarization we write it in the usual way as
$$
\cS[\phi]=\frac{1}{4}\, \int_\Sigma\, \Big(\delta_{ij}\, \bar\Theta^i
\wedge \star\, \bar\Theta^j + \delta^{ij}\, \bar{\tilde{\Theta}}_i
\wedge \star\, \bar{\tilde{\Theta}}_j\Big) +
\frac{1}{2}\,\int_{\Sigma}\, \bar\Theta^i \wedge \bar{\tilde{\Theta}}_i \ ,
$$
and we gauge it along the $\IT^3$ leaves of the foliation, i.e. we
introduce covariant derivatives on $\Sigma$ only for the coordinates $(\tilde{x}, \tilde{y}, \tilde{z})$
adapted to the leaves. The
procedure is formally identical to that described in
Section~\ref{smCTQ}: the self-duality constraint $\delta\cS[\phi, A]/\delta A_i=0$ leads to the reduced sigma-model 
$$
S[\phi]=\frac{1}{2}\,\int_\Sigma\, \delta_{ij}\, \bar\Theta^i \wedge \star\, \bar \Theta^j 
$$
for the nilmanifold $\IT_\sfH,$ with the image of $\phi$ projected to the leaf space.
In other words, the Riemannian submersion described by the gauging is given by the bundle projection 
$$
\bar{\pi}: (M_\sfH, \cH, \omega) \longrightarrow (M_\sfH / \IT^3, g, 0)
$$
where $M_\sfH / \IT^3 \simeq \IT_\sfH$, with background metric
\be 
g = \delta_{ij} \, \Theta^i \otimes \Theta^j= \de x \otimes \de x +
(\de y - m\,x\, \de z)\otimes  (\de y - m\,x\, \de z) + \de z \otimes
\de z \ , \label{metrlef}
\ee
and vanishing Kalb-Ramond field $b=0$; this reproduces the description
of~\cite{Hull2009} which was obtained using a different procedure (see
also~\cite{SzMar}). We have thereby obtained the natural background on
the Heisenberg nilmanifold from the gauging of a generalized isometry
induced by the foliation given by the fibers of the principal
$\IT^3$-bundle $M_{\sfH}.$ In our framework, we deal with globally defined sections of tensor bundles over $M_{\sfH}$ and the reduced metric $g$ on the quotient $\IT_\sfH$ is still a globally defined section.

\medskip

\subsection{Strongly T-Dual Sigma-Model with $H$-Flux}\label{sec:Hflux} ~\\[5pt]
We shall now show how the expected geometric T-dual background with NS--NS
$H$-flux emerges within our framework. Again we will closely follow
\cite{Hull2009, SzMar}. For this, we pull back the Born structure $(K,
\eta, \cH)$, characterizing the nilmanifold polarization, by an
$\sfO(3,3)(M_\sfH)$-transformation $\vartheta$ to get a polarization
which is defined as follows. We consider again a free $\IT^3$-action on $M_\sfH$, which now however gives a principal $\IT^3$-bundle 
$$
\bar{\pi}\rq{}: M_\sfH \longrightarrow \IT^3 \ .
$$ 
The short exact sequence of vector bundles induced by this principal bundle is
\be \label{eq:shortexHflux}
0 \longrightarrow L_{\tt v}'(M_\sfH) \longrightarrow TM_\sfH \longrightarrow \bar{\pi}'{}^*(T \IT^3) \longrightarrow 0
\ee
and we choose an isotropic splitting 
$$
\bar s':\bar{\pi}'{}^*(T \IT^3) \longrightarrow TM_\sfH 
$$ 
with respect to $\eta.$ This defines an almost para-Hermitian structure $(K',\eta)$ on $M_\sfH.$ The compatible generalized metric $\cH\rq{}$ is defined by the horizontal lift of the standard Euclidean metric on the torus $\IT^3.$ 

Locally we may describe this polarization by the Lie algebra 
\be
[Z'_x, Z'_z]= m\, \tilde{Z}^{\prime\,y} \ , \quad [Z'_x, Z'_y]= -m\, \tilde{Z}^{\prime\,z} \qquad \mbox{and} \qquad [Z'_z, Z'_y]= m\, \tilde{Z}^{\prime\,x} \label{t3lie}
\ee
of generators for the left-invariant vector fields on $M_\sfH$, with all other brackets vanishing; here we rearranged the
generators of the Drinfel'd double group $\sfD_\sfH=T^*\sfH$, so that $\{Z'_m, \tilde{Z}^{\prime\,m} \}$ defines a new frame for the sections of $TM_{\sfH}$. This defines a new Born structure $(K', \eta,\cH').$ It may be regarded as a different choice of horizontal sub-bundle of $TM_{\sfH},$ since the vertical distribution spanned by $\{ \tilde{Z}^m \}$ remains unchanged so the fibers are still 3-tori $\IT^3.$ 
The new eigenbundles are thus spanned by globally defined vector fields
\begin{align*}
Z'_x &= \frac{\partial}{\partial x} \ , \quad Z'_y= \frac{\partial}{\partial y} -m\,x\, \frac{\partial}{\partial \tilde{z}} \qquad \mbox{and} \qquad Z'_z=\frac{\partial}{\partial z} - m\,y\, \frac{\partial}{ \partial \tilde{ x}} + m\,x\, \frac{\partial}{\partial \tilde{y}} \ , \\[4pt]
\tilde{Z}^{\prime\,x} &= \frac\partial{\partial\tilde x} \ , \quad
  \tilde{Z}^{\prime\,y} = \frac\partial{\partial\tilde y} \qquad
  \mbox{and} \qquad \tilde{Z}^{\prime\,z}
  =\frac\partial{\partial\tilde z} \ .
\end{align*}
The dual 1-forms are given explicitly by
\begin{align*}
\Theta^{\prime\,x} &= \de x \ , \quad \Theta^{\prime\,y} =\de y \qquad
  \mbox{and} \qquad \Theta^{\prime\,z} =\de z \ , \\[4pt]
\tilde{\Theta}'_x &= \de \tilde{x} +  m\,y\, \de z \ ,
                    \quad
\tilde{\Theta}'_y =\de \tilde{y} + m\,z\, \de x \qquad \mbox{and}
                   \qquad
\tilde{\Theta}'_z =\de \tilde{z} +m\,x\, \de y 
                     \ .
\end{align*}
Analogously to the discussion in \cite{Hull2009}, the action of the Lie algebroid
represented by $\{\tilde{Z}^{\prime\,i} \}$ on $M_\sfH$ generates an action of
$\IT^3$ on $M_\sfH$; this highlights the existence of the structure of
$M_\sfH$ as a principal $\IT^3$-bundle given by
$\bar{\pi}': M_\sfH \rightarrow \IT^3,$ with $M_\sfH / \IT^3 \simeq
\IT^3.$ Such a polarization is given by an isotropic splitting of the
canonical short exact sequence of vector bundles associated with any
fiber bundle, so that 
$$
T M_\sfH=L'_{\tt h}(M_\sfH)\oplus L'_{\tt v}(M_\sfH) \ .
$$

There is a corresponding Born sigma-model $\cS'(\cH', \omega')$ defined
by the new Born structure on $M_\sfH$ obtained from the action of
$\vartheta \in {\sfO}(3,3)(M_\sfH)$ on the nilmanifold polarization, with
$$ 
\cH' = \delta_{ij}\,\Theta^{\prime\,i} \otimes \Theta^{\prime\,j} +
\delta^{ij}\,\tilde{\Theta}'_i \otimes \tilde{\Theta}'_j \qquad
\mbox{and} \qquad \omega'= \Theta^{\prime\,i} \wedge \tilde{\Theta}'_i\ .
$$
The metric $\cH'$ is a bundle-like metric obtained from the horizontal
lift $g'_+=\delta_{ij}\,\Theta^{\prime\,i} \otimes \Theta^{\prime\,j}$
of the standard Euclidean metric on $\IT^3.$ The fundamental 2-form
$\omega'$ has non-trivial monodromies under the identification $x\sim
x+1$, $y\sim y+1$ and $z\sim z+1$, under which it changes by $B_+$-transformations in the
group of large diffeomorphisms of the doubled twisted torus $M_\sfH$. Comparing the corresponding field strength
$$
\cK'=\de\omega'=3\,m\,\de x\wedge\de y\wedge\de z
$$
to \eqref{eq:cKgroup} shows that the only non-vanishing flux in this new
polarization is an $H$-flux $H'_{xyz}=m$.

The Born sigma-model then reads
$$
\cS'[\phi]=\frac{1}{4}\, \int_\Sigma\, \Big(\delta_{ij} \,
\bar\Theta^{\prime\,i} \wedge \star\, \bar\Theta^{\prime\,j} + \delta^{ij}\,
\bar{\tilde{\Theta}}'_i \wedge \star\, \bar{\tilde{\Theta}}'_j \Big) +
\frac{1}{2}\, \int_\Sigma\, \bar\Theta^{\prime\,i} \wedge \bar{\tilde{\Theta}}'_i \ .
$$ 
In this case we consider again the sub-bundle $L_{\tt v}'(M_\sfH)$ as Lie
algebroid for the gauging. It is easy to show that 
$$
\pounds_{{X}_{\tt v}'}\, g'_+=0 \qquad \mbox{and} \qquad
\pounds_{{X}_{\tt v}'}\,\omega'(Y_{\tt h}, Z_{\tt h})=0\ ,
$$
for all ${X}_{\tt v}'\in \mathsf{\Gamma}(L_{\tt v}'(M_\sfH))$ and $Y_{\tt h}, Z_{\tt h} \in \mathsf{\Gamma}(L'_{\tt h}(M_\sfH)).$
Thus the gauging can be implemented by covariantizing the pullback
differentials of the leaf coordinates $(\tilde{x}, \tilde{y},
\tilde{z})$ in the usual way and the reduced sigma-model on 
the leaf space $M_{\sfH} / \IT^3\simeq \IT^3$ is given by $S'(g',b')$
with the background
$$
g'= \de x \otimes \de x + \de y \otimes \de y + \de z \otimes \de z \
,
$$ 
and
$$ 
b'=-m\,(x\, \de y \wedge \de z + y\, \de z \wedge \de x + z\, \de x
\wedge \de y) \ .
$$
The Kalb-Ramond field $b'$ is only locally defined on $\IT^3$, but the
topological term of the sigma-model can be defined by a Wess-Zumino
extension using
the corresponding $H$-flux
$$
H'=\de b' = -3\,m\,\de x\wedge\de y\wedge\de z \ ,
$$ 
which is a globally defined integral 3-form on $\IT^3.$
In other words, we obtain a reduced sigma-model on the leaf space
described by the Riemannian submersion 
$$
\bar{\pi}': (M_{\sfH}, \cH', \omega') \longrightarrow (M_\sfH/\IT^3, g', b')\
,
$$ 
with $M_\sfH/\IT^3\simeq\IT^3$,
which is again the projection map of the principal $\IT^3$-bundle
associated with this polarization. This is an explicit example of how
standard T-dual geometric backgrounds, in this case $(\IT_\sfH, g,
b=0)$ and $(\IT^3, g', b')$, emerge from our formalism in the spirit
of Section~\ref{specT}.

\medskip

\subsection{T-Fold Polarization} ~\\[5pt]
We shall now describe the standard {T-fold} background which is T-dual
to the $\IT_\sfH$ and $\IT^3$ backgrounds within our formalism. So far
we considered the two natural $\IT^3$-actions on the doubled twisted
torus $M_\sfH.$ Let us now discuss what happens when we try to use the
leaf space of the $\IT_\sfH$ foliation for the Lie algebroid action. Here
$M_\sfH$ is again foliated by both $\IT^3$ and $\IT_\sfH,$ but the
distributions identified with their tangent bundles now have opposite
eigenvalues to those of the nilmanifold polarization from
Section~\ref{nilmanifold}. In this case we cannot work with any
fibration, since $M_\sfH$ does not admit the structure of a principal
$\IT_\sfH$-bundle, similarly to the discussion of~\cite{Hull2009}. 
We shall describe how to define the new para-Hermitian structure
starting from the natural left-invariant Riemannian metric $\cH$ in
this polarization. 

For instance, let us consider the global coframe
\begin{align}
\Theta^x & = \de x \ , \quad \Theta^y = \de y-m\,x\, \de \tilde{z} \qquad \mbox{and} \qquad \Theta^z = \de z+m\,x\, \de \tilde{y} \ , \nonumber \\[4pt]
\tilde{\Theta}_{x}&= \de \tilde{x}-m\,\tilde{z}\,\de \tilde{y} \ , \quad \tilde{\Theta}_y = \de \tilde{y} \qquad \mbox{and} \qquad \tilde{\Theta}_z =\de \tilde{z} \ , \nonumber
\end{align}
descending from the left-invariant 1-forms on $T^*\sfH,$ where
$(\tilde{x}, \tilde{y}, \tilde{z})$ are local coordinates adapted to
the $\IT_\sfH$ foliation. Then the Riemannian metric $\cH$ descending
from the left-invariant Riemannian metric on $T^*\sfH$ is given in
this local basis by
$$
\cH= \delta_{ij}\, \Theta^i \otimes \Theta^j +
\delta^{ij}\,\tilde{\Theta}_i \otimes \tilde{\Theta}_j \ .
$$ 
To define the new almost para-Hermitian structure associated with the
T-fold polarization, we consider the natural foliations of the doubled
twisted torus, just as we did in Section~\ref{nilmanifold}. In that
case, we saw that there is a natural almost para-Hermitian structure
arising from the principal $\IT^3$-bundle structure of $M_\sfH,$ since
there we were interested in the quotient with respect to the $\IT^3$ foliation. We now wish to discuss the outcome of the other possible quotient, given by the $\IT_\sfH$ foliation.

Since there is no other natural fiber bundle structure in this case,
we use the left-invariant metric $\cH$ to define the splitting we
need. For this, we consider the short exact sequence of vector bundles
associated with the foliation $\cF_\sfH$ by $\IT_\sfH$: 
$$
0 \longrightarrow T \cF_\sfH \longrightarrow TM_\sfH \longrightarrow N
\IT_\sfH \longrightarrow 0 \ .
$$
We then choose an isotropic splitting 
$$
s: N \IT_\sfH \longrightarrow TM_\sfH
$$ 
whose image ${\rm Im}(s)= T\cF_\sfH^\perp $ is the orthogonal
complement of the tangent bundle $T\cF_\sfH$ with respect to $\cH$; it
is maximally isotropic with respect to the split signature metric $\eta$
inherited from the Drinfel\rq{}d  double structure on $T^*\sfH$.

Let $\{\tilde Z^{\prime\prime\,i}\}=\{ \tilde{Z}^{\prime\prime\,x} , \tilde{Z}^{\prime\prime\,y},
\tilde{Z}^{\prime\prime\,z} \}$ be a basis of left-invariant vector
fields on $\IT_\sfH$ with dual 1-forms $\{\tilde\Theta''_i\}=\{
\tilde{\Theta}''_x, \tilde{\Theta}''_y,
\tilde{\Theta}''_z \}.$ Then ${\rm Im}(s)$ is locally generated
by the vector fields
$$ 
Z''_i = \frac{\partial}{\partial x^i } +
\frac12\,N''_{ij}\, \tilde{Z}^{\prime\prime\,j} \ ,
$$
where 
\be
N\rq{}\rq{}=\frac{2\,m\,x}{1+(m\,x)^2}\,
\begin{pmatrix}
0 & 0 & 0 \\
0 & 0 & -1 \\
0 & 1 & 0
\end{pmatrix}
\ . \nonumber
\ee
This basis is completed by the vector fields $\{
\tilde{Z}^{\prime\prime\,j} \}$ to form a local frame for $\mathsf{\Gamma}(TM_\sfH).$ Then the dual coframe is given by
$$
\Theta^{\prime\prime\,i} = \de x^i \qquad \mbox{and} \qquad
\tilde{\Theta}''_i= \tilde{\Theta}_i - \tfrac12\,N''_{ji}\,\de x^j
$$ 
and the Riemannian metric $\cH''$ in this coframe has the form
\be \nonumber
\cH'' =
\begin{pmatrix}
g''_+  & 0 \\
0 & g''_-
\end{pmatrix}
\ ,
\ee
where the local expressions for $g''_+$ and $g''_-$ are
\begin{equation} \nonumber
g''_+ =\frac{1}{1+(m\,x)^2}\,
\begin{pmatrix}
1+(m\,x)^2 & 0 & 0 \\
0 & 1 & 0 \\
0 & 0 & 1 \\
\end{pmatrix} \qquad \mbox{and} \qquad
g''_-=
\begin{pmatrix}
1 & 0 & 0 \\
0 & 1+(m\,x)^2 & 0 \\
0 & 0 & 1+(m\,x)^2
\end{pmatrix}
\ , 
\nonumber
\end{equation}
with $g''_+$ written in the coframe $\{\Theta^{\prime\prime\,i}
\}=\{\de x,\de y,\de z\}$,
and $g''_-$ written in the coframe $\{ \tilde{\Theta}''_i \}$
which is globally defined on the Heisenberg nilmanifold
$\IT_\sfH$. The fiberwise Riemannian metric $g''_+$ on $T\cF_\sfH^\perp=
{\rm Im}(s)$ is transverse invariant,
i.e. $\pounds_{{X}_-}\,g''_+=0,$ for all ${X}_-\in \mathsf{\Gamma}(T\cF_\sfH)$, and  ${\rm Ker}(g''_+)={\rm Im}(s)$ by definition of the splitting $s.$ Therefore $(M_\sfH,g''_+, \IT_\sfH)$ is a Riemannian foliation.

In this local frame, we can use the functions $N''_{ij}$ to also write down the local
decomposition of the fundamental 2-form
$\omega''$, as
discussed in Section~\ref{BSMGau}. Comparing its field strength
$\cK''=\de\omega''$ with \eqref{eq:cKgroup} identifies both non-vanishing
$Q$-flux and $H$-flux in this polarization:
\be \label{eq:QHTfold}
Q^{\prime\prime\,yz}{}_x=-m \qquad \mbox{and} \qquad H''_{xyz} =
\frac{m\,\big(1-(m\,x)^2\big)}{\big(1+(m\,x)^2\big)^2} \ .
\ee
The local forms of both the compatible generalized metric $\cH''$ and
the fundamental 2-form
$\omega''$ should be understood in a continuation of $x\in[0,1)$ to a
covering space $\IR_x$ of the $x$-circle $ \mathbb{S}_x^1$. Under the
identification $x\sim x+1$ describing the covering map $\IR_x\to \mathbb{S}_x^1$, they change by an
$\sfO(3,3)(M_\sfH)$-transformation in the group of large
diffeomorphisms of the doubled twisted torus $M_\sfH$.

The Lie algebroid to be used for the gauging is $T\cF_\sfH.$ It is
easy to show that the Born sigma-model $\cS''(\cH'', \omega'')$
satisfies the generalized isometry conditions with respect to $\{
\tilde{Z}^{\prime\prime\, m}  \},$ i.e. the Lie derivative of
$\omega\rq{}\rq{}$ also has vanishing orthogonal component, with
respect to the $\IT_\sfH$ foliation $\cF_\sfH$, along any vector field on $\IT_\sfH$. 
This yields the Riemannian submersion 
$$
{\mathit\Pi}'': (M_{\sfH}, \cH'', \omega'') \longrightarrow (M_{\sfH}/
\IT_\sfH, g'', b'')
$$
with 
\begin{equation}
g'' =\frac{1}{1+(m\,x)^2}\,
\begin{pmatrix}
1+(m\,x)^2 & 0 & 0 \\
0 & 1 & 0 \\
0 & 0 & 1 \\
\end{pmatrix} \qquad \mbox{and} \qquad
b'' =\frac{m\,x}{1+(m\,x)^2}\,
\begin{pmatrix}
0 & 0 & 0 \\
0 & 0 & -1 \\
0 & 1 & 0
\end{pmatrix}
\nonumber
\end{equation}
in the holonomic basis for 1-forms $\{ \de x, \de y, \de z\}.$ In this
case we recover the standard non-geometric background $(g'', b'')$
T-dual to the previous $\IT_\sfH$ and $\IT^3$ backgrounds, with
$H$-flux $H''=\de b''$ as in \eqref{eq:QHTfold} (up to a sign). Here the quotient space
$M_{\sfH}/ \IT_\sfH$ is an example of a leaf space which does
not have the structure of a smooth manifold. In particular, since
$(M_\sfH, g''_+, \IT_\sfH)$ is a Riemannian foliation with compact
leaves, the leaf space $M_\sfH / \IT_\sfH$ is an orbifold
\cite{Boyer2007, Molino, Mrcun2003}. Thus the quotient map ${\mathit\Pi}\rq{}\rq{}: M_\sfH \rightarrow M_{\sfH}/ \IT_\sfH$ is an orbifold submersion.

\medskip

\subsection{Essentially Doubled Polarization} ~\\[5pt]
Similarly to the T-fold polarization, in the essentially doubled
polarization we consider the principal $\IT^3$-bundle $\bar{\pi}\rq{}:
M_\sfH \rightarrow \IT^3$ and try to describe the (now disallowed) gauging with respect to the distribution ${\rm Im}(\bar s')$ given by the isotropic splitting of the short exact sequence of vector bundles \eqref{eq:shortexHflux} induced by the principal fibration with respect to the split signature metric $\eta.$ In other words, in this polarization there is only one integrable distribution on $M_{\sfH}$ whose induced foliation has leaves $\IT^3,$ while ${\rm Im}(\bar s')$ is a non-involutive sub-bundle of $TM_\sfH$. Then the allowed gauging along the foliation would simply give the naive T-dual of the polarization with NS--NS $H$-flux from Section~\ref{sec:Hflux}, obtained locally by interchanging the coordinates $(x,y,z)$ and $(\tilde x,\tilde y,\tilde z)$; indeed, it is only from this perspective that the essentially doubled background satisfies the strong constraint of double field theory, as discussed by~\cite{Jonke2018}.

We may follow the same steps of Section \ref{Rfluxcot} and discuss why
it is not possible to recover a conventional reduced sigma-model even
locally. Again this relies on the fact that the vector sub-bundle
$L_-={\rm Im}(\bar s')$ is non-integrable in this polarization, so the
only foliation present is given by $L_+= L'_{\tt v}(M_\sfH),$ which gives the
naive T-dual of the sigma-model for the spacetime $\IT^3$ with $H$-flux.
This polarization is characterized by globally defined coframes 
\begin{align*}
\Theta^{R \,x} =\de x - m\, \tilde{z}\, \de \tilde{y} \ , \quad \Theta^{R\,y} =\de y  - m \, \tilde{x} \, \de \tilde{z}\qquad
  \mbox{and} \qquad \Theta^{R\,z} =\de z + m \, \tilde{x}\, \de \tilde{y}
\end{align*}
for $L_+^*,$ and
\be \nonumber
\tilde{\Theta}^R_x = \de \tilde{x} \ ,
                    \quad
\tilde{\Theta}^R_y =\de \tilde{y}  
                   \qquad \mbox{and} \qquad
\tilde{\Theta}^R_z =\de \tilde{z} 
\ee
for $L_-^*.$
The compatible generalized metric is given
as usual by 
$$
\cH^R= \delta_{ij}\, \Theta^{R\, i} \otimes \Theta^{R\, j}+
\delta^{ij}\,\tilde{\Theta}^R_i \otimes \tilde{\Theta}^R_j
$$
and the fundamental 2-form is
$$
\omega^R= \Theta^{R\, i}\wedge \tilde{\Theta}^R_i \ , 
$$
with field strength
$$
\cK^R=\de\omega^R = 3\,m\,\de\tilde x\wedge\de\tilde y\wedge\de\tilde
z \ .
$$
Comparing with \eqref{eq:cKgroup} thus shows that this polarization is
characterized by a single non-vanishing $R$-flux
\be \label{eq:Rflux}
R^{xyz} = m \ .
\ee

The corresponding Born sigma-model then reads
$$
\cS^R[\phi]= \frac{1}{4}\, \int_\Sigma\, \Big( \delta_{ij}\,
\bar\Theta^{R\, i} \wedge \star\, \bar\Theta^{R\, j} + \delta^{ij}\,
\bar{\tilde{\Theta}}^R_i \wedge \star\, \bar{\tilde{\Theta}}^R_j\Big)
+ \frac12\,\int_\Sigma\, \bar\Theta^{R\, i} \wedge
\bar{\tilde{\Theta}}^R_i \ .
$$
As discussed in Section \ref{Rfluxcot},  there is no foliation in this
case whose adapted local coordinates are $(\tilde{x}, \tilde{y},
\tilde{z}),$ and therefore we can only force the gauging of these
naive dual coordinates (with no geometric interpretation) by introducing ``covariantized'' maps
$$
\De^\cA \bar{\tilde{x}}_i= \de \bar{\tilde{x}}_i - \bar\cA_i \ ,
$$
where $\cA_i$ are the components of the tensor induced by the vector
bundle morphism $\bar{\cA}: T\Sigma \rightarrow \phi^*L_-$ covering
the identity. The ``gauged'' Born sigma-model is written in the usual way
and eliminating the auxiliary fields $\cA_i$ through their equations
of motion $\delta\cS^R[\phi, \cA]/\delta\cA_i=0$ gives the
self-duality constraints
$$
\delta_{ij}\, R^{ikl}\, \bar{\tilde{x}}_k\, \star\, \de\bar x^j +
\delta_{ij}\, R^{ikl}\, \bar{\tilde{x}}_k\, R^{jmn}\,
\bar{\tilde{x}}_m\, \star\, \De^\cA \bar{\tilde{x}}_n +
\de\bar x^l + R^{ikl}\, \bar{\tilde{x}}_k\, \De^\cA
\bar{\tilde{x}}_i=0\ ,
$$
where the components of the antisymmetric $R$-flux structure constants
are given by \eqref{eq:Rflux} and in this equation there is no sum
over the indices $k,m.$ Similarly to~\cite{Hull2009}, solving this constraint for
$\De^\cA\bar{\tilde x}_i$ eliminates all dependence on $\de\bar{\tilde x}_i$, but this does not give a
local reduced sigma-model which is independent of the naive dual
coordinates $\tilde x_i$. Thus the only possible gauging leads to the
naive T-dual of the sigma-model for the 3-torus $\IT^3$ with $H$-flux, i.e. the reduced sigma-model $S^R(g^R,b^R)$ obtained from writing the reduced
sigma-model $S'(g',b')$ of Section~\ref{sec:Hflux} locally in the
coordinates $(\tilde{x}, \tilde{y}, \tilde{z})$ instead of $(x,y,z)$,
so that the background is now given by
$$
g^R= \de\tilde x \otimes \de\tilde x + \de\tilde y \otimes \de\tilde y + \de\tilde z \otimes \de\tilde z \
,
$$ 
and
$$ 
b^R=-m\,(\tilde x\, \de\tilde y \wedge \de\tilde z + \tilde y\,
\de\tilde z \wedge \de\tilde x + \tilde z\, \de\tilde x
\wedge \de\tilde y) \ .
$$

\bigskip

\bibliographystyle{ieeetr}
%\bibliography{bibprova1}

\end{document}